\def\VersionLong{}
\def\VersionFinal{}
\ifdefined\VersionLong%
	\newcommand{\LongVersion}[1]{#1}
	\newcommand{\ShortVersion}[1]{}
\else
	\newcommand{\LongVersion}[1]{}
	\newcommand{\ShortVersion}[1]{#1}
\fi

\documentclass[runningheads,a4paper,10pt,envcountsame]{llncs}

\usepackage{comment}
\ifdefined\VersionLong%
        \includecomment{LongVersionBlock}
        \excludecomment{ShortVersionBlock}
\else
        \excludecomment{LongVersionBlock}
        \includecomment{ShortVersionBlock}
\fi
\includecomment{lncs}

\usepackage[ruled,lined,linesnumbered,noend]{algorithm2e}
\DontPrintSemicolon{}
\SetKwInOut{Input}{Input}
\SetKwInOut{Output}{Output}
\SetKw{KwRemove}{remove}
\SetKw{KwPush}{add}
\SetKw{KwPop}{remove}
\SetKw{KwFrom}{from}
\SetKw{KwTo}{to}
\SetKw{KwCompute}{compute}
\SetKw{KwReturn}{return}
\SetKw{KwBreak}{break}
\SetKw{KwPick}{pick}
\SetKw{KwLet}{let}
\SetKw{KwAnd}{and}
\SetKw{KwBe}{be}
\SetKw{KwSplit}{split}
\SetKw{KwInto}{into}
\SetKw{KwFind}{find}
\SetKw{KwFeed}{feed}
\SetKwProg{Fn}{Function}{:}{}
\SetAlgorithmName{Algorithm}{algorithm}{List of Algorithms}

\makeatletter
\renewcommand{\nllabel}[1]
 {{\let\@currentlabel\algocf@currentlabel
  \let\@currentcounter\algocf@currentcounter
  \label{#1}}}%

\renewcommand{\algocf@nl@sethref}[1]{%
  \renewcommand{\theHAlgoLine}{\thealgocfproc.#1}%
  \hyper@refstepcounter{AlgoLine}%
  \gdef\algocf@currentlabel{#1}%
  \gdef\algocf@currentcounter{AlgoLine}%
 }%
\makeatother

\usepackage[svgnames,table]{xcolor}

\usepackage{graphicx}
\usepackage{stmaryrd}
\usepackage{amssymb, amsmath}
\usepackage[
 	colorlinks=true,
 	citecolor=darkgreen,
	linkcolor=darkblue,
	urlcolor=darkpurple,
]{hyperref}
\usepackage[capitalise,english,nameinlink]{cleveref} %
\usepackage{mathtools}

\usepackage[pdf,singlefile]{graphviz}
\usepackage{amsfonts}
\usepackage{tikz}
\usepackage{booktabs}
\usetikzlibrary{graphs}
\usetikzlibrary{automata,positioning,arrows.meta,decorations.pathmorphing,decorations.pathreplacing,decorations.shapes}
\usepackage{multirow}
\usepackage{array}

\usepackage{colortbl}
\definecolor{darkblue}{rgb}{0.0,0.0,0.6}
\definecolor{darkgreen}{rgb}{0, 0.5, 0}
\definecolor{darkpurple}{rgb}{0.7, 0, 0.7}
\definecolor{darkblue}{rgb}{0, 0, 0.7}

\crefalias{AlgoLine}{algoline}
\crefname{algoline}{\text{line}}{\text{lines}} %
\Crefname{algoline}{\text{line}}{\text{lines}} %
\crefname{line}{\text{line}}{\text{lines}} %
\crefname{item}{\text{item}}{\text{items}} %
\crefname{example}{\text{Example}}{\text{Examples}} %
\crefname{assumption}{\text{Assumption}}{\text{Assumptions}} %
\crefname{algorithm}{\text{Algorithm}}{\text{Algorithms}}
\crefname{algocf}{\text{Algorithm}}{\text{Algorithms}}
\Crefname{algocf}{\text{Algorithm}}{\text{Algorithms}}

\ifdefined\VersionWithComments%
	\usepackage{draftwatermark}
	\SetWatermarkText{draft}
 	\SetWatermarkScale{15}
	\SetWatermarkColor[gray]{0.9}
\fi

\ifdefined\VersionWithComments%
	\usepackage[colorinlistoftodos,textsize=footnotesize]{todonotes}
\else
	\usepackage[disable]{todonotes}
\fi

\newcommand{\gennote}[4][]{\todo[linecolor=#3,backgroundcolor=#3!25,bordercolor=#3#1]{#4: #2}}

\newcommand{\instructions}[1]{{\gennote[,inline]{\bfseries #1}{red}{Instructions}}}

\ifdefined\VersionFinal%
\else
	\usepackage[pagewise]{lineno} %
	\linenumbers%
	
\fi

\tikzstyle{defproblem} = [
 draw=black,
 fill=cyan!10,
 text=black,
 line width=0.5pt,
  text width = \linewidth - 1.6 ex - 1pt,
  inner sep = 0.8 ex,
  rounded corners=4pt]
\newcommand{\recallResult}[2]
{%
	\smallskip

	\noindent\fcolorbox{black}{green!15}{
		\begin{minipage}{.95\columnwidth}
			\noindent\textbf{\cref{#1} (recalled).}
			{\em{}#2}
		\end{minipage}
	}

	\smallskip
}

\tikzstyle{rqanswer} = [
 draw=black,
 fill=gray!30,
 text=black,
 line width=0.5pt,
  text width = \linewidth - 1.6 ex - 1pt,
  inner sep = 0.8 ex,
  rounded corners=4pt]

\usepackage{subcaption}

\usepackage{paralist} %
\newenvironment{ienumeration}
	{\begin{inparaenum}[\itshape i\upshape)]}
	{\end{inparaenum}}
 \newenvironment{myitemize}
	{\ifdefined\VersionLong\begin{itemize}\else\begin{inparaitem}[]\fi}
	{\ifdefined\VersionLong\end{itemize}\else\end{inparaitem}\fi}

\usepackage{booktabs}

\newcommand{\setR}{\mathbb{R}}
\newcommand{\setN}{\mathbb{N}}

\newcommand{\R}{\setR}

\newcommand{\Intervals}{\mathcal{I}}
\newcommand{\setRnn}{\R_{\ge 0}}
\newcommand{\Rnn}{\setRnn}

\newcommand{\setdiff}{\triangle}

\newcommand{\KTrue}{\ensuremath{\mathit{true}}}

\newcommand{\Boolean}[1]{\mathbb{B}_+(#1)}

\newcommand{\word}[1][]{w#1}
\newcommand{\suffixes}{\mathit{suff}}
\newcommand{\Alphabet}{\Sigma}
\newcommand{\TimedWords}{\mathcal{T}(\Alphabet)}
\newcommand{\emptyword}{\epsilon}

\newcommand{\loc}{\ell}
\newcommand{\Loc}{L}
\newcommand{\initLoc}{\loc_0}
\newcommand{\InitLoc}{\Loc_0}
\newcommand{\Final}{F}
\newcommand{\action}{a}

\newcommand{\delay}{d}

\newcommand{\letters}[1]{\Phi(#1)}
\newcommand{\partition}{\Pi}
\newcommand{\A}{\mathcal{A}}
\newcommand{\TA}{\A}

\newcommand{\Edge}{\Delta}
\newcommand{\transition}{\delta}
\newcommand{\formula}{\varphi}

\newcommand{\Transitions}{\Edge}
\newcommand{\Lg}{\mathcal{L}}
\newcommand{\lang}{\mathfrak{L}}

\newcommand{\eqQKey}[1][\Lg]{\mathtt{eq}_{#1}}
\newcommand{\eqQ}[2][\Lg]{\eqQKey[#1](#2)}
\newcommand{\nerode}[1]{\equiv_{#1}}
\newcommand{\prefix}{s}
\newcommand{\PrefixSet}{S}
\newcommand{\Basis}{P}
\newcommand{\basis}{p}
\newcommand{\RowsP}{R_{P}}

\newcommand{\suffix}{e}
\newcommand{\SuffixSet}{E}
\newcommand{\targetLg}{\lang_{\mathrm{tgt}}}

\newcommand{\hypothesisA}{\A_{\mathrm{hyp}}}
\newcommand{\evidenceA}{\hypothesisA^{e}}
\newcommand{\Table}{T}
\newcommand{\TableCell}[3][]{\Table#1(#2, #3)}
\newcommand{\TableRow}[1]{\Table[#1]}

\newcommand{\cex}{\mathit{cex}}
\newcommand{\Lstar}{\ensuremath{\mathrm{L}^*}}
\newcommand{\NLstar}{\ensuremath{\mathrm{NL}^*}}
\newcommand{\ALstar}{\ensuremath{\mathrm{AL}^*}}
\newcommand{\ALdstar}{\ensuremath{\mathrm{AL}^{**}}}

\newcommand{\NLstarRTA}{\ensuremath{\mathrm{NL}^*_{\mathrm{RTA}}}}
\newcommand{\ALstarRTA}{\ensuremath{\mathrm{AL}^*_{\mathrm{RTA}}}}
\newcommand{\PhiTable}{\letters{\PrefixSet \cup \SuffixSet}}
\newcommand{\normalizeLetter}{g_{\alpha}}
\newcommand{\normalize}{\ensuremath{G}_{\alpha}}
\newcommand{\normalized}[1]{#1_{\alpha}}

\newcommand{\tbcolor}{\cellcolor{green!25}\bf}

\newcommand{\ourTool}{\textsc{LearnARTA}}

\ifdefined\VersionWithComments%
 	\definecolor{colorok}{RGB}{80,80,150}
\else
	\definecolor{colorok}{RGB}{0,0,0}
\fi

\newcommand{\eg}{\textcolor{colorok}{e.\,g.,}\xspace}
\newcommand{\ie}{\textcolor{colorok}{i.\,e.,}\xspace}

\makeatletter
\def\orcidID#1{\smash{\href{https://orcid.org/#1}{\protect\raisebox{-1.25pt}{\protect\includegraphics{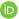}}}}}
\makeatother

\title{Learning Alternating Real-Time Automata}
\author{%
Kazuki Kinoshita\inst{1}
\and
Masaki Waga\inst{1,2}\orcidID{0000-0001-9360-7490}\LongVersion{%
\thanks{This is the author (and extended) version of the manuscript of the same name published in the proceedings of the 3rd International Joint Conference on Quantitative Evaluation of Systems and Formal Modeling and Analysis of Timed Systems (QEST+FORMATS 2026).
    The final version is available at \url{www.springer.com}.
}}
}
\institute{
\ifdefined\VersionAnonymous
\else
  Kyoto University, Kyoto, Japan
\and
  National Institute of Informatics, Tokyo, Japan
\fi
}

\ifdefined\VersionLong
\else
\RequirePackage{graphicx}
\usepackage[firstpageonly=true,angle=0,vpos=.147\paperheight,hpos=.39\linewidth,vanchor=t,hanchor=l]{draftwatermark}
\SetWatermarkText{%
\raisebox{-2cm}%
{\begin{minipage}{\linewidth}\centering\includegraphics[width=12.5mm] 
{QEST+FORMATS_2026_available.pdf}\hfill\includegraphics[width=12.5mm] 
{QEST+FORMATS_2026_functional.pdf}\end{minipage}}%
}
\fi

\begin{document}

\maketitle              %
\ifdefined\VersionFinal
\else
\pagestyle{plain}
\fi
\begin{abstract}
  We present \ALstarRTA{}, the first active learning algorithm for alternating real-time automata (ARTAs) via membership and equivalence queries. The algorithm combines ideas from \ALstar{} for alternating finite automata and \NLstarRTA{} for nondeterministic real-time automata.
  We first define ARTAs and show that alternation improves succinctness, although it does not increase expressive power.
  We then present \ALstarRTA{} and analyze its query complexity.
  Our empirical evaluation suggests that \ALstarRTA{} generally learns smaller automata than \NLstarRTA{} at the cost of more queries.
\end{abstract}
\section{Introduction}\label{section:introduction}
\emph{Active automata learning}, pioneered by Angluin's \Lstar{} algorithm~\cite{DBLP:journals/iandc/Angluin87},
aims to identify an unknown language exactly through membership and equivalence queries.
In a membership query, the learner asks whether a word belongs to the target language $\targetLg$.
In an equivalence query, the learner asks whether a hypothesis automaton $\hypothesisA$ recognizes the target language $\targetLg$; if $\targetLg \neq \Lg(\hypothesisA)$, the teacher returns a concrete witness $\word \in \targetLg \setdiff \Lg(\hypothesisA)$ of deviation, where $\targetLg \setdiff \Lg(\hypothesisA)$ is the symmetric difference of $\targetLg$ and $\Lg(\hypothesisA)$.
For any regular language, \Lstar{} identifies the minimum deterministic finite automaton (DFA) recognizing it with a polynomial number of queries.

\LongVersion{The }\Lstar{}\LongVersion{ algorithm} has been extended in various directions.
One direction is to learn more succinct automata by allowing nondeterministic or alternating branching.
\NLstar{}~\cite{DBLP:conf/ijcai/BolligHKL09} learns a class of \emph{nondeterministic} finite automata (NFAs) called \emph{residual} NFAs.
\ALstar{}~\cite{DBLP:conf/ijcai/AngluinEF15} and \ALdstar{}~\cite{DBLP:journals/iandc/BerndtLLR22} learn \emph{alternating} finite automata (AFAs).
Due to the succinctness of NFAs and AFAs with respect to equivalent DFAs, these algorithms can identify a language in a more succinct form both in theory and in practice.

Another direction is to extend the class of target languages, \eg{} to \emph{timed} languages, where each letter is equipped with the delay since the previous letter.
Many \Lstar{}-style learning algorithms have been proposed for various subclasses of \emph{timed automata (TAs)}~\cite{DBLP:journals/tcs/AlurD94}, such as
\emph{deterministic} TAs~\cite{DBLP:conf/cav/Waga23,DBLP:conf/hybrid/TengZ024,DBLP:journals/chinaf/TengCMZAZ25},
\emph{event-recording automata}~\cite{DBLP:journals/tcs/GrinchteinJL10}, and
\emph{real-time automata (RTAs)}~\cite{DBLP:journals/chinaf/AnWZZZ21,DBLP:journals/tecs/AnZZZ21}.
Despite recent advances in active learning of timed languages, most methods are limited to \emph{deterministic} branching, even when nondeterminism does not add expressive power.
One notable exception is \NLstarRTA{}~\cite{DBLP:journals/tecs/AnZZZ21}, which learns \emph{nondeterministic} RTAs, but no algorithms have been proposed for any subclasses of TAs with \emph{alternating} branching.

\paragraph{Contributions.}
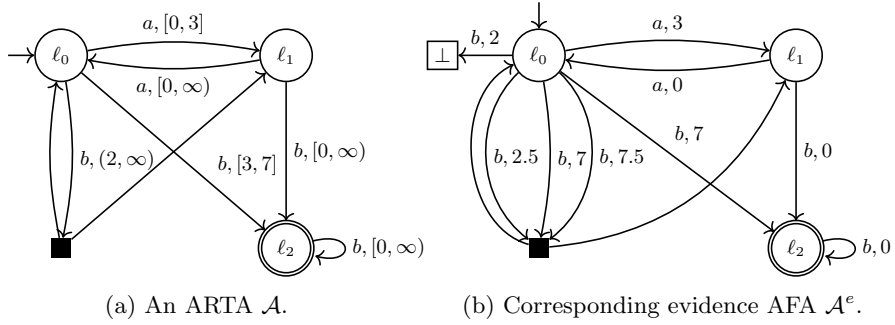
\begin{figure}[tbp]
  \begin{subfigure}{.44\linewidth}
    \centering
    \begin{tikzpicture}[auto, semithick,node distance=1.5cm,scale=0.85,every node/.style={initial text={},transform shape}]
      \node[state, initial] (q0) at (0, 0) {$\loc_{0}$};
      \node[state] (q1) at (3.5, 0) {$\loc_{1}$};
      \node[state,accepting] (q2) at (3.5, -3) {$\loc_{2}$};
      \node[rectangle,fill=black,minimum width=.3cm,minimum height=.3cm] (q3) at (0, -3) {};

      \path[->]
        (q0) edge[bend left=10] node[above] {$a, [0,3]$} (q1)
        (q0) edge[bend left=10] node[pos=0.5,right] {$b, (2, \infty)$} (q3)
        (q0) edge node[pos=0.7,above right] {$b, [3,7]$} (q2)
        (q1) edge[bend left=10] node[below] {$a, [0,\infty)$} (q0)
        (q1) edge node[right] {$b, [0, \infty)$} (q2)
        (q2) edge[loop right] node {$b, [0, \infty)$} (q2)
        (q3) edge[bend left=10] node[left] {} (q0)
        (q3) edge node[below] {} (q1)
        ;
    \end{tikzpicture}
    \caption{An ARTA $\A$.}%
    \label{figure:running_arta}
  \end{subfigure}
  \hfill
  \begin{subfigure}{.55\linewidth}
    \centering
    \begin{tikzpicture}[auto, semithick,node distance=1.5cm,scale=0.85,every node/.style={initial text={},transform shape}]
      \node[state, initial above] (q0) at (0, 0) {$\loc_{0}$};
      \node[state] (q1) at (4, 0) {$\loc_{1}$};
      \node[state,accepting] (q2) at (4, -3) {$\loc_{2}$};
      \node[rectangle,fill=black,minimum width=.3cm,minimum height=.3cm] (q3) at (0, -3) {};
      \node[rectangle,draw=black,minimum width=.3cm,minimum height=.3cm] (bot) at (-1.5, 0) {$\bot$};

      \path[->]
        (q0) edge[bend left=10] node[above] {$a, 3$} (q1)
        (q0) edge[bend right=50] node[pos=0.5,right] {$b, 2.5$} (q3)
        (q0) edge[bend left=10] node[pos=0.5,right] {$b, 7$} (q3)
        (q0) edge[bend left=50] node[pos=0.5,right] {$b, 7.5$} (q3)
        (q0) edge node[pos=0.5,above right] {$b, 7$} (q2)
        (q0) edge node[above] {$b, 2$} (bot)
        (q1) edge[bend left=10] node[below] {$a, 0$} (q0)
        (q1) edge node[right] {$b, 0$} (q2)
        (q2) edge[loop right] node {$b, 0$} (q2)
        (q3) edge[bend left=70] node[left] {} (q0)
        (q3) edge[bend right=30] node[below] {} (q1)
        ;
    \end{tikzpicture}
    \caption{Corresponding evidence AFA $\A^e$.}%
    \label{figure:running_evidence_AFA}
  \end{subfigure}
  \caption{An ARTA and its corresponding evidence AFA. The initial location formula is $\loc_0$. The black square represents a universal branching. Some transitions are simplified for illustration, \eg{} $\transition\bigl(\loc_0, (b, 7)\bigr) = (\loc_0 \land \loc_1) \lor \loc_2$ in $\A^e$ is split into two edges; $(\loc_0, b, [3, 7], (\loc_0 \land \loc_1) \lor \loc_2) \in \Transitions$ in $\A$ is similarly split into two edges and merged with others.}%
  \label{figure:running}
\end{figure}
We propose an active learning algorithm \ALstarRTA{} that learns \emph{alternating RTAs (ARTAs)}.
\cref{figure:running_arta} depicts an ARTA over $\Alphabet = \{a, b\}$.
Each transition of an ARTA is labeled with a letter and an interval representing the elapsed time, like RTAs~\cite{DBLP:journals/jalc/Dima01}.
The target of each transition of an ARTA is a positive Boolean expression over the locations, like alternating finite automata (AFAs)~\cite{DBLP:journals/jacm/ChandraKS81}.
For example, in the ARTA in \cref{figure:running_arta}, if we read a letter $b$ with a delay $\delay \in (2, 3)$ at $\loc_0$, we may simultaneously jump to both $\loc_0$ and $\loc_1$ and this execution is accepted when both spawned branches are accepted at the same time.

We first show that alternating branching improves the succinctness of RTAs, although it does not increase the expressive power.
Namely, ARTAs are equi-expressive with RTAs but are \emph{doubly} exponentially more succinct than deterministic RTAs, whereas nondeterministic RTAs are only \emph{singly} exponentially more succinct than deterministic RTAs~\cite{DBLP:journals/tecs/AnZZZ21}.
This result generalizes the classical succinctness results for AFAs~\cite{DBLP:journals/tcs/Leiss81,DBLP:journals/tcs/Leiss85a} to real-time languages, the class of timed languages recognizable by RTAs.

We then present \ALstarRTA{}, which is based on \ALstar{} and \NLstarRTA{}.
At a high level, \ALstarRTA{} follows \NLstarRTA{}: it initially constructs an evidence AFA (\cref{figure:running_evidence_AFA}) following the idea of \ALstar{} and then generalizes concrete delays on transitions to obtain an ARTA (\cref{figure:running_arta}).
Because delays are drawn from an infinite domain, evidence AFAs necessarily have incomplete transitions.
We adapt the conditions in \ALstar{} to accommodate incomplete transitions.
We show that \ALstarRTA{} terminates for any real-time language with an exponential query bound in the size of the minimal deterministic RTA with a polynomial dependence on the alphabet size, counterexample length, and timing parameters.

We implemented a prototype library \ourTool{} for \ALstarRTA{} and empirically compared \ALstarRTA{} against \NLstarRTA{}.
The results suggest that \ALstarRTA{} generally learns smaller automata than \NLstarRTA{}, suggesting that the aforementioned succinctness result is also helpful in practice.
However, this gain comes at a higher learning cost, particularly in terms of the number of queries and the learning time, as reported in~\cite{DBLP:conf/ijcai/AngluinEF15}.

Overall, our contributions are summarized as follows.
\begin{itemize}
  \item We formulate ARTAs and show that alternation improves the succinctness of real-time languages, although it does not increase expressive power.
  \item We propose the \ALstarRTA{} algorithm for the active learning of ARTAs and analyze its query complexity.
  \item We empirically show that \ALstarRTA{} usually learns a smaller automaton than \NLstarRTA{}, at the cost of more queries.
\end{itemize}
\paragraph{Organization.}

The remainder of this paper is organized as follows.
After reviewing real-time automata and AFAs in \cref{section:preliminaries},
we define ARTAs and discuss their expressive features in \cref{section:ARTA}.
Then, we present the learning algorithm and its termination in \cref{section:learning}.
Finally, we empirically evaluate \ALstarRTA{} in \cref{section:experiments}, discuss related work in \cref{section:related_work}, and conclude in \cref{section:conclusions}.

\section{Preliminaries}\label{section:preliminaries}

Let $\Alphabet$ be a finite alphabet, let $\Rnn$ be the set of non-negative reals, and let $\setN$ be the set of natural numbers.
For sets $X, Y$, we denote their symmetric difference (\ie{} $(X \cup Y) \setminus (X \cap Y)$) as $X \setdiff Y$.
For a set $X$, a word is a finite sequence of $X$.
A language is a set of words.
We let
$X^*$ be the set of words over $X$ and
$\emptyword$ be the empty word.
For $\word, \word' \in X^*$, we denote their concatenation by $\word \cdot \word'$.
For a word $\word$, we let $\suffixes(\word)$ be the set of suffixes of $\word$.
We let $\Intervals$ be the set of intervals over $\Rnn$ whose endpoints are in $\setN \cup \{\infty\}$.

\subsection{Timed words and real-time automata}

A \emph{timed word} $\word = (\action_1, \delay_1), (\action_2, \delay_2), \dots, (\action_n, \delay_n) \in (\Alphabet \times \Rnn)^*$ over $\Alphabet$ is a word over $\Alphabet \times \Rnn$, where each $\delay_i$ represents the delay between $\action_{i -1}$ and $\action_{i}$\footnote{We use a sequence of letters with (relative) delays rather than (absolute) timestamps for the sake of presentation; this choice is not essential.}.
We let $\TimedWords$ be the set of timed words over $\Alphabet$.
A \emph{timed language} is a set of timed words.
For a timed language $\lang \subseteq \TimedWords$, we let $\letters{\lang} \subseteq \Alphabet \times \Rnn$ be the set of letters with delays in $\lang$, \ie{} $\letters{\lang} = \{(\action,\delay) \mid \exists \word, \word' \in \TimedWords.\, \word \cdot (\action, \delay) \cdot \word' \in \lang\}$.

\begin{definition}[real-time automata]
  A \emph{real-time automaton (RTA)}~\cite{DBLP:journals/jalc/Dima01} is a 5-tuple
  $\TA = (\Loc, \Alphabet, \initLoc, \Final, \Edge)$, where
  $\Loc$ is a finite set of locations,
  $\Alphabet$ is the alphabet,
  $\initLoc\in\Loc$ is the initial location,
  $\Final\subseteq\Loc$ is the set of accepting locations, and
  $\Edge\subseteq\Loc\times\Alphabet\times\Intervals\times\Loc$ is a finite transition relation.
\end{definition}

A \emph{run} of an RTA $\TA$ over a timed word $\word=(\action_1,\delay_1),(\action_2,\delay_2), \dots, (\action_n,\delay_n)$
is a sequence of locations $\loc_0,\loc_1,\ldots,\loc_n$ such that
for each $i\in\{1,2,\ldots,n\}$ there is an interval $I_i\in\Intervals$ satisfying
$(\loc_{i-1},\action_i,I_i,\loc_i)\in\Edge$ and $\delay_i\in I_i$.
A run is \emph{accepting} if its last location is in $\Final$.
A timed word $\word$ is accepted by an RTA $\TA$ if there is an accepting run of $\TA$ over $\word$.
The timed language recognized by $\TA$ is the set of timed words accepted by $\TA$.
We let $\Lg(\TA)$ be the timed language recognized by $\TA$.
A \emph{real-time language} is a timed language recognized by an RTA.\@

An RTA is \emph{deterministic} if, for any location $\loc\in\Loc$, letter $\action\in\Alphabet$, and 
$(\loc, \action, I_1, \loc_1), (\loc, \action, I_2, \loc_2) \in \Edge$,
$(I_1, \loc_1) \neq (I_2, \loc_2)$ implies $I_1 \cap I_2 = \emptyset$.

\begin{definition}[regions of delays~\cite{DBLP:journals/tcs/AlurD94,DBLP:journals/chinaf/AnWZZZ21}]
\label{definition:regions_of_delays}
Let $\TA$ be an RTA and let $K \in \setN$ be the largest integer appearing in the intervals on transitions of $\TA$.
For each $t \in\Rnn$, we define the \emph{region} of $t$, denoted by $\llbracket t\rrbracket$, as\LongVersion{ follows:
\[
  \llbracket t \rrbracket =
  \begin{cases}
    [t,t] & \text{if we have $t \in \setN$ and $t \leq K$},\\
    (\lfloor t\rfloor, \lfloor t\rfloor+1) & \text{if we have $t\notin\setN$ and $t < K$},\\
    (K,\infty) & \text{if we have $t > K$}.
  \end{cases}
\]}\ShortVersion{
$\llbracket t \rrbracket = [t,t]$ if we have $t \in \setN$ and $t \leq K$,
$\llbracket t \rrbracket = (\lfloor t\rfloor, \lfloor t\rfloor+1)$ if we have $t \not\in \setN$ and $t \leq K$, and
$\llbracket t \rrbracket = (K, \infty)$ if we have $t > K$.
}
\end{definition}

The following Myhill-Nerode-style theorem holds for real-time languages.

\begin{theorem}[\cite{DBLP:journals/chinaf/AnWZZZ21}]
  \label{theorem:myhill_nerode}
  For a timed language $\lang \subseteq \TimedWords$, we define an equivalence relation ${\nerode{\lang}} \subseteq \TimedWords \times \TimedWords$ by
  $\word_1 \nerode{\lang} \word_2$
  if and only if 
  $\word_1 \cdot \word\in \lang \iff \word_2 \cdot \word\in \lang$ holds for any $\word\in\TimedWords$.
  $\lang$ is a real-time language if and only if
  the number of equivalence classes of $\nerode{\lang}$ is finite and $\nerode{\lang}$ satisfies the following\LongVersion{ conditions} for some $K \in \setN$:
  \begin{itemize}
    \item For any $\word \in \TimedWords$, $\action \in\Alphabet$, and $t_1,t_2 \in \Rnn$, $\llbracket t_1\rrbracket=\llbracket t_2\rrbracket$ implies $\word \cdot(\action,t_1)\nerode{\lang} \word\cdot(\action,t_2)$.
    \item For any $\word \in \TimedWords$, $\action\in\Alphabet$, and $t_1,t_2\in\Rnn$, $t_1 > K \land t_2 > K$ implies $\word \cdot (\action,t_1)\nerode{\lang} \word \cdot(\action,t_2)$.\qed{}
  \end{itemize}
\end{theorem}

Moreover, for any real-time language $\lang$, there is a \emph{unique} minimal DRTA $\TA$ such that $\Lg(\TA) = \lang$, and the number of locations of $\TA$ equals the number of equivalence classes of $\nerode{\lang}$. 

\subsection{Alternating finite automata}

An \emph{alternating finite automaton (AFA)}~\cite{DBLP:journals/jacm/ChandraKS81} generalizes DFAs by allowing both
existential and universal branching.
For a finite set $X$, we let $\Boolean{X}$ be the set of positive Boolean expressions over $X$, \ie{} the set defined as follows:
\[
  \formula \Coloneq \top \mid \bot \mid x \in X \mid  \formula \land \formula \mid \formula \lor \formula.
\]
For a set $\Loc$ of locations, we call $\formula \in \Boolean{\Loc}$ a \emph{location formula}.

\begin{definition}[alternating finite automata]
  An \emph{alternating finite automaton (AFA)} is a 5-tuple $\A=(\Loc,\Alphabet,\InitLoc,\Final, \transition)$, where
  $\Loc$ is a finite set of locations,
  $\Alphabet$ is a finite alphabet,
  $\InitLoc \in \Boolean{\Loc}$ is the initial location formula,
  $\Final\subseteq\Loc$ is the set of accepting locations, and
  $\transition \colon \Loc \times \Alphabet \to \Boolean{\Loc}$ is the transition function.
\end{definition}

We generalize $\transition$ to $\transition^* \colon \Boolean{\Loc}\times\Alphabet^*\to \Boolean{\Loc}$ as follows:
\begin{gather*}
  \transition^*(\formula, \emptyword) = \formula \qquad
  \transition^*(\top, \word) = \top \qquad
  \transition^*(\bot, \word) = \bot \qquad
  \transition^*(\loc, \action \cdot \word) = \transition^*(\transition(\loc, \action), \word) \\
  \transition^*(\formula \lor \formula', \word) = \transition^*(\formula, \word) \lor \transition^*(\formula', \word)\qquad
  \transition^*(\formula \land \formula', \word) = \transition^*(\formula, \word) \land \transition^*(\formula', \word).
\end{gather*}

For a location formula $\formula$ of an AFA $\A$,
the evaluation function $E_{\A} \colon \Boolean{\Loc} \to \{\top,\bot\}$ is defined as follows:
\begin{gather*}
  E_{\A}(\top) = \top \qquad
  E_{\A}(\bot) = \bot \qquad
  E_{\A}(\loc) = \top \text{ if $\loc \in \Final$} \qquad
  E_{\A}(\loc) = \bot \text{ if $\loc \notin \Final$} \\
  E_{\A}(\formula \land \formula') = E_{\A}(\formula) \land E_{\A}(\formula') \qquad
  E_{\A}(\formula \lor \formula') = E_{\A}(\formula) \lor E_{\A}(\formula')
\end{gather*}

A word $\word \in \Alphabet^*$ is \emph{accepted} by an AFA $\A$ if we have $E_{\A}(\transition^*(\InitLoc, \word)) = \top$.
The language $\Lg(\A) \subseteq \Alphabet^*$ recognized by $\A$ is the set of words accepted by $\A$.

\subsection{The \ALstar{} algorithm for active learning of AFAs}\label{section:alstar}
The \emph{\ALstar{} algorithm}~\cite{DBLP:conf/ijcai/AngluinEF15} learns an AFA $\hypothesisA$ recognizing the target language $\targetLg \subseteq \Alphabet^*$ using finitely many \emph{membership} and \emph{equivalence} queries to the teacher.
In a membership query, the learner asks if $\word \in \Alphabet^*$ belongs to the target language $\targetLg$, \ie{} $\word \in \targetLg$.
In an equivalence query, the learner asks if the hypothesis automaton $\hypothesisA$ recognizes the target language $\targetLg$, \ie{} $\Lg(\hypothesisA) = \targetLg$.
When we have $\Lg(\hypothesisA) \neq \targetLg$,
the teacher returns a counterexample $\cex \in \Lg(\hypothesisA) \setdiff \targetLg$.

\begin{figure}[tbp]
  \begin{subfigure}{.55\linewidth}
    \centering
    \scriptsize
    \begin{tabular}{r|cccccc}
      & $\emptyword$ & $a$ & $b$ & $ab$ & $ba$ \\\hline
      \rowcolor{gray!15}
      $\emptyword$ & $\bot$ & $\top$ & $\bot$ & $\top$ & $\top$\\
      \rowcolor{gray!15}
      $a$ & $\top$ & $\top$ & $\top$ & $\bot$ & $\top$\\
      $b$ & $\bot$ & $\top$ & $\bot$ & $\bot$ & $\bot$ \\
      $aa$ & $\top$ & $\top$ & $\bot$ & $\bot$ & $\bot$ \\
      $ab$ & $\top$ & $\top$ & $\bot$ & $\bot$ & $\bot$ \\
      \rowcolor{gray!15}
      $ba$ & $\top$ & $\top$ & $\bot$ & $\bot$ & $\bot$ \\
      $bb$ & $\bot$ & $\top$ & $\bot$ & $\bot$ & $\bot$ \\
      $baa$ & $\top$ & $\top$ & $\bot$ & $\bot$ & $\bot$ \\
      $bab$ & $\bot$ & $\bot$ & $\bot$ & $\bot$ & $\bot$ \\
    \end{tabular}
    \caption{An observation table. Shaded rows show a minimum-cardinality monotone basis.}%
    \label{figure:ALstar:observation_tables}
  \end{subfigure}
  \hfill
  \begin{subfigure}{.44\linewidth}
    \centering
    \begin{tikzpicture}[auto, semithick,node distance=1.5cm,scale=0.95,every node/.style={initial text={},transform shape}]
      \node[state, initial] (q0) at (0, 0) {$\loc_{0}$};
      \node[state,accepting] (q1) at (2.5, -2.0) {$\loc_{1}$};
      \node[state,accepting] (q2) at (2.5, 0.0) {$\loc_{2}$};
      \node[rectangle,fill=black,minimum width=.3cm,minimum height=.3cm] (q3) at (0, -2.0) {};

      \path[->]
        (q0) edge[bend left=0] node[above] {$a$} (q2)
        (q0) edge[bend left=10] node[pos=0.5,right] {$b$} (q3)
        (q1) edge[loop right] node {$a$} (q1)
        (q2) edge[bend left=10] node[right] {$a$} (q1)
        (q2) edge[bend right=10] node[left] {$b$} (q1)
        (q3) edge[bend left=10] node[above] {} (q0)
        (q3) edge[bend left=0] node[above] {} (q2)
        ;

    \end{tikzpicture}
    \caption{An AFA $\hypothesisA$.}%
    \label{figure:ALstar:AFA}
  \end{subfigure}
  \caption{An observation table for \ALstar{} and the corresponding AFA. The initial location formula of $\hypothesisA$ is $\loc_0$. The black square represents a universal branching.}%
  \label{figure:ALstar}
\end{figure}
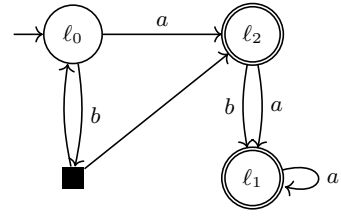
\ALstar{} uses a 2-dimensional array $\Table$ called an \emph{observation table} to maintain the information obtained during learning.
See \cref{figure:ALstar:observation_tables} for an illustration.
For finite index sets $\PrefixSet, \SuffixSet \subseteq \Alphabet^*$,
for each pair $(\prefix, \suffix) \in \PrefixSet \times \SuffixSet$,
the observation table stores whether\ShortVersion{ $\prefix \cdot \suffix \in \targetLg$}\LongVersion{ the concatenation $\prefix \cdot \suffix$ is a member of the target language $\targetLg$}.
Since $\PrefixSet$ and $\SuffixSet$ are finite, one can fill the observation table using finite membership queries.

We write $\TableCell{\prefix}{\suffix}\in\{\top,\bot\}$ for the table entry at row $\prefix\in\PrefixSet$ and column $\suffix\in\SuffixSet$.
For each $\prefix\in\PrefixSet$, we define its \emph{row vector} $\TableRow{\prefix}\in\{\top,\bot\}^{\SuffixSet}$ by
$\TableRow{\prefix}(\suffix)=\TableCell{\prefix}{\suffix}$ for each $\suffix\in\SuffixSet$.
For $U\subseteq \PrefixSet$, we write $\mathrm{Rows}(U)=\{\TableRow{u}\mid u\in U\}$.
For two row vectors $r_1,r_2\in\{\top,\bot\}^{\SuffixSet}$, we define their conjunction and disjunction $r_1 \land r_2, r_1 \lor r_2 \in \{\top, \bot\}^{\SuffixSet}$ in a pointwise manner.
Namely, $r_1 \land r_2$ and $r_1 \lor r_2$ satisfy the following for each $\suffix \in \SuffixSet$:
$(r_1 \land r_2)(\suffix) = r_1(\suffix) \land r_2(\suffix)$ and
$(r_1 \lor r_2)(\suffix) = r_1(\suffix) \lor r_2(\suffix)$.

In \ALstar{}, one central notion is the \emph{monotone basis} of a set of row vectors.
For a set $U \subseteq \{\top,\bot\}^\SuffixSet$ of row vectors, we overload $\Boolean{{U}}$ to denote the set of row vectors obtained by evaluating positive Boolean expressions over $U$, where $\top$ and $\bot$ denote constant row vectors.

\begin{definition}[monotone basis, $\Basis$-closedness, minimality of $\Basis$]
  \label{definition:monotone_basis}
  Let $\Table$ be an observation table with finite index sets $\PrefixSet, \SuffixSet \subseteq \Alphabet^*$.
  For $R \subseteq \{\top,\bot\}^{\SuffixSet}$,
  $U \subseteq R$ is a \emph{monotone basis of $R$} if we have $R \subseteq \Boolean{U}$.
  Moreover, $U$ is a \emph{minimal monotone basis of $R$} if $R \subseteq \Boolean{U}$ holds and no strict subset $U'\subsetneq U$ satisfies $R \subseteq \Boolean{U'}$.
  $\Basis \subseteq \PrefixSet$ is a \emph{monotone basis of $\Table$} if $\mathrm{Rows}(\Basis)$ is a monotone basis of $\mathrm{Rows}(\PrefixSet)$ and for any $\basis,\basis' \in \Basis$, $\basis \neq \basis'$ implies $\TableRow{\basis} \neq \TableRow{\basis'}$.
  For $\Basis \subseteq \PrefixSet$, $\Table$ is \emph{$\Basis$-closed} if $\Basis$ is a monotone basis of $\Table$\footnote{We relax the definition of $\Basis$-closedness for consistency with \ALstarRTA{} in \cref{section:learning}.}.
  A monotone basis $\Basis$ is \emph{minimal} if any $\Basis' \subsetneq \Basis$ is not a monotone basis of $\Table$.
\end{definition}
\begin{algorithm}[tbp]
 \caption{A variant of the \ALstar{} algorithm~\cite{DBLP:conf/ijcai/AngluinEF15} for AFA learning.}%
 \label{algorithm:ALstar}
 \DontPrintSemicolon{}
 \SetKwFunction{FConstructTable}{ConstructTable}
 \SetKwFunction{FConstructA}{ConstructAFA}
 \SetKwFunction{FComputeBasis}{ComputeBasis}
 \scriptsize
 \Input{Alphabet $\Alphabet$ and access to membership and equivalence queries for the target language $\targetLg \subseteq \Alphabet^*$}
 \Output{An AFA $\hypothesisA$ satisfying $\Lg(\hypothesisA)=\targetLg$}
 $\PrefixSet \gets \{\emptyword\};\; \SuffixSet \gets \{\emptyword\};\; \Basis \gets \emptyset$\;
 \While{$\KTrue$} {\nllabel{algorithm:ALstar:main_loop:begin}
   $\PrefixSet \gets \PrefixSet \cup \{ \basis \cdot \action \mid \basis \in \Basis, \action \in \Alphabet\}$\;
   \While{the observation table is not $\Basis$-closed} {\nllabel{algorithm:ALstar:close:begin}
     \KwPush{} $\prefix \in \PrefixSet$ such that $\TableRow{\prefix} \not\in \Boolean{\mathrm{Rows}(\Basis)}$ \KwTo{} $\Basis$\;
     $\PrefixSet \gets \PrefixSet \cup \{ \basis \cdot \action \mid \basis \in \Basis, \action \in \Alphabet\}$\;
   }
   $\Basis \gets \FComputeBasis{\PrefixSet, \SuffixSet, \Basis, \Table}$ \tcp*[r]{Compute a minimal monotone basis}\nllabel{algorithm:ALstar:minimize_basis}
   $\hypothesisA \gets \FConstructA{\PrefixSet, \SuffixSet, \Basis, \Table}$\;
   \Switch{$\eqQ[\targetLg]{\hypothesisA}$} {
     \Case{$\top$} {
       \KwReturn{$\hypothesisA$}\nllabel{algorithm:ALstar:eqQ:end}
     }
     \Case{$\cex$} {\nllabel{algorithm:ALstar:eqQ:cex}
       $\SuffixSet \gets \SuffixSet \cup \suffixes(\cex)$\;\nllabel{algorithm:ALstar:add_cex}
     }
   }
 }\nllabel{algorithm:ALstar:main_loop:end}
\end{algorithm}
\cref{algorithm:ALstar} outlines a variant of the \ALstar{} algorithm, partially inspired by the \ALdstar{} algorithm~\cite{DBLP:journals/iandc/BerndtLLR22}.
We start from $\PrefixSet = \SuffixSet = \{\emptyword\}$ and gradually increase them.
In each iteration, the learner first increases the candidate basis $\Basis$ and $\PrefixSet$ to ensure the observation table is $\Basis$-closed and $\Basis \cdot \Alphabet \subseteq \PrefixSet$ holds.
The learner then reduces $\Basis$ to be minimal.
Once these conditions are satisfied, the learner constructs a hypothesis AFA $\hypothesisA$ and asks an equivalence query.
The constructed AFA is such that 
\begin{ienumeration}
  \item one location is constructed for each $\basis \in \Basis$,
  \item the initial location formula corresponds to the formula representing $\TableRow{\emptyword}$,
  \item the accepting locations correspond to $\basis \in \Basis$ such that $\TableCell{\basis}{\emptyword} = \top$, and 
  \item the successor from a location corresponding to $\basis \in \Basis$ with $\action \in \Alphabet$ is the formula representing $\TableRow{\basis \cdot \action}$.
\end{ienumeration}
\cref{figure:ALstar:AFA} shows a concrete example.
If $\Lg(\hypothesisA) = \targetLg$ holds, the algorithm terminates and returns $\hypothesisA$.
Otherwise, the teacher returns\LongVersion{ a counterexample} $\cex \in \Lg(\hypothesisA) \setdiff \targetLg$, and the learner adds all suffixes of $\cex$ to $\SuffixSet$ to refine the observation table.

\section{Alternating real-time automata}\label{section:ARTA}

We introduce \emph{alternating real-time automata (ARTAs)} by combining notions of RTAs and AFAs.
Informally, an ARTA extends an RTA by allowing transition targets to be location formulas in $\Boolean{\Loc}$, which enables existential and universal branching over timed transitions.

\begin{definition}[alternating real-time automata]
  An \emph{alternating real-time automaton (ARTA)} is a 5-tuple
  $\TA=(\Loc,\Alphabet,\InitLoc,\Final, \Edge)$, where
  $\Loc$ is a finite set of locations, 
  $\Alphabet$ is a finite alphabet,
  $\InitLoc \in \Boolean{\Loc}$ is the initial location formula,
  $\Final \subseteq \Loc$ is the set of accepting locations, and
  $\Edge \subseteq \Loc \times \Alphabet \times \Intervals \times \Boolean{\Loc}$ is a finite transition relation satisfying the following condition on determinism:
  for any $(\loc,\action,I_1,\formula_1), (\loc,\action,I_2,\formula_2)\in\Edge$,
  $(\loc,\action,I_1,\formula_1) \neq (\loc,\action,I_2,\formula_2)$ implies $I_1 \cap I_2 = \emptyset$.
\end{definition}

Thanks to the condition above on the determinism of the transition relation $\Edge$,
for an ARTA $\TA = (\Loc,\Alphabet,\InitLoc,\Final, \Edge)$, one can define a transition function $\transition\colon \Loc \times \Alphabet \times \Rnn \to \Boolean{\Loc}$ as follows:
\begin{align*}
  \transition(\loc,(\action,\delay)) &= \phi
  &&\text{if there exists }(\loc,\action,I,\phi)\in\Edge\text{ with }\delay\in I,\\
  \transition(\loc,(\action,\delay)) &= \bot
  &&\text{if no such edge exists.}
\end{align*}

We generalize $\transition$ to $\transition^* \colon \Boolean{\Loc}\times(\Alphabet\times\Rnn)^*\to \Boolean{\Loc}$ as follows:
\begin{gather*}
  \transition^*(\formula, \emptyword) = \formula \quad
  \transition^*(\top, \word) = \top \quad
  \transition^*(\bot, \word) = \bot \quad
  \transition^*(\loc, (\action, \delay) \cdot \word) = \transition^*(\transition(\loc, (\action, \delay)), \word) \\
  \transition^*(\formula \lor \formula', \word) = \transition^*(\formula, \word) \lor \transition^*(\formula', \word)\qquad
  \transition^*(\formula \land \formula', \word) = \transition^*(\formula, \word) \land \transition^*(\formula', \word).
\end{gather*}

For a location formula $\formula$ of an ARTA $\TA = (\Loc,\Alphabet,\InitLoc,\Final, \Edge)$,
the evaluation function $E_{\TA} \colon \Boolean{\Loc} \to \{\top,\bot\}$ is defined as follows:
\begin{gather*}
  E_{\TA}(\top) = \top \qquad
  E_{\TA}(\bot) = \bot \qquad
  E_{\TA}(\loc) = \top \text{ if $\loc \in \Final$} \qquad
  E_{\TA}(\loc) = \bot \text{ if $\loc \notin \Final$} \\
  E_{\TA}(\formula \land \formula') = E_{\TA}(\formula) \land E_{\TA}(\formula') \qquad
  E_{\TA}(\formula \lor \formula') = E_{\TA}(\formula) \lor E_{\TA}(\formula')
\end{gather*}

An ARTA $\TA = (\Loc,\Alphabet,\InitLoc,\Final, \Edge)$ \emph{accepts} a timed word $\word \in \TimedWords$ if we have $E_{\TA}(\transition^*(\InitLoc, \word)) = \top$.
The timed language $\Lg(\TA)$ recognized by an ARTA $\TA$ is the set of timed words accepted by $\TA$.

\begin{example}
 \label{example:ARTA}
 \cref{figure:running_arta} depicts an ARTA $\TA$ over $\Alphabet = \{a,b\}$.
 $\TA$ accepts $(a, 1),(b, 4)$ because we have $E_{\TA}(\transition^*(\loc_0, (a, 1), (b, 4))) = E_{\TA}(\transition^*(\loc_1, (b, 4))) = E_{\TA}(\loc_2) = \top$.
 $\TA$ also accepts $(b, 2.7), (b, 3.5)$ because we have $\transition(\loc_0, (b, 2.7)) = \loc_0 \land \loc_1$, $\transition(\loc_0, (b, 3.5)) = (\loc_0 \land \loc_1) \lor \loc_2$, and $\transition(\loc_1, (b, 3.5)) = \loc_2$, and thus, $E_{\TA}(\transition^*(\loc_0, (b, 2.7), (b, 3.5))) = E_{\TA}(\transition^*(\loc_0 \land \loc_1, (b, 3.5))) = E_{\TA}(((\loc_0 \land \loc_1) \lor \loc_2) \land \loc_2) = \top$ holds.
 In contrast, $\TA$ does not accept $(b, 2.7), (a, 1)$ because we have $E_{\TA}(\transition^*(\loc_0, (b, 2.7), (a, 1))) = E_{\TA}(\transition^*(\loc_0 \land \loc_1, (a, 1))) = E_{\TA}(\loc_1 \land \loc_0) = \bot$.
\end{example}

As in the classical setting, where AFAs accept exactly the regular languages~\cite{DBLP:journals/jacm/ChandraKS81}, alternation does not increase the expressive power of RTAs, \ie{} ARTAs recognize real-time languages.

\newcommand{\expressivePowerStatement}{%
  DRTAs and ARTAs have the same expressive power.
  Moreover, for any ARTA with $n$ locations,
  there is a DRTA recognizing the same timed language with at most $2^{2^{n}}$ locations.
}
\begin{theorem}[expressive power of ARTAs]%
 \label{theorem:expressive_power}
 \expressivePowerStatement{}
 \qed{}
\end{theorem}

This double-exponential blow-up in determinization is unavoidable in general.

\newcommand{\succinctnessStatement}{%
  There exist a constant $c>0$ and a family of ARTAs $(\TA_k)_{k \geq 1}$ with at most $k$ locations such that for any sufficiently large $k$,
  any DRTA $\TA_\mathcal{D}$ satisfying $\Lg(\TA_\mathcal{D}) = \Lg(\TA_k)$ has at least $2^{2^{c k}}$ locations.
}
\begin{theorem}[succinctness]%
 \label{theorem:arta_drta_succinctness}
 \succinctnessStatement{}
 \qed{}
\end{theorem}
\section{Active learning algorithm for ARTAs}\label{section:learning}

We present our algorithm \ALstarRTA{} for the active learning of ARTAs,
extending \ALstar{}\LongVersion{ (\cref{section:alstar})} with ideas used in learning real-time automata~\cite{DBLP:journals/chinaf/AnWZZZ21}\LongVersion{ and symbolic automata~\cite{DBLP:conf/tacas/DrewsD17}} to handle infinite domain inputs, \ie{} $\Alphabet \times \Rnn$ rather than $\Alphabet$.
As in \ALstar{}, the learner interacts with a teacher through membership and equivalence queries.
In \ALstarRTA{}, membership queries are asked on timed words and equivalence queries are asked on ARTAs.

\begin{algorithm}[t]
  \SetKwFunction{FConstructARTA}{ConstructARTA}
  \SetKwFunction{FComputeBasis}{ComputeBasis}
  \caption{\ALstarRTA{} for active learning of ARTAs}%
  \label{algorithm:ALstarRTA}
  \DontPrintSemicolon{}
  \footnotesize
  \Input{Alphabet $\Alphabet$ and access to membership and equivalence queries for the target real-time language $\targetLg \subseteq \TimedWords$}
  \Output{An ARTA $\hypothesisA$ satisfying $\Lg(\hypothesisA)=\targetLg$}
  $\PrefixSet \gets \{\emptyword\}; \; \SuffixSet \gets \{\emptyword\}\; ; \Basis \gets \emptyset$\;\nllabel{algorithm:ALstarRTA:init}
  \While{$\KTrue$}{\nllabel{algorithm:ALstarRTA:main_loop:begin}
   \If{the observation table is not $\Basis$-closed} {\nllabel{algorithm:ALstarRTA:inner_begin}
    \KwPush{} $\prefix \in \PrefixSet$ such that $\TableRow{\prefix} \not\in \Boolean{\mathrm{Rows}(\Basis)}$ \KwTo{} $\Basis$\;\nllabel{algorithm:ALstarRTA:close}
   } \ElseIf{$\exists \basis \in \Basis, \sigma \in \letters{\PrefixSet \cup \SuffixSet}$ such that $\basis \cdot \sigma \notin \PrefixSet$} {
      \KwPush{} $\basis \cdot \sigma$ \KwTo{} $\PrefixSet$ \tcp*[r]{Ensure evidence-closedness}\nllabel{algorithm:ALstarRTA:evidence_closed}\nllabel{algorithm:ALstarRTA:inner_end}
    } \Else {
      $\Basis' \gets \FComputeBasis{\PrefixSet, \SuffixSet, \Basis, \Table}$ \tcp*[r]{Minimize the monotone basis}\nllabel{algorithm:ALstarRTA:minimize_basis}
      \lIf {$|\Basis'| < |\Basis|$} {
        $\Basis \gets \Basis'$
      } \Else {
        \LongVersion{\tcp{Use the construction in \cref{section:arta_construction}}}
        $\hypothesisA \gets \FConstructARTA{\PrefixSet, \SuffixSet, \Basis, \Table}$\;\nllabel{algorithm:ALstarRTA:construct_hypothesis}
        \Switch{$\eqQ[\targetLg]{\hypothesisA}$} {\nllabel{algorithm:ALstarRTA:eq_query}
          \Case{$\top$} {
            \KwReturn{$\hypothesisA$}\nllabel{algorithm:ALstarRTA:eqQ:end}
          }
          \Case{$\cex$} {\nllabel{algorithm:ALstarRTA:eqQ:cex}
            $\normalized{\cex} \gets \normalize(\cex)$ \tcp*[r]{Normalize the counterexample}\nllabel{algorithm:ALstarRTA:normalize_cex}
            $\SuffixSet \gets \SuffixSet \cup \suffixes(\normalized{\cex})$\;\nllabel{algorithm:ALstarRTA:add_cex}
         }
        }
      }
    }
  }\nllabel{algorithm:ALstarRTA:main_loop:end}
\end{algorithm}
\cref{algorithm:ALstarRTA} outlines \ALstarRTA{}.
Compared to \ALstar{}, the key new ingredients are
to lift the observation table from words to timed words, and
to impose evidence-closedness on observation tables to handle real-valued timestamps.

\subsection{Observation tables for \ALstarRTA{}}

As in \ALstar{}, we use an observation table to store the answers to membership queries for concatenations of a row element and a column element.

\begin{definition}[observation tables for \ALstarRTA{}]%
  \label{definition:timed_observation_table}
  Let $\targetLg \subseteq \TimedWords$ be the target real-time language.
  An observation table for \ALstarRTA{} is a 2-dimensional array $\Table \colon \PrefixSet \times \SuffixSet \to \{\top, \bot\}$ with
  finite index sets $\PrefixSet, \SuffixSet \subseteq \TimedWords$ such that
  \begin{ienumeration}
    \item $\TableCell{\prefix}{\suffix}=\top$ if and only if $\prefix\cdot \suffix\in \targetLg$,
    \item $\PrefixSet$ is prefix-closed, and
    \item $\SuffixSet$ is suffix-closed.
  \end{ienumeration}
\end{definition}

For each $\prefix\in\PrefixSet$, we write $\TableRow{\prefix}\in\{\top,\bot\}^{\SuffixSet}$ for its row vector, defined by $\TableRow{\prefix}(\suffix)=\TableCell{\prefix}{\suffix}$.
For $U\subseteq \PrefixSet$, we write $\mathrm{Rows}(U)=\{\TableRow{u}\mid u\in U\}$.
Row vectors are combined by pointwise conjunction and disjunction as in \cref{section:alstar}.

\begin{definition}[cohesiveness]
  \label{definition:arta_table_properties}
  Let $\Table$ be an observation table for \ALstarRTA{} with index sets $\PrefixSet, \SuffixSet \subseteq \TimedWords$.
  For $\Basis \subseteq \PrefixSet$, $\Table$ is \emph{$\Basis$-closed} if we have $\mathrm{Rows}(\PrefixSet)\subseteq \Boolean{\mathrm{Rows}(\Basis)}$.
  For $\Basis \subseteq \PrefixSet$, $\Table$ is \emph{evidence-closed} if $\basis\cdot \sigma \in \PrefixSet$ holds for any $\basis\in\Basis$ and $\sigma\in\letters{\PrefixSet \cup \SuffixSet}$.
  $\Table$ is \emph{floor-distinct} if $(\action,\delay_1),(\action,\delay_2)\in\PhiTable$ and $\delay_1\neq \delay_2$ and both $\delay_1, \delay_2 \notin \setN$ implies $\lfloor \delay_1\rfloor\neq \lfloor \delay_2\rfloor$.
  For $\Basis \subseteq \PrefixSet$, $\Table$ is \emph{cohesive} if it is $\Basis$-closed, evidence-closed, %
  and floor-distinct.
\end{definition}
Intuitively, $\Basis$-closedness enables \ALstar{}-style representability via a monotone basis and evidence-closedness ensures that we have a successor row for any $\basis \in \Basis$ and letter with delay in $\PrefixSet \cup \SuffixSet$.
We require floor-distinctness to ensure that the inferred guard intervals are well-defined in \ALstarRTA{}.

\subsection{ARTA construction from a cohesive observation table}\label{section:arta_construction}

Once we obtain a cohesive observation table  $\Table$, one can construct a hypothesis ARTA $\hypothesisA$. %
The construction is two-fold: first, we construct an evidence AFA, where each transition is labeled with a letter with a delay, and then we generalize concrete delays on transitions to obtain an ARTA.
The construction of an evidence AFA is essentially the same as the AFA construction in \ALstar{}.

 Let $\Table$ be an observation table with index sets $\PrefixSet, \SuffixSet \subseteq \TimedWords$ and let $\Basis\subseteq \PrefixSet$.
 Let $\RowsP = \mathrm{Rows}(\Basis)\subseteq \{\top,\bot\}^{\SuffixSet}$.
 For each $\suffix\in\SuffixSet$, we define a basis formula
 $M_{\RowsP}(\suffix) = \bigwedge_{\substack{r\in\RowsP, r(\suffix)=\top}} r$,
 and for each row vector $r\in \{\top,\bot\}^{\SuffixSet}$ we define its decomposition
 $b_{\RowsP}(r) = \bigvee_{\substack{\suffix\in\SuffixSet, r(\suffix)=\top}} M_{\RowsP}(\suffix)$,
 where empty conjunctions and disjunctions are interpreted as $\top$ and $\bot$, respectively.
\begin{definition}[evidence AFA]
 \label{definition:evidence_afa}
 Assume the observation table $\Table$ with index sets $\PrefixSet, \SuffixSet \subseteq \TimedWords$ and monotone basis $\Basis\subseteq \PrefixSet$ is cohesive.
 The evidence AFA is the AFA $\evidenceA = (\RowsP, \PhiTable, b_{\RowsP}(\TableRow{\emptyword}), \Final_{\evidenceA}, \transition_{\evidenceA})$, where
 the accepting locations are $\Final_{\evidenceA} = \{r\in\RowsP\mid r(\emptyword)=\top\}$ and
 the transition function $\transition_{\evidenceA}\colon \RowsP\times \PhiTable \to \Boolean{\RowsP}$ is defined by
 $\transition_{\evidenceA}(r,\sigma)=b_{\RowsP}(\TableRow{\basis\cdot \sigma})$ for $\basis\in\Basis$ with $\TableRow{\basis}=r$.
 Such $\basis \in \Basis$ is unique by the definition of monotone bases.
\end{definition}

The evidence AFA $\evidenceA$ is defined over the \emph{finite} set $\PhiTable \subseteq \Alphabet \times \Rnn$ of letters with delays appearing in\LongVersion{ the observation table} $\Table$.
The transition function $\transition_{\evidenceA}$ is well-defined by evidence-closedness.
Following~\cite{DBLP:journals/chinaf/AnWZZZ21}, we obtain an ARTA over $\Alphabet$ by inferring guard intervals from these finitely many delays using the \emph{partition function}.

\begin{definition}%
 [partition function~\cite{DBLP:journals/chinaf/AnWZZZ21}]
 \label{definition:partition_function}
 Let $l=\delay_1,\dots,\delay_n$ be a\LongVersion{ finite} sequence of delays satisfying
\begin{ienumeration}
  \item $i<j$ implies $\delay_i<\delay_j$ and
  \item $\lfloor \delay_i\rfloor\neq \lfloor \delay_j\rfloor$ for all $i\neq j$ with $\delay_i, \delay_j \not\in\setN$.
\end{ienumeration}
We let $\delay_0 = -0.5$.
We define the partition function $\partition(l)=I_1,\dots,I_n$, where intervals $I_1,\dots,I_n\in\Intervals$ are given as follows. Note that we have $\partition(\emptyword) = \emptyword$.
\[
  I_i =
  \begin{cases}
    (\delay_{i-1},\delay_{i}] & \text{if $i<n$ and $\delay_{i-1},\delay_{i}\in\setN$}\\
    [\lceil \delay_{i-1}\rceil,\delay_{i}] & \text{if $i<n$, $\delay_{i-1}\notin\setN$, and $\delay_{i}\in\setN$}\\
    (\delay_{i-1},\lceil \delay_{i}\rceil) & \text{if $i<n$, $\delay_{i-1}\in\setN$, and $\delay_{i}\notin\setN$}\\
    [\lceil \delay_{i-1}\rceil,\lceil \delay_{i}\rceil) & \text{if $i<n$ and $\delay_{i-1},\delay_{i}\notin\setN$}\\
    (\delay_{i-1},\infty) & \text{if $i=n$ and $\delay_{i-1}\in\setN$}\\
    [\lceil \delay_{i-1}\rceil,\infty) & \text{if $i=n$ and $\delay_{i-1}\notin\setN$}\\
  \end{cases}
\]
\end{definition}
\begin{algorithm}[t]
\caption{ARTA construction from an evidence AFA}
\label{algorithm:arta_construction}
\DontPrintSemicolon{}
\ShortVersion{\scriptsize}\LongVersion{\small}
\Input{Evidence AFA $\evidenceA=(\RowsP,\PhiTable,\InitLoc,\Final,\transition_{\evidenceA})$}
\Output{Hypothesis ARTA}
\BlankLine
$\Edge_{\hypothesisA}\gets \emptyset$\;
\For{$r\in\RowsP$ \KwAnd $\action\in\Alphabet$}{
    \KwLet{} $l=\delay_1,\dots,\delay_n$ \KwBe{} the increasing sequence of all delays $\delay$ such that $(\action,\delay)\in\PhiTable$\;
    \For{$i\gets 2$ \KwTo $n$}{
      \If {$\transition_{\evidenceA}(r,(\action,\delay_{i-1})) = \transition_{\evidenceA}(r,(\action,\delay_{i}))$} {
        \KwPop{} $\delay_{i-1}$ \KwFrom{} $l$\;
      }
    }
    $\delay'_1,\dots,\delay'_{n'} \gets l$;\,
    $I_1, \dots, I_{n'} \gets \partition(l)$\;
    \For{$i \in \{1, 2, \dots, n'\}$}{
      $\Edge_{\hypothesisA}\gets \Edge_{\hypothesisA}\cup\{(r,\action,I_i,\transition_{\evidenceA}(r,(\action,\delay'_i)))\}$\;
    }
}
\KwReturn{} $(\RowsP,\Alphabet,\InitLoc,\Final,\Edge_{\hypothesisA})$\;
\end{algorithm}
\cref{algorithm:arta_construction} shows the construction of an ARTA from an evidence AFA.
The ARTA returned by \cref{algorithm:arta_construction} shares locations, initial formula, and accepting locations with $\evidenceA$, while only the transition guards are generalized from finitely many observed delays to intervals via $\partition$.

\begin{example}
  One can obtain the ARTA in \cref{figure:running_arta} from the evidence AFA in \cref{figure:running_evidence_AFA} by using \cref{algorithm:arta_construction} and simplifying some transitions.
  Concretely, the partition function $\partition$ generalizes the delays 2, 2.5, 7, and 7.5 on the transitions from $\loc_0$ with $b$ into $[0, 2]$, $(2, 3)$, $[3, 7]$, and $(7, \infty)$.
\end{example}

The following theorem shows that this construction is faithful to $\Table$.
This is essential for the termination of \ALstarRTA{}.

\newcommand{\faithfulnessStatement}{%
 Assume $\Table$ with index sets $\PrefixSet, \SuffixSet \subseteq \TimedWords$ and monotone basis $\Basis\subseteq \PrefixSet$ is cohesive.
 Let $\hypothesisA$ be the ARTA constructed from $\Table$.
 For any $\prefix \in \PrefixSet$ and $\suffix \in \SuffixSet$, we have
 $\TableCell{\prefix}{\suffix} = \top \iff \prefix \cdot \suffix \in \Lg(\hypothesisA)$.
}
\begin{theorem}
 [faithfulness]%
 \label{theorem:faithfulness}
 \faithfulnessStatement{}
 \qed{}
\end{theorem}

Since the endpoints of the intervals produced by \cref{definition:partition_function} are independent of $\delay_n$, the constants in the ARTA constructed by \cref{algorithm:arta_construction} are bounded.

\newcommand{\hypothesisBoundStatement}{%
 Fix $K\in\setN$ as in \cref{theorem:myhill_nerode} for $\targetLg$.
 Assume $\Table$ with index sets $\PrefixSet, \SuffixSet \subseteq \TimedWords$ and monotone basis $\Basis\subseteq \PrefixSet$ is cohesive.
 Let $\hypothesisA$ be the ARTA constructed from $\Table$.
 For any transition $(\loc,\action, I, \formula)$ of $\hypothesisA$,
 every finite integer endpoint in $I$ is at most $K$.
}
\begin{theorem}
 \label{theorem:hypothesis_bound}
 \hypothesisBoundStatement{}
 \qed{}
\end{theorem}
\subsection{The \ALstarRTA{} algorithm}\label{section:algorithm}

Before presenting the \ALstarRTA{} algorithm, we introduce auxiliary functions $\normalizeLetter$ and $\normalize$ to normalize counterexamples returned by equivalence queries.
This technique is commonly used in RTA learning~\cite{DBLP:journals/tecs/AnZZZ21,DBLP:journals/chinaf/AnWZZZ21}.

\begin{definition}\LongVersion{[normalization function]}%
 \label{definition:normalization_functions}
 For $\alpha\in (0,1)$, we let $\normalizeLetter\colon \Alphabet\times\Rnn \to \Alphabet\times\Rnn$ by
 $\normalizeLetter((\action,\delay)) = (\action,\delay)$ if $\delay\in\setN$, and $\normalizeLetter((\action,\delay)) = (\action,\lfloor \delay\rfloor + \alpha)$ otherwise.
 We extend it to timed words $\normalize\colon \TimedWords\to\TimedWords$ by applying $\normalizeLetter$ to each letter with delay.
\end{definition}

It is easy to see that $\normalize$ preserves counterexamples during learning.

\newcommand{\normalizedCounterexampleStatement}{%
 For any real-time language $\lang$,
 for any $\alpha\in (0,1)$,
 for any ARTA $\A$, and
 for any $\word \in \lang \setdiff \Lg(\A)$,
 $\normalize(\word) \in \lang \setdiff \Lg(\A)$ holds.
}
\begin{theorem}\LongVersion{[counterexample preservation]}
 \label{theorem:normalized_counterexample}
 \normalizedCounterexampleStatement{}
 \qed{}
\end{theorem}

\cref{algorithm:ALstarRTA} outlines \ALstarRTA{}.
It maintains an observation table $\Table$ with finite index sets $\PrefixSet,\SuffixSet$ together with a candidate monotone basis $\Basis$, which are initialized at \cref{algorithm:ALstarRTA:init} and updated in the main loop \crefrange{algorithm:ALstarRTA:main_loop:begin}{algorithm:ALstarRTA:main_loop:end}.
When $\PrefixSet$ or $\SuffixSet$ is extended, the newly introduced cells in $\Table$ are filled by using membership queries.

The learner first increases $\Basis$ or $\PrefixSet$ to ensure $\Basis$-closedness (\cref{algorithm:ALstarRTA:close}) and evidence-closedness (\cref{algorithm:ALstarRTA:evidence_closed}).
When neither extension step applies, it computes $\Basis'$ by solving a binary integer program (BIP) (\cref{algorithm:ALstarRTA:minimize_basis}).
If $|\Basis'| < |\Basis|$, it updates $\Basis$ to $\Basis'$.
The details of the BIP encoding are shown in \cref{section:exact_basis_milp}.

Otherwise, the learner constructs a hypothesis $\hypothesisA$ (\cref{algorithm:ALstarRTA:construct_hypothesis}) and asks an equivalence query (\cref{algorithm:ALstarRTA:eq_query}).
If the query succeeds, the algorithm returns the hypothesis (\cref{algorithm:ALstarRTA:eqQ:end}); otherwise the teacher provides a counterexample (\cref{algorithm:ALstarRTA:eqQ:cex}).
On a counterexample, the learner first normalizes it (\cref{algorithm:ALstarRTA:normalize_cex}) and then refines the table by adding suffixes of the normalized counterexample to $\SuffixSet$ (\cref{algorithm:ALstarRTA:add_cex}).
This normalization step maintains the observation table in \cref{algorithm:ALstarRTA} as floor-distinct because
\begin{ienumeration}
 \item $\PhiTable$ increases only at \cref{algorithm:ALstarRTA:add_cex} and
 \item this normalization ensures that for any $(\action, \delay_1), (\action, \delay_2) \in \PhiTable$, $\normalizeLetter((\action, \delay_1)) = \normalizeLetter((\action, \delay_2))$ implies $\delay_1 = \delay_2$, which ensures the floor-distinctness.
\end{ienumeration}
\subsection{Worked example of \ALstarRTA{}}\label{section:worked_example}

\newcommand{\workpair}[2]{(#1,#2)}
\newcommand{\Athree}{\workpair{a}{3}}
\newcommand{\Athreehalf}{\workpair{a}{3.5}}
\newcommand{\Bzero}{\workpair{b}{0}}
\newcommand{\Bonehalf}{\workpair{b}{1.5}}
\newcommand{\Btwohalf}{\workpair{b}{2.5}}
\newcommand{\Bseven}{\workpair{b}{7}}
\newcommand{\Bsevenhalf}{\workpair{b}{7.5}}
\newcommand{\WorkLocZero}{\loc_0}
\newcommand{\WorkLocOne}{\loc_1}
\newcommand{\WorkLocTwo}{\loc_2}
\newcommand{\workbranch}[2]{\node[rectangle,fill=black,minimum width=.3cm,minimum height=.3cm,inner sep=0pt] (#1) at #2 {};}
\newcommand{\workconst}[3]{\node[draw,rectangle,minimum width=.45cm,minimum height=.36cm,inner sep=1pt] (#1) at #2 {$#3$};}

We illustrate \ALstarRTA{} on the target language $\targetLg$ recognized by the ARTA in \cref{figure:running_arta}.
We fix $\alpha=\tfrac12$.
\begin{figure}[t]
 \begin{subfigure}{0.2\linewidth}
  \centering
  \small
  \begin{tabular}{r|c}
   & $\emptyword$\\\hline
   $\emptyword$ & $\bot$\\
  \end{tabular}
  \caption{Initial observation table $\Table_1$.}%
  \label{figure:observation_tables:1}
 \end{subfigure}
 \hfill
 \begin{subfigure}{0.2\linewidth}
  \centering
  \scriptsize
  \begin{tabular}{r|cc}
   & $\emptyword$ & $\Bseven$\\\hline
   \rowcolor{gray!15}
   $\emptyword$ & $\bot$ & $\top$\\
   $\Bseven$ & $\top$ & $\top$\\
  \end{tabular}
  \caption{Second observation table $\Table_2$.}%
  \label{figure:observation_tables:2}
 \end{subfigure}
 \hfill
 \begin{subfigure}{0.2\linewidth}
  \centering
  \small
  \begin{tikzpicture}[auto,scale=0.75,every node/.style={initial text={},transform shape}]
   \node[state, initial] (q0) at (0, 0) {$\WorkLocZero$};
   \workconst{top}{(2.0,0)}{\top}
   \path[->]
   (q0) edge node[above] {$b,[0,\infty)$} (top);
  \end{tikzpicture}
  \caption{Second hypothesis $\hypothesisA^2$.}%
  \label{figure:hypothesis:2}
 \end{subfigure}
 \hfill
 \begin{subfigure}{0.37\linewidth}
  \centering
  \scriptsize
  \begin{tabular}{r|ccc}
   & $\emptyword$ & $\Bseven$ & $\Bsevenhalf$\\\hline
   \rowcolor{gray!15}
   $\emptyword$ & $\bot$ & $\top$ & $\bot$\\
   $\Bseven$ & $\top$ & $\top$ & $\top$\\
   $\Bsevenhalf$ & $\bot$ & $\top$ & $\bot$\\
  \end{tabular}
  \caption{Third observation table $\Table_3$.}%
  \label{figure:observation_tables:3}
 \end{subfigure}
 \hfill
 \begin{subfigure}{0.30\linewidth}
  \centering
  \begin{tikzpicture}[auto, semithick,node distance=1.5cm,scale=0.75,every node/.style={initial text={},transform shape}]
    \node[state, initial] (q0) at (0, 0) {$\WorkLocZero$};
    \workconst{top}{(2,0)}{\top}
    \path[->]
      (q0) edge[loop above] node {$b,(7,\infty)$} (q0)
      (q0) edge node[above] {$b,[0,7]$} (top);
  \end{tikzpicture}
  \caption{Third hypothesis $\hypothesisA^3$.}%
  \label{figure:hypothesis:3}
 \end{subfigure}
 \hfill
 \begin{subfigure}{0.36\linewidth}
  \centering
  \scriptsize
  \begin{tabular}{r|cccc}
   & $\emptyword$ & $\Btwohalf$ & $\Bseven$ & $\Bsevenhalf$\\\hline
   \rowcolor{gray!15}
   $\emptyword$ & $\bot$ & $\bot$ & $\top$ & $\bot$\\
   $\Btwohalf$ & $\bot$ & $\bot$ & $\top$ & $\bot$\\
   $\Bseven$ & $\top$ & $\top$ & $\top$ & $\top$\\
   $\Bsevenhalf$ & $\bot$ & $\bot$ & $\top$ & $\bot$\\
  \end{tabular}
  \caption{Fourth observation table $\Table_4$.}%
  \label{figure:observation_tables:4}
 \end{subfigure}
 \hfill
 \begin{subfigure}{0.32\linewidth}
  \centering
  \small
  \begin{tikzpicture}[auto, semithick,node distance=1.5cm,scale=0.75,every node/.style={initial text={},transform shape}]
    \node[state, initial above] (q0) at (0, 0) {$\WorkLocZero$};
    \workconst{top}{(2,0)}{\top}
    \path[->]
      (q0) edge[loop below] node {$b,[0,3)$} (q0)
      (q0) edge[loop left] node {$b,(7,\infty)$} (q0)
      (q0) edge node[above] {$b,[3,7]$} (top);
  \end{tikzpicture}
  \caption{Fourth hypothesis $\hypothesisA^4$.}%
  \label{figure:hypothesis:4}
 \end{subfigure}
 \hfill
 \begin{subfigure}{.5\linewidth}
  \centering
  \scriptsize
  \begin{tabular}{r|cccccc}
   & $\emptyword$ & $\Bzero$ & $\Btwohalf$ & $\Bseven$ & $\Bsevenhalf$ & $\Athree\Bzero$\\\hline
   \rowcolor{gray!15}
   $\emptyword$ & $\bot$ & $\bot$ & $\bot$ & $\top$ & $\bot$ & $\top$\\
   \rowcolor{gray!15}
   $\Athree$ & $\bot$ & $\top$ & $\top$ & $\top$ & $\top$ & $\bot$\\
   $\Bzero$ & $\bot$ & $\bot$ & $\bot$ & $\bot$ & $\bot$ & $\bot$\\
   $\Btwohalf$ & $\bot$ & $\bot$ & $\bot$ & $\top$ & $\bot$ & $\bot$\\
   \rowcolor{gray!15}
   $\Bseven$ & $\top$ & $\top$ & $\top$ & $\top$ & $\top$ & $\bot$\\
   $\Bsevenhalf$ & $\bot$ & $\bot$ & $\bot$ & $\top$ & $\bot$ & $\bot$\\
   $\Athree\Athree$ & $\bot$ & $\bot$ & $\bot$ & $\top$ & $\bot$ & $\top$\\
   $\Athree\Bzero$ & $\top$ & $\top$ & $\top$ & $\top$ & $\top$ & $\bot$\\
   $\Athree\Btwohalf$ & $\top$ & $\top$ & $\top$ & $\top$ & $\top$ & $\bot$\\
   $\Athree\Bseven$ & $\top$ & $\top$ & $\top$ & $\top$ & $\top$ & $\bot$\\
   $\Athree\Bsevenhalf$ & $\top$ & $\top$ & $\top$ & $\top$ & $\top$ & $\bot$\\
   $\Bseven\Athree$ & $\bot$ & $\bot$ & $\bot$ & $\top$ & $\bot$ & $\bot$\\
   $\Bseven\Bzero$ & $\top$ & $\top$ & $\top$ & $\top$ & $\top$ & $\bot$\\
   $\Bseven\Btwohalf$ & $\top$ & $\top$ & $\top$ & $\top$ & $\top$ & $\bot$\\
   $\Bseven\Bseven$ & $\top$ & $\top$ & $\top$ & $\top$ & $\top$ & $\bot$\\
   $\Bseven\Bsevenhalf$ & $\top$ & $\top$ & $\top$ & $\top$ & $\top$ & $\bot$
  \end{tabular}
  \caption{Fifth observation table $\Table_5$.}%
  \label{figure:observation_tables:5}
 \end{subfigure}
 \hfill
 \begin{subfigure}{0.45\linewidth}
  \centering
  \small
  \begin{tikzpicture}[auto, semithick,node distance=1.5cm,scale=0.78,every node/.style={initial text={},transform shape}]
    \node[state, initial] (q0) at (0, 0) {$\WorkLocZero$};
    \node[state, accepting] (q1) at (4, 0) {$\WorkLocOne$};
    \node[state] (q2) at (4, -3) {$\WorkLocTwo$};
    \workbranch{c12}{(1.7,1.5)}
    \workbranch{c012}{(0,-3)}
    \path[->]
      (q0) edge node[above left] {$a,[0,\infty)$} (c12)
      (q0) edge[bend right=15] node[left] {$b,(0,3)$} (c012)
      (q0) edge node[above,pos=0.3] {$b,[3,7]$} (q1)
      (q0) edge[bend left=15] node[pos=0.3,right] {$b,(7,\infty)$} (c012)
      (q1) edge[bend left=10] node[below right,pos=0.7] {$a,[0,\infty)$} (c012)
      (q1) edge[loop above] node {$b,[0,\infty)$} (q1)
      (q2) edge[bend right=10] node[above right,pos=0.8] {$a,[0,\infty)$} (q0)
      (q2) edge node[right] {$b,[0,\infty)$} (q1)
      (c12) edge (q1)
      (c12) edge (q2)
      (c012) edge (q0)
      (c012) edge[bend left=10] (q1)
      (c012) edge (q2);
  \end{tikzpicture}
  \caption{Fifth hypothesis $\hypothesisA^5$.}%
  \label{figure:hypothesis:5}
 \end{subfigure}
 \caption{Observation tables and ARTAs in the example in \cref{section:worked_example}. Shaded rows in the observation tables show a minimum-cardinality monotone basis.}%
 \label{figure:worked_example:1}
\end{figure}
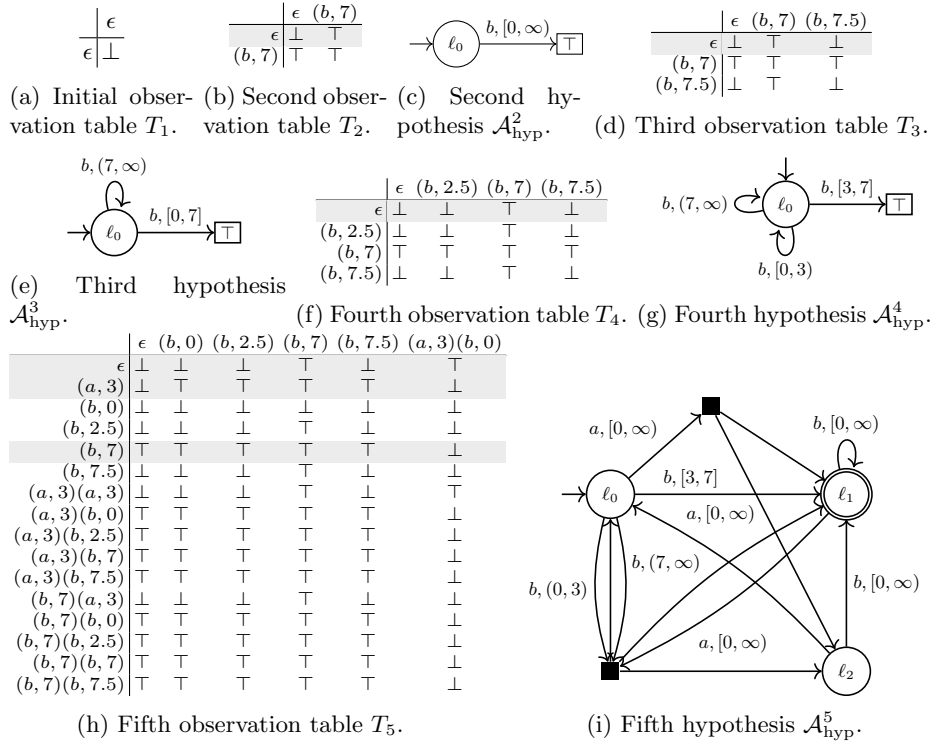

\cref{figure:observation_tables:1} shows the initial observation table $\Table_1$, where
we have $\PrefixSet=\SuffixSet=\{\emptyword\}$ and $\Basis=\emptyset$.
Note that $\TableRow{\emptyword}$ is represented by $\bigvee_{r \in \emptyset} r$.
The corresponding hypothesis ARTA $\hypothesisA^1$ is such that $\Loc = \Final = \emptyset$ and $\InitLoc = \bot$, which does not accept any timed word.
The learner asks an equivalence query, and the teacher returns a counterexample $(b,7)$, which belongs to $\targetLg$ but is rejected by $\hypothesisA^1$.

The learner then adds $(b,7)$ to $\SuffixSet$ and makes the observation table cohesive.
\cref{figure:observation_tables:2,figure:hypothesis:2} show the resulting table $\Table_2$ and the corresponding hypothesis ARTA $\hypothesisA^2$, respectively.
Here, we have $\Basis = \{\emptyword\}$; $\TableRow{(b,7)}$ is represented by $\bigwedge_{r \in \emptyset} r$. 
The learner asks an equivalence query, and the teacher returns a counterexample $(b,7.5)$, which does not belong to $\targetLg$ but is accepted by $\hypothesisA^2$.

The learner then adds $(b,7.5)$ to $\SuffixSet$ and makes the observation table cohesive.
The resulting observation table $\Table_3$ and the corresponding hypothesis ARTA $\hypothesisA^3$ are shown in \cref{figure:observation_tables:3,figure:hypothesis:3}.
We still have $\Basis=\{\emptyword\}$; the row of the new prefix timed word $(b,7.5)$ is the same as that of $\emptyword$.
The learner asks an equivalence query, and the teacher returns a counterexample $(b,2.5)$\LongVersion{, which does not belong to $\targetLg$ but is accepted by $\hypothesisA^3$}.

The learner then adds $(b,2.5)$ to $\SuffixSet$ and makes the observation table cohesive.
The resulting observation table $\Table_4$ and the corresponding hypothesis ARTA $\hypothesisA^4$ are shown in \cref{figure:observation_tables:4,figure:hypothesis:4}.
We again have $\Basis=\{\emptyword\}$; the new timed word $(b,2.5)$ also has the same row as $\emptyword$.
The learner asks an equivalence query, and the teacher returns a counterexample $(a,3)(b,0)$\LongVersion{, which belongs to $\targetLg$ but is rejected by $\hypothesisA^4$}.

The learner then adds the suffixes of $(a,3)(b,0)$ to $\SuffixSet$ and makes the observation table cohesive.
The resulting observation table $\Table_5$ and the corresponding hypothesis ARTA $\hypothesisA^5$ are shown in \cref{figure:observation_tables:5,figure:hypothesis:5}.
Here, the monotone basis changes to $\Basis=\{\emptyword,\Athree,\Bseven\}$, and the corresponding hypothesis has three locations.
The learner asks an equivalence query, and the teacher returns a counterexample $(a,3.5)(b,7.5)$\LongVersion{, which does not belong to $\targetLg$ but is accepted by $\hypothesisA^5$}.

The fifth hypothesis ARTA already has three locations that match the target ARTA.
The learner makes four additional equivalence queries, mainly to refine the transitions rather than to introduce qualitatively new behavior.
After the ninth equivalence query, the teacher returns no counterexample, and the final hypothesis is equivalent to the target ARTA in \cref{figure:running_arta}.

\subsection{Monotone basis identification via BIP}\label{section:exact_basis_milp}

In \ALstarRTA{}, we reduce the identification of a minimum-cardinality monotone basis to a binary integer program (BIP) over the distinct row vectors $\mathrm{Rows}(\PrefixSet)$.
For each row vector $x \in \mathrm{Rows}(\PrefixSet)$, we let $\mathrm{Pos}(x)=\{\suffix \in \SuffixSet \mid x(\suffix)=\top\}$ and $\mathrm{Neg}(x)=\{\suffix \in \SuffixSet \mid x(\suffix)=\bot\}$.
For each triple $(x,\suffix^+,\suffix^-)$ with $x \in \mathrm{Rows}(\PrefixSet)$, $\suffix^+ \in \mathrm{Pos}(x)$, and $\suffix^- \in \mathrm{Neg}(x)$, we identify one \emph{separator obligation}.
A row $r \in \mathrm{Rows}(\PrefixSet)$ \emph{covers} $(x,\suffix^+,\suffix^-)$ if $r(\suffix^+) = \top$ and $r(\suffix^-)=\bot$.
We then introduce a binary variable $y_{r} \in \{0,1\}$ for each $r \in \mathrm{Rows}(\PrefixSet)$ and solve the following BIP:
\begin{align*}
  \min & \sum_{r \in \mathrm{Rows}(\PrefixSet)} y_{r} & \\
  \text{s.t.} &
  \sum_{\substack{r \in \mathrm{Rows}(\PrefixSet), r(\suffix^+) = \top, r(\suffix^-) = \bot}} y_{r} \geq 1
  && \forall x \in \mathrm{Rows}(\PrefixSet), \suffix^+ \in \mathrm{Pos}(x), \suffix^- \in \mathrm{Neg}(x). %
\end{align*}
Let $U = \{r \in \mathrm{Rows}(\PrefixSet) \mid y_{r}=1\}$ be the selected row set.
We map each $r \in U$ back to $\prefix_{r} \in \PrefixSet$ satisfying $\TableRow{\prefix_{r}}= r$ and set $\Basis = \{\prefix_{r} \mid r \in U\}$.

\newcommand{\MILPBasisEncodingStatement}{%
 For $R \subseteq \{\top,\bot\}^{\SuffixSet}$ and $U \subseteq R$,
 $U$ is a monotone basis of $R$ if and only if for any $x \in R$, $\suffix^+ \in \mathrm{Pos}(x)$, and $\suffix^- \in \mathrm{Neg}(x)$, there is $r \in U$ covering $(x,\suffix^+,\suffix^-)$.
}
\begin{proposition}
 \label{proposition:milp_basis_encoding}
 \MILPBasisEncodingStatement{}
 \qed{}
\end{proposition}

By \cref{proposition:milp_basis_encoding}, the feasible solutions of the BIP are exactly the monotone bases of $R$, and solving it yields a minimum-cardinality monotone basis.
In our implementation, this optimization is solved \emph{approximately} using a relative gap tolerance and a time limit, for the sake of efficiency.

\subsection{Query complexity of \ALstarRTA{}}

Finally, we show the query complexity of \ALstarRTA{} using the following bound on the number of distinct rows in observation tables.
In what follows, we assume that the BIP from \cref{section:exact_basis_milp} is exactly solved, and we always obtain a minimum-cardinality monotone basis $\Basis$.
We also assume that the mapping from $\mathrm{Rows}(\PrefixSet)$ to $\PrefixSet$ is deterministic; once $\prefix \in \PrefixSet$ for a row vector $r \in \{\top, \bot\}^\SuffixSet$ is selected, the same $\prefix$ is always chosen for $r$ until $\SuffixSet$ changes.

\newcommand{\rowBoundStatement}{%
 Let $n$ be the number of locations of the minimal DRTA recognizing $\targetLg$.
 Every observation table $\Table$ has at most $n$ distinct row vectors\LongVersion{, \ie{} $|\mathrm{Rows}(\PrefixSet)| \leq n$ holds}.
}
\begin{lemma}%
 \label{lemma:row_bound}
 \rowBoundStatement{}
 \qed{}
\end{lemma}
\newcommand{\complexityAnalysisStatement}{%
 Let $\targetLg \subseteq \TimedWords$ be a real-time language,
 let $n$ be the number of locations of the minimal DRTA recognizing $\targetLg$, and
 let $h$ be the maximum length of counterexamples returned by the teacher.
 Fix $K\in\setN$ as in \cref{theorem:myhill_nerode} for $\targetLg$.
 \cref{algorithm:ALstarRTA} terminates and returns an ARTA recognizing $\targetLg$ using at most
 $M + 1$ equivalence queries and
 $M^3h^2n + M^2hn + Mh + 1$ membership queries, where $M = 2^n + |\Alphabet| (2K + 2)$.
}
\begin{theorem}[query complexity]%
 \label{theorem:complexity_analysis}
 \complexityAnalysisStatement{}
\end{theorem}
\begin{proof}[sketch]
 \textbf{Number of equivalence queries.}
 For each suffix $\suffix \in \SuffixSet$, we let $t(\suffix) \in \{\bot,\top\}^n$ as follows.
 Let $\word_1, \word_2, \dots, \word_n \in \TimedWords$ such that for each $i \neq j$, $\word_i \not\nerode{}\word_j$.
 Then, we let $t(\suffix)(i) = \top \iff \word_i \cdot \suffix \in \targetLg$.
 Since there are at most $2^n$ distinct $t(\suffix)$, discovery of a new $t(\suffix)$ occurs at most $2^n$ times in \cref{algorithm:ALstarRTA}.

 Assume the counterexample $\cex$ returned by an equivalence query does not reveal a new $t(\suffix)$.
 Still, processing $\normalize(\cex)$ must change the observation table so that $\normalize(cex)$ is not a counterexample to the next hypothesis ARTA because
 we have $\normalize(\cex)$ as a suffix, and thus, the next hypothesis must classify $\normalize(\cex)$ correctly by \cref{theorem:faithfulness,theorem:normalized_counterexample}.
 In such a case, we must discover a new $(\action, \delay)$ in $\letters{\PrefixSet \cup \SuffixSet}$ such that either $\delay \leq K$ or $\delay > K$ and there is no $(\action, \delay') \in \letters{\PrefixSet \cup \SuffixSet}$ with $\delay' > K$; otherwise, $\cex$ may still be a counterexample to the hypothesis ARTA after processing it.
 The number of such refinements of $\letters{\PrefixSet \cup \SuffixSet}$ by equivalence queries is at most $|\Alphabet| (2K + 2)$.
 Overall, the number of equivalence queries is bounded by $2^n + |\Alphabet| (2K + 2) + 1$.

 \noindent
 \textbf{Number of membership queries.}
 After each failed equivalence query, at most $h$ new suffixes are added.
 Thus, we have $|\SuffixSet| \leq M h + 1$.
 Moreover, each failed equivalence query increases $\letters{\PrefixSet \cup \SuffixSet}$ at most by $h$, and we have
 $|\letters{\PrefixSet \cup \SuffixSet}| \leq M h$.
 Since $\Basis$ is a monotone basis of $\Table$, $|\Basis| = |\mathrm{Rows}(\Basis)| \leq |\mathrm{Rows}(\PrefixSet)| \leq n$ follows from \cref{lemma:row_bound}.
 $\PrefixSet$ increases only to ensure evidence-closedness at \cref{algorithm:ALstarRTA:evidence_closed}.
 By \cref{lemma:row_bound}, there are at most $n$ distinct row vectors in the observation tables appearing in this computation.
 By determinism of the mapping from $\mathrm{Rows}(\PrefixSet)$ to $\PrefixSet$ in the monotone basis identification,
 there are at most $n$ distinct $\prefix \in \PrefixSet$ in the monotone bases between each equivalence query.
 Thus, between each equivalence query, 
 $\PrefixSet$ increases at most by $|\letters{\PrefixSet \cup \SuffixSet}| \times n \leq M h n$, and we have $|\PrefixSet| \leq M^2 hn + 1$.
 Overall, the number of membership queries is bounded by
 $|\PrefixSet| \times |\SuffixSet| \leq (M^2hn + 1) \times (Mh + 1) = M^3h^2n + M^2hn + Mh + 1$.
 \qed{}
\end{proof}

In \cref{theorem:complexity_analysis}, we have an exponential query bound for \ALstarRTA{}, while it is polynomial for \NLstarRTA{}.
Nevertheless, in \cref{section:experiments}, we empirically observe that the difference in the number of equivalence queries is moderate in practice.

\section{Empirical evaluation}\label{section:experiments}

We empirically compare our \ALstarRTA{} algorithm with the \NLstarRTA{} algorithm.
We use our prototype library \ourTool{}\footnote{\ourTool{} is publicly available at \url{https://github.com/MasWag/LearnARTA}. The artifact is available at \url{https://doi.org/10.5281/zenodo.19650471}.} as an implementation of \ALstarRTA{} and \textsc{NRTALearning}~\cite{Leslieaj/NRTALearning} as an implementation of \NLstarRTA{}.
In \ourTool{}, monotone basis identification is implemented by approximately solving the BIP from \cref{section:exact_basis_milp}.
In both implementations, each equivalence query is answered by an exhaustive analysis, such as a reachability analysis.

We used the 190 randomly generated NRTAs taken from~\cite{DBLP:journals/tecs/AnZZZ21}.
The benchmark consists of 17 groups, each with different numbers of locations and different alphabet sizes.
The constant appearing in the guards is at most 20.

For each execution, we measured 
\begin{myitemize}
 \item the number of membership and equivalence queries and
 \item the total execution time, including the time to answer the queries.
\end{myitemize}
For membership queries, we count the number of unique queried timed words.
We conducted all the experiments on a computing server running Ubuntu 24.04 LTS with an Intel Xeon w5-3435X and 125 GiB of RAM.
\begin{table}[t]
  \centering
  \small
  \caption{Summary of the results of experiments. We show the mean value for each group. The columns ``\# EqQ'' and ``\# MemQ'' show the number of equivalence and membership queries, respectively. The columns ``$|\Loc_{\hypothesisA}|$'' show the number of locations of the learned automaton.\@ The columns ``Total Time'' show the total execution time including the time to answer queries.}
  \label{table:summary_results}
  \begin{tabular}{lrrrrrrrr}
    \toprule
    \multirow{2}{*}{$(|\Loc|,|\Alphabet|)$} & \multicolumn{4}{c}{\ALstarRTA{}} & \multicolumn{4}{c}{\NLstarRTA{}} \\
    \cmidrule(lr){2-5}\cmidrule(lr){6-9}
 & \# EqQ & \# MemQ & $|\Loc_{\hypothesisA}|$ & Total Time & \# EqQ & \# MemQ & $|\Loc_{\hypothesisA}|$ & Total Time \\
\midrule
(3,2) & 13.1 & 975.05 & \tbcolor{}3 & \tbcolor{}0:00.01 & \tbcolor{}10.9 & \tbcolor{}303.75 & 4 & 0:00.07 \\
(4,2) & 17.15 & 2449.9 & \tbcolor{}4.05 & \tbcolor{}0:00.02 & \tbcolor{}15.65 & \tbcolor{}810.4 & 5.05 & 0:00.12 \\
(5,2) & 22.6 & 6280.5 & \tbcolor{}5 & \tbcolor{}0:00.06 & \tbcolor{}20.7 & \tbcolor{}1279.1 & 6 & 0:00.17 \\
(6,2) & \tbcolor{}22.3 & 10440.1 & \tbcolor{}6.1 & \tbcolor{}0:00.18 & 22.4 & \tbcolor{}1737.5 & 7.2 & 0:00.31 \\
(8,2) & 31.3 & 28058.3 & \tbcolor{}8 & 0:00.77 & \tbcolor{}27.6 & \tbcolor{}3093.6 & 9.1 & \tbcolor{}0:00.45 \\
(8,4) & 48.4 & 55362.2 & \tbcolor{}8 & \tbcolor{}0:00.66 & \tbcolor{}46.4 & \tbcolor{}8449.5 & 9.1 & 0:02.05 \\
(10,2) & 38.2 & 50558.9 & \tbcolor{}10 & 0:01.61 & \tbcolor{}36.4 & \tbcolor{}5665.3 & 11.1 & \tbcolor{}0:01.02 \\
(10,4) & 52.4 & 106987.2 & \tbcolor{}10 & \tbcolor{}0:02.28 & \tbcolor{}49.2 & \tbcolor{}9664.4 & 11.1 & 0:02.42 \\
(10,6) & 71.7 & 171537.9 & \tbcolor{}10 & \tbcolor{}0:02.92 & \tbcolor{}62.8 & \tbcolor{}17346.3 & 11 & 0:06.16 \\
(10,8) & 97 & 443801.3 & \tbcolor{}10 & 0:20.75 & \tbcolor{}83.1 & \tbcolor{}20298.7 & 11 & \tbcolor{}0:07.55 \\
(10,10) & 114.2 & 723024.4 & \tbcolor{}10 & 0:33.37 & \tbcolor{}92 & \tbcolor{}26558.5 & 11 & \tbcolor{}0:13.48 \\
(12,2) & 56 & 144836.1 & \tbcolor{}12.3 & 0:08.06 & \tbcolor{}48.4 & \tbcolor{}9245.6 & 13.3 & \tbcolor{}0:02.13 \\
(12,4) & 70 & 295512.5 & \tbcolor{}12 & 0:06.59 & \tbcolor{}66 & \tbcolor{}14968.3 & 13 & \tbcolor{}0:04.39 \\
(14,4) & 81.7 & 590253.4 & \tbcolor{}14.1 & 0:40.19 & \tbcolor{}74 & \tbcolor{}20697.2 & 15.1 & \tbcolor{}0:07.69 \\
(16,4) & 91.4 & 812425.1 & \tbcolor{}16 & 1:39.27 & \tbcolor{}85.5 & \tbcolor{}41004.4 & 17.1 & \tbcolor{}0:28.44 \\
(18,4) & 100.7 & 1136254 & \tbcolor{}18.3 & 2:08.25 & \tbcolor{}97.5 & \tbcolor{}35404 & 19.4 & \tbcolor{}0:18.07 \\
(20,4) & 105.7 & 1122431.9 & \tbcolor{}19.8 & 1:01.14 & \tbcolor{}103.9 & \tbcolor{}56806.9 & 21 & \tbcolor{}0:46.75 \\
\bottomrule
\end{tabular}
\end{table}

\paragraph{Results and discussion.}
\cref{table:summary_results} summarizes the results.
\LongVersion{In \cref{table:summary_results}, we}\ShortVersion{We} observe that \ALstarRTA{} returned a smaller automaton than \NLstarRTA{}.
This is consistent with the better worst-case succinctness of ARTAs (\cref{theorem:arta_drta_succinctness}) than NRTAs~\cite{DBLP:journals/tecs/AnZZZ21}.

In contrast, \ALstarRTA{} almost always used more queries than \NLstarRTA{}.
This is consistent with the theoretical analysis in \cref{theorem:complexity_analysis}.
The increase in the number of equivalence queries is also reported in untimed cases~\cite{DBLP:conf/ijcai/AngluinEF15}, where \ALstar{} required more equivalence queries than \NLstar{} and \Lstar{}.
This increase is particularly evident when the alphabet is large.
This is likely because in \ALstarRTA{}, the set of prefixes is initialized with $\{\emptyword\}$ rather than $\{\emptyword\} \cup \{(a, 0)\mid a \in \Alphabet\}$, which prevents irrelevant prefixes but often requires more equivalence queries.

The increase in the number of membership queries is partly due to the different counterexample handling strategies:
\NLstarRTA{} uses an efficient counterexample handling based on Rivest and Schapire's algorithm~\cite{DBLP:journals/iandc/RivestS93}, which reduces the size of observation tables and reduces the number of membership queries.
In contrast, \ALstarRTA{} uses a simple counterexample handling, much like \NLstar{}.

\LongVersion{We also observe that }\ALstarRTA{} often took longer time than \NLstarRTA{}.
This is likely due to the difficulty of basis minimization in \ALstarRTA{}, in addition to the larger number of queries.

\section{Related work}\label{section:related_work}
\paragraph{Learning nondeterministic or alternating automata.}

\NLstar{}~\cite{DBLP:conf/ijcai/BolligHKL09} is an \Lstar{}-style algorithm for learning residual NFAs, which is a subclass of NFAs with some desired properties of DFAs.
Thanks to the succinctness of residual NFAs, \NLstar{} can learn smaller automata than \Lstar{}.
\ALstar{}~\cite{DBLP:conf/ijcai/AngluinEF15} is an extension of \NLstar{} to learn an AFA rather than an NFA.
Although the authors conjectured that \ALstar{} produces residual AFAs, it is later shown that \ALstar{} does not always produce a residual AFA~\cite{DBLP:journals/iandc/BerndtLLR22}, and \ALdstar{}, a modification of \ALstar{} to learn residual AFAs, is proposed.
Our \ALstarRTA{} is mainly based on \ALstar{} rather than \ALdstar{} because we aim at identifying smaller ARTAs that are not necessarily residual.

\paragraph{Learning timed languages.}

\Lstar{}-style learning algorithms have been proposed for various variants of timed automata, including
one-clock deterministic timed automata~\cite{DBLP:conf/tacas/AnCZZZ20,DBLP:conf/atva/XuAZ22},
(multi-clock) deterministic timed automata~\cite{DBLP:conf/cav/Waga23,DBLP:conf/hybrid/TengZ024,DBLP:journals/chinaf/TengCMZAZ25},
Mealy machines with timers~\cite{DBLP:conf/icfem/KogelKG23,DBLP:conf/qestformats/KogelSG24,DBLP:conf/qestformats/BruyereGPSV25},
event-recording automata~\cite{DBLP:journals/tcs/GrinchteinJL10}, and
RTAs~\cite{DBLP:journals/chinaf/AnWZZZ21,DBLP:journals/tecs/AnZZZ21}.
As we discussed in \cref{section:introduction}, most of these algorithms focus on learning deterministic automata.
Extending these algorithms to learn automata other than RTAs with nondeterministic or alternating branching is a future direction.

\paragraph{Learning symbolic automata.}

Symbolic automata~\cite{DBLP:conf/popl/DAntoniV14} generalize classical automata so that each transition is labeled with a predicate rather than a letter.
Notably, symbolic automata capture real-time automata.
So far, no algorithm has been proposed for learning nondeterministic or alternating symbolic automata.
Extending learning algorithms for symbolic automata, such as $\Lambda^*$~\cite{DBLP:conf/tacas/DrewsD17} and $\mathit{MAT}^*$~\cite{DBLP:conf/cav/ArgyrosD18}, or generalizing the theoretical analysis in~\cite{DBLP:journals/lmcs/FismanFZ23} to support nondeterministic or alternating branching is another future direction.

\section{Conclusions and future directions}\label{section:conclusions}

We studied ARTAs, especially focusing on an \Lstar{}-style learning algorithm for them.
Although the introduction of alternating branching does not increase expressive power, it improves succinctness.
Our \ALstarRTA{} algorithm can learn ARTAs with a termination guarantee.
Our empirical evaluation suggests that \ALstarRTA{} generally learns smaller automata than \NLstarRTA{} at the cost of more queries.

Investigating efficient counterexample handling for ARTAs, particularly along the lines of Rivest and Schapire, is a future direction.
Generalizing alternating automata learning for other subclasses of timed automata, \eg{} event-recording automata, is another future direction.

\subsubsection*{Acknowledgements.}

This work is partially supported
    by
    JST BOOST Grant No.\ JPMJBY24H8,
    JST PRESTO Grant No.\ JPMJPR22CA, and
    JSPS KAKENHI Grant No.\ 22K17873.

\clearpage

\bibliographystyle{splncs04}
\bibliography{references}

\appendix
\LongVersion{%

\section{Omitted proofs}
\subsection{Proof of \cref{theorem:expressive_power}}

\recallResult{theorem:expressive_power}{\expressivePowerStatement}

\begin{proof}
The first direction is immediate because a DRTA is a special case of an ARTA, where all transition targets are atomic locations.

For the converse, let $\A=(\Loc,\Alphabet,\InitLoc,\Final,\Edge)$ be an ARTA and let $K \in \setN$ be the maximum integer constant that appears in the guards of $\Edge$.
Let $n = |\Loc|$.
Let $\mathrm{Reg}$ be the set of regions for $K$.
Let $\A_{\mathrm{AFA}}$ be the AFA $\A_{\mathrm{AFA}} = (\Loc, (\Alphabet \times \mathrm{Reg}), \InitLoc, \Final, \transition_{\mathrm{AFA}})$ with a transition function $\transition_{\mathrm{AFA}}\colon \Loc \times (\Alphabet\times\mathrm{Reg}) \to \Boolean{\Loc}$ defined as
\[
  \transition_{\mathrm{AFA}}(\loc,(\action,\rho)) =
  \begin{cases}
    \formula & \text{if there is }(\loc,\action,I,\formula)\in\Edge\text{ with }\rho \subseteq I,\\
    \bot & \text{otherwise.}
  \end{cases}
\]
This is well-defined because for any interval $I$ with integer endpoints and any $\delay_1,\delay_2\in\Rnn$ with $\llbracket\delay_1\rrbracket=\llbracket\delay_2\rrbracket$, we have $\delay_1\in I \iff \delay_2\in I$.
Moreover, the determinism condition on $\Edge$ ensures that the target formula is unique.
For any timed word $\word=(\action_1,\delay_1), (\action_2, \delay_2), \dots, (\action_m,\delay_m)$, 
we have $\word \in \Lg(\A)\iff \mathrm{reg}(\word)\in\Lg(\A_{\mathrm{AFA}})$, where 
$\mathrm{reg}(\word)=(\action_1,\llbracket\delay_1\rrbracket), (\action_2, \llbracket\delay_2\rrbracket),\dots,(\action_m,\llbracket\delay_m\rrbracket)$.

For any AFA with $n$ locations, there is a DFA $\mathcal{D}$ with at most $2^{2^{n}}$ locations satisfying $\Lg(\mathcal{D})=\Lg(\A_{\mathrm{AFA}})$ (\eg{} \cite[Theorem~5.2]{DBLP:journals/jacm/ChandraKS81}).
We convert $\mathcal{D}$ into a DRTA $\TA_{\mathcal{D}}$ by replacing each transition $q \xrightarrow{(\action,\rho)} q'$ over $\Alphabet \times \mathrm{Reg}$ with a transition $q \xrightarrow{\action, \rho} q'$.
Because regions $\rho\in\mathrm{Reg}$ are pairwise disjoint, $\TA_{\mathcal{D}}$ is deterministic.
Thus, $\TA_{\mathcal{D}}$ has at most $2^{2^{n}}$ locations.
By construction, for any timed word $\word$, we have
\[
  \word\in\Lg(\TA_{\mathcal{D}})
  \iff \mathrm{reg}(\word)\in\Lg(\mathcal{D})
  \iff \mathrm{reg}(\word)\in\Lg(\A_{\mathrm{AFA}})
  \iff \word\in\Lg(\A).
\]
\qed{}
\end{proof}
\subsection{Proof of \cref{theorem:arta_drta_succinctness}}

\recallResult{theorem:arta_drta_succinctness}{\succinctnessStatement}

\begin{proof}
By the classical succinctness results for AFAs~\cite{DBLP:journals/tcs/Leiss81,DBLP:journals/tcs/Leiss85a}, there exist a constant $c_0 > 0$ and a family of AFAs $(\mathcal{B}_k)_{k \geq 1}$ with at most $k$ locations such that for any sufficiently large $k$,
any DFA recognizing $\Lg(\mathcal{B}_k)$ has at least $2^{2^{c_0 k}}$ locations.
Write $\mathcal{B}_k = (\Loc_k,\Alphabet,\InitLoc^{(k)},\Final_k,\transition_k)$, where $\transition_k\colon \Loc_k\times\Alphabet\to\Boolean{\Loc_k}$.

For each $k$, we define an ARTA $\TA_k=(\Loc_k,\Alphabet,\InitLoc^{(k)},\Final_k,\Edge_k)$ with
$\Edge_k = \{ (\loc,\action,[0,\infty),\transition_k(\loc,\action)) \mid \loc\in\Loc_k, \action \in \Alphabet\}$.
Such ARTAs satisfy $|\Loc_k| \leq k$ and all guards are $[0,\infty)$.
By construction, for any $\loc\in\Loc_k$, $\action\in\Alphabet$, and $\delay\in\Rnn$, we have
$\transition_{\TA_k}(\loc,(\action,\delay))=\transition_k(\loc,\action)$, where $\transition_{\TA_{k}}$ is the transition function induced by $\Edge_k$.
For a word $\word = \action_1, \action_2, \cdots, \action_m \in \Alphabet^*$, we let
$\tilde{\word} = (\action_1, 0), (\action_2, 0), \dots, (\action_m,0) \in \TimedWords$.
By induction on $|\word|$, for any $\word\in\Alphabet^*$, we have $\word\in\Lg(\mathcal{B}_k) \iff \tilde{\word}\in\Lg(\TA_k)$.

Let $\mathcal{D}_k=(\Loc^{\mathcal{D}}_k,\Alphabet,\loc^{\mathcal{D}}_{0,k}, \Final^{\mathcal{D}}_{k},\Edge^{\mathcal{D}}_{k})$ be a DRTA that satisfies $\Lg(\mathcal{D}_k)=\Lg(\TA_k)$.
We let $\mathcal{M}_k$ be a DFA $\mathcal{M}_k = (\Loc^{\mathcal{D}}_k \cup \{\loc_{\mathit{sink}}\}, \Alphabet, \loc^{\mathcal{D}}_{0,k}, \Final^{\mathcal{D}}_k, \transition^{\mathcal{D}}_{k})$, where
\[
  \transition^{\mathcal{D}}_{k}(\loc,\action) =
  \begin{cases}
    \loc' & \text{if $(\loc,\action,I,\loc')\in\Edge^{\mathcal{D}}_{k}$ and $0\in I$ for some $I\in\Intervals$,}\\
    \loc_{\mathit{sink}} & \text{otherwise,}
  \end{cases}
\]
and $\transition^{\mathcal{D}}_{k}(\loc_{\mathit{sink}},\action) = \loc_{\mathit{sink}}$.
$\transition^{\mathcal{D}}_{k}(\loc,\action)$ is well-defined because $\mathcal{D}_k$ is deterministic.
By induction on $|\word|$, for any $\word\in\Alphabet^*$, we have
$\word\in\Lg(\mathcal{M}_{k}) \iff \tilde{\word} \in \Lg(\mathcal{D}_k)$.

Overall, we have $\Lg(\mathcal{M}_{k})=\Lg(\mathcal{B}_k)$.
Therefore, $|\Loc^{\mathcal{D}}_k| + 1 \geq 2^{2^{c_0 k}}$ for any sufficiently large $k$.
Then, for any constant $c < c_0$, for any sufficiently large $k$, we have $|\Loc^{\mathcal{D}}_k|\ge 2^{2^{c k}}$.
\qed{}
\end{proof}
\subsection{Proof of \cref{theorem:faithfulness}}

Before proving \cref{theorem:faithfulness}, we show some lemmas.

\begin{lemma}%
 \label{lemma:basis_row_correctness}
 Assume $\Table$ with index sets $\PrefixSet, \SuffixSet \subseteq \TimedWords$ and monotone basis $\Basis\subseteq \PrefixSet$ is cohesive.
 Let $\evidenceA$ be the evidence AFA of $\Table$.
 For any $\basis\in\Basis$ and
 for any $\suffix\in\SuffixSet$, we have
 $r(\suffix)=\top \iff E_{\evidenceA}(\transition_{\evidenceA}^*(r,\suffix))=\top$,
 where $r=\TableRow{\basis}\in\RowsP$.
\end{lemma}
\begin{proof}
 We prove by induction on the length of $\suffix$.

 When $\suffix=\emptyword$,
 from the definition of accepting locations of $\evidenceA$, $r$ is accepting if and only if $r(\emptyword)=\top$.
 Thus, $E_{\evidenceA}(\transition_{\evidenceA}^*(r,\emptyword))=\top \iff r(\emptyword)=\top$ holds.

 Assume $\suffix \neq \emptyword$ and
 fix $\suffix=\sigma\cdot \suffix'$ with $\sigma\in\PhiTable$ and $\suffix'\in\SuffixSet$.
 Since $\Table$ is evidence-closed, we have $\basis\cdot \sigma\in\PrefixSet$.
 Thus, the row vector $\TableRow{\basis\cdot\sigma}$ is defined.
 Moreover, $\suffix'\in\SuffixSet$ holds by suffix-closedness of $\SuffixSet$.
 We have $r(\sigma\cdot \suffix') = \TableRow{\basis}(\sigma\cdot \suffix') = \TableCell{\basis}{\sigma\cdot \suffix'} = \TableCell{\basis\cdot \sigma}{\suffix'} = \TableRow{\basis\cdot\sigma}(\suffix')$.
 By \cref{definition:evidence_afa}, the one-step transition on $\sigma$ satisfies $\transition_{\evidenceA}(r,\sigma)=b_{\RowsP}(\TableRow{\basis\cdot\sigma})$, where
 \begin{equation}
  b_{\RowsP}(x)=\bigvee_{\substack{\suffix''\in\SuffixSet\\x(\suffix'')=\top}} M_{\RowsP}(\suffix'')
   \qquad\text{and}\qquad
   M_{\RowsP}(\suffix'')=\bigwedge_{\substack{\rho\in\RowsP\\\rho(\suffix'')=\top}}\rho.\label{eq:lemma:basis_row_correctness_1}
 \end{equation}

 \emph{($\Rightarrow$).}
 Assume $r(\sigma\cdot \suffix')=\top$.
 Then, we have $\TableRow{\basis\cdot\sigma}(\suffix')=\top$, and
 by \cref{eq:lemma:basis_row_correctness_1},
 the conjunction $M_{\RowsP}(\suffix')$ appears as a disjunct of $b_{\RowsP}(\TableRow{\basis\cdot\sigma})$, and hence as a disjunct of $\transition_{\evidenceA}(r,\sigma)$.
 Therefore, $E_{\evidenceA}(\transition_{\evidenceA}^*(M_{\RowsP}(\suffix'), \suffix')) = \top$ implies $E_{\evidenceA}(\transition_{\evidenceA}^*(\transition_{\evidenceA}(r,\sigma), \suffix')) = \top$.
 Every conjunct $\rho$ of $M_{\RowsP}(\suffix')$ satisfies $\rho(\suffix')=\top$.
 By the induction hypothesis, for any such $\rho$, we have $E_{\evidenceA}(\transition_{\evidenceA}^*(\rho,\suffix'))=\top$.
 This implies 
 \[
  E_{\evidenceA}(\transition_{\evidenceA}^*(r,\suffix)) = E_{\evidenceA}(\transition_{\evidenceA}^*(r,\sigma \cdot \suffix')) = E_{\evidenceA}(\transition_{\evidenceA}^*(\transition_{\evidenceA}(r,\sigma), \suffix')) = \top.
 \]

 \emph{($\Leftarrow$).}
 We prove the contrapositive.
 Assume $r(\sigma\cdot \suffix')=\bot$.
 Then, we have $\TableRow{\basis\cdot\sigma}(\suffix')=\bot$.
 Let $x= \TableRow{\basis\cdot\sigma}\in\{\top,\bot\}^{\SuffixSet}$.
 Since $\basis\cdot\sigma\in\PrefixSet$ and $\Table$ is $\Basis$-closed, we have $x\in\Boolean{\RowsP}$. 
 Consider any disjunct $M_{\RowsP}(\suffix'')$ of $b_{\RowsP}(x)$.
 By \cref{eq:lemma:basis_row_correctness_1}, we have $x(\suffix'')=\top$.
 Since $x(\suffix')=\bot$ and $x\in\Boolean{\RowsP}$, there is $\rho\in\RowsP$ satisfying $\rho(\suffix'')=\top$ and $\rho(\suffix')=\bot$.
 By the induction hypothesis, we have $E_{\evidenceA}(\transition_{\evidenceA}^*(\rho,\suffix'))=\bot$.
 Since $\rho$ is a conjunct of $M_{\RowsP}(\suffix'')$, we have $E_{\evidenceA}(\transition_{\evidenceA}^*(M_{\RowsP}(\suffix''),\suffix'))=\bot$.
 As the argument holds for every disjunct of $b_{\RowsP}(x)$, we have 
 $E_{\evidenceA}(\transition_{\evidenceA}^*(\transition_{\evidenceA}(r, \sigma),\suffix'))=\bot$.
 \qed{}
\end{proof}
\begin{lemma}%
 \label{lemma:decomposition_correctness}
 Assume $\Table$ with index sets $\PrefixSet, \SuffixSet \subseteq \TimedWords$ and monotone basis $\Basis\subseteq \PrefixSet$ is cohesive.
 Let $\evidenceA$ be the evidence AFA of $\Table$.
 Let $x\in\{\top,\bot\}^{\SuffixSet}$ be a row vector satisfying $x\in\Boolean{\RowsP}$.
 For any $\suffix\in\SuffixSet$, $x(\suffix)=\top \iff E_{\evidenceA}(\transition_{\evidenceA}^*(b_{\RowsP}(x),\suffix))=\top$ holds.
\end{lemma}
\begin{proof}
 Fix $\suffix\in\SuffixSet$.
 From the definition of $b_{\RowsP}$, we have
 \begin{displaymath}
  b_{\RowsP}(x)=\bigvee_{\substack{\suffix'\in\SuffixSet\\x(\suffix')=\top}} M_{\RowsP}(\suffix')
   \qquad\text{and}\qquad
   M_{\RowsP}(\suffix')=\bigwedge_{\substack{\rho\in\RowsP\\\rho(\suffix')=\top}}\rho.
 \end{displaymath}
 If $x(\suffix)=\top$, $M_{\RowsP}(\suffix)$ is one of the disjuncts of $b_{\RowsP}(x)$.
 Every conjunct $\rho$ of $M_{\RowsP}(\suffix)$ satisfies $\rho(\suffix)=\top$.
 By \cref{lemma:basis_row_correctness}, we have $E_{\evidenceA}(\transition_{\evidenceA}^*(\rho,\suffix))=\top$ for any such $\rho$.
 Thus, $E_{\evidenceA}(\transition_{\evidenceA}^*(b_{\RowsP}(x),\suffix)) = E_{\evidenceA}(\transition_{\evidenceA}^*(M_{\RowsP}(\suffix),\suffix)) =\top$ holds.

 For the converse, we prove the contrapositive.
 Assume $x(\suffix)=\bot$.
 Consider any disjunct $M_{\RowsP}(\suffix')$ of $b_{\RowsP}(x)$ with $x(\suffix')=\top$.
 Since $x\in\Boolean{\RowsP}$, there is a $\rho\in\RowsP$ with $\rho(\suffix')=\top$ but $\rho(\suffix)=\bot$.
 This is because if $\rho'(\suffix')=\top$ implies $\rho'(\suffix)=\top$ for any $\rho' \in \RowsP$,
 $x'(\suffix')=\top$ also implies $x'(\suffix)=\top$ for any $x' \in \Boolean{\RowsP}$
 due to the monotonicity of positive Boolean combinations.
 By \cref{lemma:basis_row_correctness}, we have $E_{\evidenceA}(\transition_{\evidenceA}^*(\rho,\suffix))=\bot$.
 Since $\rho$ is a conjunct of $M_{\RowsP}(\suffix')$, we have $E_{\evidenceA}(\transition_{\evidenceA}^*(M_{\RowsP}(\suffix'),\suffix))=\bot$.
 As this holds for every disjunct of $b_{\RowsP}(x)$, we have $E_{\evidenceA}(\transition_{\evidenceA}^*(b_{\RowsP}(x),\suffix))=\bot$.
 \qed{}
\end{proof}
\begin{lemma}%
 \label{lemma:prefix_row_invariant}
 Assume $\Table$ with index sets $\PrefixSet, \SuffixSet \subseteq \TimedWords$ and monotone basis $\Basis\subseteq \PrefixSet$ is cohesive.
 Let $\evidenceA$ be the evidence AFA of $\Table$.
 For any $\prefix\in\PrefixSet$ and for any $\suffix\in\SuffixSet$, we have
 $E_{\evidenceA}(\transition_{\evidenceA}^*(\transition_{\evidenceA}^*(b_{\RowsP}(\TableRow{\emptyword}),\prefix),\suffix)) = E_{\evidenceA}(\transition_{\evidenceA}^*(b_{\RowsP}(\TableRow{\prefix}),\suffix))$.
\end{lemma}
\begin{proof}
 For $\word\in\PhiTable^*$ and $r\in\RowsP$, let $\nu_{\word}(r)=E_{\evidenceA}(\transition_{\evidenceA}^*(r,\word))$.
 Since $\transition_{\evidenceA}^*(\cdot,\word)$ distributes over $\land$ and $\lor$ and $E_{\evidenceA}$ is homomorphic, for any $\formula\in\Boolean{\RowsP}$, the value $E_{\evidenceA}(\transition_{\evidenceA}^*(\formula,\word))$ is obtained by propositional evaluation of $\formula$ under $\nu_{\word}$.
 We prove the claim by induction on the length of $\prefix\in\PrefixSet$.

 When $\prefix=\emptyword$ holds,
 we have $\transition_{\evidenceA}^*(b_{\RowsP}(\TableRow{\emptyword}),\emptyword)=b_{\RowsP}(\TableRow{\emptyword})$ from the definition of $\transition_{\evidenceA}^*$.

 Let $\prefix=\prefix'\cdot\sigma$ with $\sigma\in\PhiTable$.
 For any $\suffix\in\SuffixSet$, we have the following:
 \begin{align*}
   & E_{\evidenceA}(\transition_{\evidenceA}^*(\transition_{\evidenceA}^*(b_{\RowsP}(\TableRow{\emptyword}),\prefix'\cdot\sigma),\suffix))\\
   =& E_{\evidenceA}(\transition_{\evidenceA}^*(\transition_{\evidenceA}^*(b_{\RowsP}(\TableRow{\emptyword}),\prefix'),\sigma\cdot\suffix)) & \text{(definition of $\transition^*$)}\\
   =& E_{\evidenceA}(\transition_{\evidenceA}^*(b_{\RowsP}(\TableRow{\prefix'}),\sigma\cdot\suffix)). & \text{(induction hypothesis)}
 \end{align*}

 For any $r=\TableRow{\basis}\in\RowsP$, we have
 \begin{align*}
    \transition_{\evidenceA}^*(r,\sigma\cdot\suffix)
   =& \transition_{\evidenceA}^*(\transition_{\evidenceA}(r,\sigma),\suffix) & \text{(definition of $\transition^*$)}\\
   =& \transition_{\evidenceA}^*(b_{\RowsP}(\TableRow{\basis\cdot\sigma}),\suffix). & \text{(\cref{definition:evidence_afa})}
 \end{align*}
 Since $\basis\cdot\sigma\in\PrefixSet$ holds by cohesiveness and $\Table$ is $\Basis$-closed, we have $\TableRow{\basis\cdot\sigma}\in\Boolean{\RowsP}$.
 Thus, we have
 $E_{\evidenceA}(\transition_{\evidenceA}^*(r,\sigma\cdot\suffix))=\TableRow{\basis\cdot\sigma}(\suffix)$
 from \cref{lemma:decomposition_correctness}.
 Since $E_{\evidenceA}(\transition_{\evidenceA}^*(\cdot,\sigma\cdot\suffix))$ respects propositional equivalence,
 we have
 \begin{align*}
  & E_{\evidenceA}(\transition_{\evidenceA}^*(b_{\RowsP}(\TableRow{\prefix'}),\sigma\cdot\suffix))\\
  = & \TableRow{\prefix'\cdot\sigma}(\suffix) \\
  = & E_{\evidenceA}(\transition_{\evidenceA}^*(b_{\RowsP}(\TableRow{\prefix'\cdot\sigma}),\suffix)). & \text{(\cref{lemma:decomposition_correctness} with $x=\TableRow{\prefix'\cdot\sigma}$)}
 \end{align*}
 \qed{}
\end{proof}
\begin{lemma}[transition preservation]
 \label{lemma:observed_transition_preservation}
 Assume $\Table$ with index sets $\PrefixSet, \SuffixSet \subseteq \TimedWords$ and monotone basis $\Basis\subseteq \PrefixSet$ is cohesive.
 Let $\evidenceA = (\RowsP, \PhiTable, b_{\RowsP}(\TableRow{\emptyword}), \Final, \transition_{\evidenceA})$
 be the evidence AFA and let $\hypothesisA$ be the ARTA constructed from $\evidenceA$ with \cref{algorithm:arta_construction}.
 For any $\formula\in\Boolean{\RowsP}$ and any $w\in\PhiTable^*$, we have
 $\transition_{\evidenceA}^*(\formula, w) = \transition_{\hypothesisA}^*(\formula, w)$.
 In particular, $\Lg(\evidenceA)\cap \PhiTable^*=\Lg(\hypothesisA)\cap \PhiTable^*$.
\end{lemma}
\begin{proof}
 Fix an observed letter $\sigma=(\action,\delay)\in\PhiTable$.
 By construction of $\hypothesisA$ in \cref{algorithm:arta_construction}, for each  $r\in\RowsP$, the unique outgoing $\action$-transition from $r$ whose guard contains $\delay$ has target $\transition_{\evidenceA}(r,\sigma)$.
 Therefore, $\transition_{\evidenceA}(r, \sigma)=\transition_{\hypothesisA}(r, \sigma)$ holds for any $r\in\RowsP$.
 By structural induction on $\formula\in\Boolean{\RowsP}$, we obtain
 $\transition_{\evidenceA}^*(\formula,\sigma)=\transition_{\hypothesisA}^*(\formula,\sigma)$.

 We now show the claim for any $\word\in\PhiTable^*$ by induction on the length of $\word$.
 The base case $\word=\emptyword$ is immediate.
 For the induction step, write $\word = \sigma \cdot \word'$ with $\sigma \in \PhiTable$.
 Then, we obtain the following.
 \begin{align*}
  \transition_{\evidenceA}^*(\formula,\sigma\cdot \word')
  & = \transition_{\evidenceA}^*(\transition_{\evidenceA}^*(\formula,\sigma),\word') & \text{(definition of $\transition^*$)}\\
  & = \transition_{\evidenceA}^*(\transition_{\hypothesisA}^*(\formula,\sigma),\word') & \text{(one-step equality above)}\\
  & = \transition_{\hypothesisA}^*(\transition_{\hypothesisA}^*(\formula,\sigma),\word') & \text{(induction hypothesis)}\\
  & = \transition_{\hypothesisA}^*(\formula,\sigma\cdot \word') & \text{(definition of $\transition^*$)}
 \end{align*}
 The claim on the languages follows by taking $\formula=\InitLoc=b_{\RowsP}(\TableRow{\emptyword})$.
 \qed{}
\end{proof}

The following proves \cref{theorem:faithfulness}.

\recallResult{theorem:faithfulness}{\faithfulnessStatement}

\begin{proof}%
 [\cref{theorem:faithfulness}]
 Let $\evidenceA$ be the evidence AFA constructed from $\Table$.
 By \cref{lemma:observed_transition_preservation}, we have
 $\Lg(\evidenceA)\cap \PhiTable^*=\Lg(\hypothesisA)\cap \PhiTable^*$.
 Therefore, it suffices to prove the two equivalences in the statement with $\evidenceA$ in place of $\hypothesisA$, and then lift them to $\hypothesisA$ using this equality.

 For any $\prefix\in\PrefixSet$ and $\suffix\in\SuffixSet$, we have
 \begin{align*}
  &\TableCell{\prefix}{\suffix}=\top\\
  \iff& \TableRow{\prefix}(\suffix)=\top\\
  \iff& E_{\evidenceA}(\transition_{\evidenceA}^*(b_{\RowsP}(\TableRow{\prefix}),\suffix))=\top & \text{(\cref{lemma:decomposition_correctness} with $x=\TableRow{\prefix} \in \Boolean{\RowsP}$)}\\
  \iff& E_{\evidenceA}(\transition_{\evidenceA}^*(\transition_{\evidenceA}^*(b_{\RowsP}(\TableRow{\emptyword}),\prefix),\suffix))=\top & \text{(\cref{lemma:prefix_row_invariant})}\\
  \iff& E_{\evidenceA}(\transition_{\evidenceA}^*(b_{\RowsP}(\TableRow{\emptyword}),\prefix \cdot \suffix))=\top & \text{(definition of $\transition^*$)}\\
  \iff& E_{\evidenceA}(\transition_{\evidenceA}^*(\InitLoc,\prefix \cdot \suffix))=\top & \text{($\InitLoc=b_{\RowsP}(\TableRow{\emptyword})$ by \cref{definition:evidence_afa})}\\
  \iff& \prefix \cdot \suffix\in\Lg(\evidenceA) & \text{(semantics of AFAs)} \\
  \iff& \prefix \cdot \suffix\in\Lg(\hypothesisA). & \text{(\cref{lemma:observed_transition_preservation} and $\prefix \cdot \suffix\in\PhiTable^*$)}.
 \end{align*}
 \qed{}
\end{proof}
\subsection{Proof of \cref{theorem:hypothesis_bound}}

\recallResult{theorem:hypothesis_bound}{\hypothesisBoundStatement}

\begin{proof}
 Let $\hypothesisA = (\RowsP,\Alphabet,\InitLoc,\Final, \Edge)$.
 Fix a location $r \in \RowsP$ and an action $\action\in\Alphabet$.
 Let $\delay_1, \delay_2, \dots, \delay_n$ be the increasing sequence of all delays $\delay$ such that $(\action, \delay) \in \PhiTable$.
 In \cref{algorithm:arta_construction}, we delete some elements from this list and obtain a subsequence
 $l= \tilde{\delay_1}, \tilde{\delay_2}, \cdots, \tilde{\delay_m}$, and compute guard intervals $I_1, I_2, \dots,I_m=\partition(l)$.
 By \cref{definition:partition_function}, every finite endpoint of some $I_i$ is either $0$ or $\lceil\tilde{\delay_j}\rceil$ with some $j < m$.

 Let $\basis\in\Basis$ satisfy $\TableRow{\basis}=r \in \{\top, \bot\}^{\SuffixSet}$.
 Since $\Table$ is cohesive, for every $(\action,d)\in\PhiTable$ we have $\basis\cdot(\action,d)\in\PrefixSet$,
 so the row vector $\TableRow{\basis\cdot(\action,d)}$ is defined.
 Now take any two observed delays $d,d'>K$ with $(\action,d),(\action,d')\in\PhiTable$.
 By the second condition in \cref{theorem:myhill_nerode}, we have $\basis\cdot(\action,d)\nerode{\targetLg}\basis\cdot(\action,d')$, and thus,
 $\TableRow{\basis\cdot(\action,d)}=\TableRow{\basis\cdot(\action,d')}$ holds.
 Therefore, the target formulas coincide, \ie{} we have $\transition_{\evidenceA}(r,(\action,d)) = b_{\RowsP}(\TableRow{\basis\cdot(\action,d)}) = b_{\RowsP}(\TableRow{\basis\cdot(\action,d')}) = \transition_{\evidenceA}(r,(\action,d'))$.
 Since $\transition_{\evidenceA}(r,(\action,\delay))$ is constant for any observed delay $\delay$ strictly greater than $K$,
 at most one delay greater than $K$ can be in $l$, and it must be the last element.
 Therefore, we have $\lceil \tilde{\delay_j} \rceil \le K$ for every $j<m$, and thus, the natural number in any interval $I_i$ with $i \leq m$ is at most $K$.
 \qed{}
\end{proof}
\subsection{Proof of \cref{theorem:normalized_counterexample}}

Before proving \cref{theorem:normalized_counterexample}, we show some lemmas.

\begin{lemma}[membership preservation of $\normalizeLetter$]
 \label{lemma:normalization_letter_preserves}
 Let $\A$ be an ARTA and let $\alpha\in(0,1)$ be fixed.
 For every location formula $\formula\in\Boolean{\Loc}$ and every letter with delay $(\action,\delay)\in\Alphabet\times\Rnn$, we have
 $\transition_{\A}^*(\formula,(\action,\delay))=\transition_{\A}^*(\formula,\normalizeLetter((\action,\delay)))$.
\end{lemma}
\begin{proof}
 Let $(\action,\delay') = \normalizeLetter((\action,\delay))$.
 If $\delay\in\setN$, then $\delay'=\delay$ and the claim is immediate.
 Assume $\delay\notin\setN$.
 Then, we have $\delay'=\lfloor\delay\rfloor+\alpha$.

 For any interval guard $I\in\Intervals$, since the endpoints are integers or $\infty$, we have $\delay\in I \iff \delay'\in I$.
 For any $\loc\in\Loc$, 
 from the definition of $\transition_{\A}$, the set of $\action$-edges $(\loc,\action,I,\phi)\in\Edge$ enabled by $\delay$ is the same as the set enabled by $\delay'$.
 By the determinism of ARTAs, we have $\transition_{\A}(\loc,(\action,\delay))=\transition_{\A}(\loc,(\action,\delay'))$.

 We lift the equality to arbitrary location formulas by structural induction on $\formula\in\Boolean{\Loc}$.
 The cases $\top$ and $\bot$ are immediate.
 For an atomic location $\loc$, we have $\transition_{\A}^*(\loc,(\action,\delay))=\transition_{\A}(\loc,(\action,\delay)) = \transition_{\A}(\loc,(\action,\delay')) = \transition_{\A}^*(\loc,(\action,\delay'))$.
 The Boolean cases follow directly from the definition.
 Overall, we have $\transition_{\A}^*(\formula,(\action,\delay))=\transition_{\A}^*(\formula,(\action,\delay'))=\transition_{\A}^*(\formula,\normalizeLetter((\action,\delay)))$.
 \qed{}
\end{proof}
\begin{lemma}[membership preservation of $\normalize$]
 \label{lemma:normalization_preserves}
 Let $\A$ be an ARTA and let $\alpha\in(0,1)$ be fixed.
 For every location formula $\formula\in\Boolean{\Loc}$ and every timed word $\word \in \TimedWords$, we have
 $\transition_{\A}^*(\formula, \word)=\transition_{\A}^*(\formula,\normalize(\word))$.
 In particular, for any $\word\in\TimedWords$, we have $\word \in \Lg(\A) \iff \normalize(\word) \in \Lg(\A)$.
\end{lemma}
\begin{proof}
 The proof is by induction on the length of $\word$.
 For $\word=\emptyword$, the claim is immediate.
 For the induction step, let $\word=(\action,\delay)\cdot \word'$.
 Then, we have 
 \begin{align*}
  \transition_{\A}^*(\formula,(\action,\delay)\cdot \word') &= \transition_{\A}^*(\transition_{\A}^*(\formula,(\action,\delay)),\word') & \text{(by definition of $\transition^*_{\A}$)}\\
  &= \transition_{\A}^*(\transition_{\A}^*(\formula,\normalizeLetter((\action,\delay))),\word') & \text{(by \cref{lemma:normalization_letter_preserves})}\\
  &= \transition_{\A}^*(\transition_{\A}^*(\formula,\normalizeLetter((\action,\delay))),\normalize(\word')) & \text{(by induction hypothesis)}\\
  &= \transition_{\A}^*(\formula,\normalizeLetter((\action,\delay))\cdot \normalize(\word')) & \text{(by definition of $\transition^*_{\A}$)}\\
  &= \transition_{\A}^*(\formula,\normalize(\word)) & \text{(by definition of $\normalize$)}.
 \end{align*}

 The claim about membership in $\lang$ is immediately derived from the first claim by taking $\formula=\InitLoc$.
 \qed{}
\end{proof}

The following proves \cref{theorem:normalized_counterexample}

\recallResult{theorem:normalized_counterexample}{\normalizedCounterexampleStatement}

\begin{proof}
 Let $\tilde\A$ be an ARTA recognizing $\lang$.
 By \cref{lemma:normalization_preserves}, we have
 $\word \in \Lg(\A) \iff \normalize(\word) \in \Lg(\A)$ and 
 $\word \in \Lg(\tilde\A) \iff \normalize(\word) \in \Lg(\tilde\A)$.
 Thus, we also have
 $\word \in \lang \setdiff \Lg(\A) \iff \normalize(\word) \in \lang \setdiff \Lg(\A)$.
 \qed{}
\end{proof}
\subsection{Proof of \cref{proposition:milp_basis_encoding}}

\recallResult{proposition:milp_basis_encoding}{\MILPBasisEncodingStatement}

\begin{proof}
 Assume that $U$ is a monotone basis of $R$, and fix $x \in R$, $\suffix^+ \in \mathrm{Pos}(x)$, and $\suffix^- \in \mathrm{Neg}(x)$.
 Since $x \in \Boolean{U}$, some positive Boolean expression over $U$ evaluates to $\top$ at $\suffix^+$ and to $\bot$ at $\suffix^-$.
 By induction on that expression, one finds a variable $r \in U$ that is $\top$ at $\suffix^+$ and $\bot$ at $\suffix^-$.
 Therefore, some $r \in U$ covers $(x,\suffix^+,\suffix^-)$.

 Assume that every separator obligation is covered.
 Fix $x \in R$.
 For each pair $(\suffix^+,\suffix^-)\in \mathrm{Pos}(x)\times\mathrm{Neg}(x)$, we let $r_{\suffix^+,\suffix^-}\in U$ be a row covering $(x,\suffix^+,\suffix^-)$.
 For each $\suffix^+\in\mathrm{Pos}(x)$, we let $\psi_{\suffix^+}=\bigwedge_{\suffix^- \in \mathrm{Neg}(x)} r_{\suffix^+,\suffix^-}$.
 Note that this conjunction is $\top$ if $\mathrm{Neg}(x)=\emptyset$.
 Then, we let $\phi_x=\bigvee_{\suffix^+\in\mathrm{Pos}(x)} \psi_{\suffix^+}$.
 Note that this disjunction is $\bot$ if $\mathrm{Pos}(x)=\emptyset$.
 By construction, $\phi_x \in \Boolean{U}$.
 Moreover, $\phi_x(\suffix^+) = \top$ for each $\suffix^+ \in \mathrm{Pos}(x)$ because every conjunct of $\psi_{\suffix^+}$ is $\top$ at $\suffix^+$, and $\phi_x(\suffix^-) = \bot$ for each $\suffix^- \in \mathrm{Neg}(x)$ because every disjunct contains the conjunct $r_{\suffix^+,\suffix^-}$, which is $\bot$ at $\suffix^-$.
 Thus, $\phi_x = x$, and $x \in \Boolean{U}$ holds.
 Since this holds for every $x \in R$, $U$ is a monotone basis of $R$.
\qed{}
\end{proof}
\subsection{Proof of \cref{lemma:row_bound}}

\recallResult{lemma:row_bound}{\rowBoundStatement}

\begin{proof}
 By the definition of row vectors, 
 for any $\prefix, \prefix' \in \PrefixSet$ satisfying $\TableRow{\prefix} \neq \TableRow{\prefix'}$,
 there is $\suffix\in\SuffixSet$ such that
 $\prefix \cdot \suffix \in \targetLg$ but $\prefix'\cdot\suffix\notin\targetLg$ (or vice versa).
 Therefore, $\prefix \not\nerode{\targetLg} \prefix'$ by definition of $\nerode{\targetLg}$ in \cref{theorem:myhill_nerode}.
 By \cref{theorem:myhill_nerode}, the minimal DRTA recognizing $\targetLg$ has exactly as many locations as there are equivalence classes of $\nerode{\targetLg}$.
 Thus, there are at most $n$ distinct row vectors in $\Table$.
 \qed{}
\end{proof}
\subsection{Proof of \cref{theorem:complexity_analysis}}

Before proving \cref{theorem:complexity_analysis}, we show an auxiliary definition and lemmas.

\begin{definition}[$t_{\Table}(\suffix)$]
 Let $\targetLg$ be a real-time language and let $\nerode{\targetLg}$ be the equivalence relation on $\TimedWords$ for $\targetLg$ in \cref{theorem:myhill_nerode}.
 Let $n$ be the number of equivalence classes of $\nerode{\targetLg}$.
 For an observation table $\Table$ with index sets $\PrefixSet, \SuffixSet \subseteq \TimedWords$,
 we define a function $t_{\Table}\colon \SuffixSet \to \{\bot,\top\}^n$ as follows.
 Fix $\word_1, \word_2, \dots, \word_n \in \TimedWords$ such that for each $i \neq j$, $\word_i \not\nerode{}\word_j$.
 Then, we let $t_{\Table}(\suffix)(i) = \top \iff \word_i \cdot \suffix \in \targetLg$. 
\end{definition}
\begin{lemma}%
 \label{lemma:refinement_cases}
 Let $\targetLg$ be the target real-time language.
 Let $\Table$ be a cohesive observation table with index sets $\PrefixSet, \SuffixSet \subseteq \TimedWords$ and monotone basis $\Basis\subseteq \PrefixSet$.
 Assume that the hypothesis ARTA $\hypothesisA$ constructed from $\Table$ does not recognize $\targetLg$ and let $\cex$ be the counterexample returned at the failed equivalence query.
 Then, for the cohesive observation table $\Table'$ after processing $\cex$ with index sets $\PrefixSet', \SuffixSet' \subseteq \TimedWords$,
 we have at least one of the following.
 \begin{itemize}
  \item There is $\suffix \in \SuffixSet' \setminus \SuffixSet$ such that $t_{\Table'}(\suffix) \neq t_{\Table'}(\suffix')$ for any $\suffix' \in \SuffixSet$.
  \item $\letters{\PrefixSet' \cup \SuffixSet'}$ is a strict superset of $\letters{\PrefixSet \cup \SuffixSet}$.
 \end{itemize}
\end{lemma}
\begin{proof}
 We prove the claim by contradiction.
 Assume that after processing $\cex$, we have both of the following.
 \begin{itemize}
  \item For every $\suffix \in \SuffixSet' \setminus \SuffixSet$, there is $\suffix' \in \SuffixSet$ such that $t_{\Table'}(\suffix) = t_{\Table'}(\suffix')$.
  \item $\letters{\PrefixSet' \cup \SuffixSet'} = \letters{\PrefixSet \cup \SuffixSet}$.
 \end{itemize}

 Fix $\suffix \in \SuffixSet' \setminus \SuffixSet$.
 Such $\suffix$ exists because we have $\normalize(\cex) \not\in \SuffixSet$ by \cref{theorem:faithfulness,theorem:normalized_counterexample} and $\normalize(\cex) \in \SuffixSet'$ by \cref{algorithm:ALstarRTA}.
 Let $\suffix' \in \SuffixSet$ be such that $t_{\Table'}(\suffix) = t_{\Table'}(\suffix')$.
 From the definition of $t$,
 for any $\word \in \TimedWords$, we have $\word \cdot \suffix \in \targetLg$ if and only if $\word \cdot \suffix' \in \targetLg$.
 In particular, for every $\prefix \in \PrefixSet$, we have $\TableCell{\prefix}{\suffix} = \TableCell{\prefix}{\suffix'}$.
 Therefore, the observation table $\Table'$ is still $\Basis$-closed and $\Basis$ does not increase at \cref{algorithm:ALstarRTA:close} of \cref{algorithm:ALstarRTA}.
 Since $\letters{\PrefixSet' \cup \SuffixSet'} = \letters{\PrefixSet \cup \SuffixSet}$ and $\Table$ is already evidence-closed, we have $\PrefixSet = \PrefixSet'$.

 Let $\hypothesisA'$ be the ARTA constructed from $\Table'$ with the same monotone basis $\Basis$.
 We note that the actual hypothesis ARTA may not be constructed using the same basis $\Basis$, depending on the BIP-based basis computation in \cref{section:exact_basis_milp}; however, $\Basis$ is a valid monotone basis, and $\hypothesisA'$ is a valid ARTA\@.
 Since both $\Basis$ and $\letters{\PrefixSet' \cup \SuffixSet'}$ do not change, the evidence AFA constructed from $\Table'$ with $\Basis$ is semantically equivalent to the one constructed from $\Table$ with $\Basis$.
 Thus, the ARTA $\hypothesisA'$ constructed from $\Table'$ with $\Basis$ is semantically equivalent to $\hypothesisA$.
 However, since $\normalize(\cex)$ is a counterexample to $\hypothesisA$, it is also a counterexample to $\hypothesisA'$, which contradicts \cref{theorem:faithfulness}.
 \qed{}
\end{proof}
\begin{lemma}%
 \label{lemma:letter_refinement}
 Let $\targetLg$ be the target real-time language and let $K$ be the constant in the second condition of \cref{theorem:myhill_nerode}.
 Let $\Table$ be a cohesive observation table with index sets $\PrefixSet, \SuffixSet \subseteq \TimedWords$ and monotone basis $\Basis\subseteq \PrefixSet$.
 Assume that the hypothesis ARTA $\hypothesisA$ constructed from $\Table$ does not recognize $\targetLg$ and
 for any $\suffix \in \SuffixSet' \setminus \SuffixSet$ there is $\suffix' \in \SuffixSet$ satisfying $t_{\Table'}(\suffix) = t_{\Table'}(\suffix')$, where
 $\cex$ is the counterexample returned at the failed equivalence query and
 $\SuffixSet'$ is the set of suffixes after processing $\cex$.
 Then, for the set $\PrefixSet'$ of prefixes after processing $\cex$, we have one of the following.
 \begin{itemize}
  \item There is $(\action, \delay) \in \letters{\PrefixSet' \cup \SuffixSet'} \setminus \letters{\PrefixSet \cup \SuffixSet}$ with $\delay \leq K$.
  \item There is $(\action, \delay) \in \letters{\PrefixSet' \cup \SuffixSet'} \setminus \letters{\PrefixSet \cup \SuffixSet}$ with $\delay > K$ and there is no $(\action, \delay') \in \letters{\PrefixSet \cup \SuffixSet}$ with $\delay' > K$.
 \end{itemize}
\end{lemma}
\begin{proof}
 We prove the claim by contradiction.
 Assume that after processing $\cex$, for any $(\action, \delay) \in \letters{\PrefixSet' \cup \SuffixSet'} \setminus \letters{\PrefixSet \cup \SuffixSet}$, we have $\delay > K$ and there is $(\action, \delay') \in \letters{\PrefixSet \cup \SuffixSet}$ with $\delay' > K$.
 Since we have $\delay > K$ and $\delay' > K$, the second condition in \cref{theorem:myhill_nerode} implies that for any $\prefix \in \PrefixSet$ and $\suffix \in \SuffixSet$, we have $\prefix\cdot(\action,\delay)\cdot\suffix \in \targetLg$ if and only if $\prefix\cdot(\action,\delay')\cdot\suffix \in \targetLg$.
 Since each newly added suffix in $\SuffixSet' \setminus \SuffixSet$ duplicates an old column by assumption,
 each row added at \cref{algorithm:ALstarRTA:evidence_closed} of \cref{algorithm:ALstarRTA} must be the same as some existing row.
 Thus, $\Table$ is still $\Basis$-closed and $\Basis$ does not increase at \cref{algorithm:ALstarRTA:close} of \cref{algorithm:ALstarRTA}.

 Let $\hypothesisA'$ be the ARTA constructed from $\Table'$ with the same monotone basis $\Basis$.
 We note that the actual hypothesis ARTA may not be constructed using the same basis $\Basis$, depending on the BIP-based basis computation in \cref{section:exact_basis_milp}; however, $\Basis$ is a valid monotone basis, and $\hypothesisA'$ is a valid ARTA\@.
 Since $\Basis$ does not change and having both $(\action, \delay)$ and $(\action, \delay')$ with $\delay, \delay' > K$ does not contribute to the refinement of guard intervals in \cref{algorithm:arta_construction},
 the ARTA $\hypothesisA'$ constructed from $\Table'$ with $\Basis$ is semantically equivalent to $\hypothesisA$.
 However, since $\normalize(\cex)$ is a counterexample to $\hypothesisA$, it is also a counterexample to $\hypothesisA'$, which contradicts \cref{theorem:faithfulness}.
 \qed{}
\end{proof}

The following proves \cref{theorem:complexity_analysis}.

\recallResult{theorem:complexity_analysis}{\complexityAnalysisStatement}

\begin{proof}
 \textbf{Number of equivalence queries.}
 By \cref{lemma:refinement_cases}, after processing a counterexample, we discover either a new $t_{\Table}(\suffix)$ or a new $(\action, \delay)$ in $\letters{\PrefixSet \cup \SuffixSet}$.
 Since each $t_{\Table}(\suffix)$ is a Boolean vector of length $n$, there are at most $2^n$ distinct $t_{\Table}(\suffix)$, and thus, 
 discovery of a new $t_{\Table}(\suffix)$ occurs at most $2^n$ times in \cref{algorithm:ALstarRTA}.
 
 Assume the counterexample $\cex$ returned by an equivalence query does not reveal a new $t_{\Table}(\suffix)$.
 By \cref{lemma:letter_refinement}, we discover a new $(\action, \delay)$ in $\letters{\PrefixSet \cup \SuffixSet}$ such that either $\delay \leq K$ or $\delay > K$ and there is no $(\action, \delay') \in \letters{\PrefixSet \cup \SuffixSet}$ with $\delay' > K$.
 For each letter $\action$, there are at most $2K + 2$ such delays, and thus, there are at most $|\Alphabet| (2K + 2)$ refinements of $\letters{\PrefixSet \cup \SuffixSet}$ by equivalence queries.
 Overall, the number of equivalence queries is bounded by $2^n + |\Alphabet| (2K + 2) + 1$.

 \noindent
 \textbf{Number of membership queries.}
 After each failed equivalence query, at most $h$ new suffixes are added.
 Thus, we have $|\SuffixSet| \leq M h + 1$.
 Moreover, each failed equivalence query increases $\letters{\PrefixSet \cup \SuffixSet}$ at most by $h$, and we have
 $|\letters{\PrefixSet \cup \SuffixSet}| \leq M h$.
 Since $\Basis$ is a monotone basis of $\Table$, $|\Basis| = |\mathrm{Rows}(\Basis)| \leq |\mathrm{Rows}(\PrefixSet)| \leq n$ follows from \cref{lemma:row_bound}.
 $\PrefixSet$ increases only to ensure evidence-closedness at \cref{algorithm:ALstarRTA:evidence_closed}.
 By \cref{lemma:row_bound}, there are at most $n$ distinct row vectors in the observation tables appearing in this computation.
 By determinism of the mapping from $\mathrm{Rows}(\PrefixSet)$ to $\PrefixSet$ in the monotone basis identification,
 there are at most $n$ distinct $\prefix \in \PrefixSet$ in the monotone bases between each equivalence query.
 Thus, between each equivalence query, 
 $\PrefixSet$ increases at most by $|\letters{\PrefixSet \cup \SuffixSet}| \times n \leq M h n$, and we have $|\PrefixSet| \leq M^2 hn + 1$.
 Overall, the number of membership queries is bounded by
 $|\PrefixSet| \times |\SuffixSet| \leq (M^2hn + 1) \times (Mh + 1) = M^3h^2n + M^2hn + Mh + 1$.
 \qed{}
\end{proof}
\section{Omitted steps of the example in \cref{section:worked_example}}

Here, we present the concrete steps of the worked example in \cref{section:worked_example}.

The learner then adds $(a,3.5)(b,7.5)$ to $\SuffixSet$ and makes the observation table cohesive.
The resulting observation table $\Table_6$ and the corresponding hypothesis $\hypothesisA^6$ are shown in \cref{figure:observation_tables:6,figure:hypothesis:6}.
Here, we still have $\Basis=\{\emptyword,\Athree,\Bseven\}$, but the $a$-guard is refined from $[0,\infty)$ to $[0,3]$.
The learner asks an equivalence query, and the teacher returns a counterexample $(b,1.5)(b,7)$, which is rejected by the target ARTA but accepted by $\hypothesisA^6$.

The learner then adds $(b,1.5)(b,7)$ to $\SuffixSet$ and makes the observation table cohesive.
The resulting observation table $\Table_7$ and the corresponding hypothesis $\hypothesisA^7$ are shown in \cref{figure:observation_tables:7,figure:hypothesis:7}.
Here, we still have $\Basis=\{\emptyword,\Athree,\Bseven\}$, and the $b$-behavior is refined by separating the interval $[2,3)$ from the region below $2$.
The learner asks an equivalence query, and the teacher returns a counterexample $(a,3)(a,3.5)(b,7)$, which is accepted by the target ARTA but rejected by $\hypothesisA^7$.

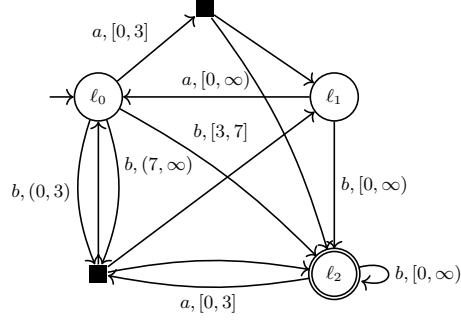
\begin{figure}[tbp]
 \begin{subfigure}{\linewidth}
  \centering
  \footnotesize
  \begin{tabular}{r|ccccccc}
   & $\emptyword$ & $\Bzero$ & $\Btwohalf$ & $\Bseven$ & $\Bsevenhalf$ & $\Athree\Bzero$ & $\Athreehalf\Bsevenhalf$\\\hline
   \rowcolor{gray!15}
   $\emptyword$ & $\bot$ & $\bot$ & $\bot$ & $\top$ & $\bot$ & $\top$ & $\bot$\\
   \rowcolor{gray!15}
   $\Athree$ & $\bot$ & $\top$ & $\top$ & $\top$ & $\top$ & $\bot$ & $\bot$\\
   $\Athreehalf$ & $\bot$ & $\bot$ & $\bot$ & $\bot$ & $\bot$ & $\bot$ & $\bot$\\
   $\Bzero$ & $\bot$ & $\bot$ & $\bot$ & $\bot$ & $\bot$ & $\bot$ & $\bot$\\
   $\Btwohalf$ & $\bot$ & $\bot$ & $\bot$ & $\top$ & $\bot$ & $\bot$ & $\bot$\\
   \rowcolor{gray!15}
   $\Bseven$ & $\top$ & $\top$ & $\top$ & $\top$ & $\top$ & $\bot$ & $\bot$\\
   $\Bsevenhalf$ & $\bot$ & $\bot$ & $\bot$ & $\top$ & $\bot$ & $\bot$ & $\bot$\\
   $\Athree\Athree$ & $\bot$ & $\bot$ & $\bot$ & $\top$ & $\bot$ & $\top$ & $\bot$\\
   $\Athree\Athreehalf$ & $\bot$ & $\bot$ & $\bot$ & $\top$ & $\bot$ & $\top$ & $\bot$\\
   $\Athree\Bzero$ & $\top$ & $\top$ & $\top$ & $\top$ & $\top$ & $\bot$ & $\bot$\\
   $\Athree\Btwohalf$ & $\top$ & $\top$ & $\top$ & $\top$ & $\top$ & $\bot$ & $\bot$\\
   $\Athree\Bseven$ & $\top$ & $\top$ & $\top$ & $\top$ & $\top$ & $\bot$ & $\bot$\\
   $\Athree\Bsevenhalf$ & $\top$ & $\top$ & $\top$ & $\top$ & $\top$ & $\bot$ & $\bot$\\
   $\Bseven\Athree$ & $\bot$ & $\bot$ & $\bot$ & $\top$ & $\bot$ & $\bot$ & $\bot$\\
   $\Bseven\Athreehalf$ & $\bot$ & $\bot$ & $\bot$ & $\bot$ & $\bot$ & $\bot$ & $\bot$\\
   $\Bseven\Bzero$ & $\top$ & $\top$ & $\top$ & $\top$ & $\top$ & $\bot$ & $\bot$\\
   $\Bseven\Btwohalf$ & $\top$ & $\top$ & $\top$ & $\top$ & $\top$ & $\bot$ & $\bot$\\
   $\Bseven\Bseven$ & $\top$ & $\top$ & $\top$ & $\top$ & $\top$ & $\bot$ & $\bot$\\
   $\Bseven\Bsevenhalf$ & $\top$ & $\top$ & $\top$ & $\top$ & $\top$ & $\bot$ & $\bot$
  \end{tabular}
  \caption{Sixth observation table $\Table_6$.}%
  \label{figure:observation_tables:6}
 \end{subfigure}
 \hfill
 \begin{subfigure}{\linewidth}
  \centering
  \footnotesize
  \begin{tikzpicture}[auto, semithick,node distance=1.5cm,scale=0.78,every node/.style={initial text={},transform shape}]
    \node[state, initial] (q0) at (0, 0) {$\WorkLocZero$};
    \node[state, accepting] (q1) at (4, -3) {$\WorkLocTwo$};
    \node[state] (q2) at (4, 0) {$\WorkLocOne$};
    \workbranch{c12}{(1.8,1.5)}
    \workbranch{c012}{(0,-3)}
    \path[->]
      (q0) edge node[above left] {$a,[0,3]$} (c12)
      (q0) edge[bend right=20] node[left] {$b,(0,3)$} (c012)
      (q0) edge[bend left=8] node[above right,pos=0.3] {$b,[3,7]$} (q1)
      (q0) edge[bend left=20] node[pos=0.3,right] {$b,(7,\infty)$} (c012)
      (q1) edge[bend left=10] node[below] {$a,[0,3]$} (c012)
      (q1) edge[loop right] node {$b,[0,\infty)$} (q1)
      (q2) edge node[above] {$a,[0,\infty)$} (q0)
      (q2) edge node[right] {$b,[0,\infty)$} (q1)
      (c12) edge[bend left=10] (q1)
      (c12) edge (q2)
      (c012) edge (q0)
      (c012) edge[bend left=10] (q1)
      (c012) edge (q2);
  \end{tikzpicture}
  \caption{Sixth hypothesis $\hypothesisA^6$.}%
  \label{figure:hypothesis:6}
 \end{subfigure}
 \caption{Observation table and the corresponding hypothesis in the sixth step.}
\end{figure}
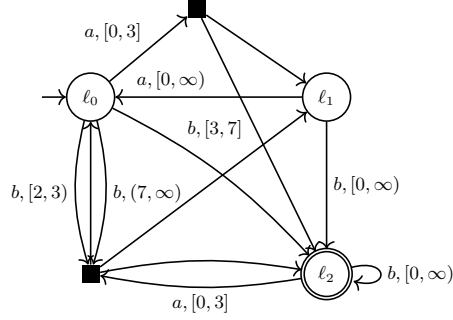
\begin{figure}[tbp]
 \begin{subfigure}{\linewidth}
  \centering
  \scriptsize
  \begin{tabular}{r|cccccccc}
   & $\emptyword$ & $\Bzero$ & $\Btwohalf$ & $\Bseven$ & $\Bsevenhalf$ & $\Athree\Bzero$ & $\Athreehalf\Bsevenhalf$ & $\Bonehalf\Bseven$\\\hline
   \rowcolor{gray!15}
   $\emptyword$ & $\bot$ & $\bot$ & $\bot$ & $\top$ & $\bot$ & $\top$ & $\bot$ & $\bot$\\
   \rowcolor{gray!15}
   $\Athree$ & $\bot$ & $\top$ & $\top$ & $\top$ & $\top$ & $\bot$ & $\bot$ & $\top$\\
   $\Athreehalf$ & $\bot$ & $\bot$ & $\bot$ & $\bot$ & $\bot$ & $\bot$ & $\bot$ & $\bot$\\
   $\Bzero$ & $\bot$ & $\bot$ & $\bot$ & $\bot$ & $\bot$ & $\bot$ & $\bot$ & $\bot$\\
   $\Bonehalf$ & $\bot$ & $\bot$ & $\bot$ & $\bot$ & $\bot$ & $\bot$ & $\bot$ & $\bot$\\
   $\Btwohalf$ & $\bot$ & $\bot$ & $\bot$ & $\top$ & $\bot$ & $\bot$ & $\bot$ & $\bot$\\
   \rowcolor{gray!15}
   $\Bseven$ & $\top$ & $\top$ & $\top$ & $\top$ & $\top$ & $\bot$ & $\bot$ & $\top$\\
   $\Bsevenhalf$ & $\bot$ & $\bot$ & $\bot$ & $\top$ & $\bot$ & $\bot$ & $\bot$ & $\bot$\\
   $\Athree\Athree$ & $\bot$ & $\bot$ & $\bot$ & $\top$ & $\bot$ & $\top$ & $\bot$ & $\bot$\\
   $\Athree\Athreehalf$ & $\bot$ & $\bot$ & $\bot$ & $\top$ & $\bot$ & $\top$ & $\bot$ & $\bot$\\
   $\Athree\Bzero$ & $\top$ & $\top$ & $\top$ & $\top$ & $\top$ & $\bot$ & $\bot$ & $\top$\\
   $\Athree\Bonehalf$ & $\top$ & $\top$ & $\top$ & $\top$ & $\top$ & $\bot$ & $\bot$ & $\top$\\
   $\Athree\Btwohalf$ & $\top$ & $\top$ & $\top$ & $\top$ & $\top$ & $\bot$ & $\bot$ & $\top$\\
   $\Athree\Bseven$ & $\top$ & $\top$ & $\top$ & $\top$ & $\top$ & $\bot$ & $\bot$ & $\top$\\
   $\Athree\Bsevenhalf$ & $\top$ & $\top$ & $\top$ & $\top$ & $\top$ & $\bot$ & $\bot$ & $\top$\\
   $\Bseven\Athree$ & $\bot$ & $\bot$ & $\bot$ & $\top$ & $\bot$ & $\bot$ & $\bot$ & $\bot$\\
   $\Bseven\Athreehalf$ & $\bot$ & $\bot$ & $\bot$ & $\bot$ & $\bot$ & $\bot$ & $\bot$ & $\bot$\\
   $\Bseven\Bzero$ & $\top$ & $\top$ & $\top$ & $\top$ & $\top$ & $\bot$ & $\bot$ & $\top$\\
   $\Bseven\Bonehalf$ & $\top$ & $\top$ & $\top$ & $\top$ & $\top$ & $\bot$ & $\bot$ & $\top$\\
   $\Bseven\Btwohalf$ & $\top$ & $\top$ & $\top$ & $\top$ & $\top$ & $\bot$ & $\bot$ & $\top$\\
   $\Bseven\Bseven$ & $\top$ & $\top$ & $\top$ & $\top$ & $\top$ & $\bot$ & $\bot$ & $\top$\\
   $\Bseven\Bsevenhalf$ & $\top$ & $\top$ & $\top$ & $\top$ & $\top$ & $\bot$ & $\bot$ & $\top$
  \end{tabular}
  \caption{Seventh observation table $\Table_7$.}%
  \label{figure:observation_tables:7}
 \end{subfigure}
 \hfill
 \begin{subfigure}{0.9\linewidth}
  \centering
  \begin{tikzpicture}[auto, semithick,node distance=1.5cm,scale=0.78,every node/.style={initial text={},transform shape}]
    \node[state, initial] (q0) at (0, 0) {$\WorkLocZero$};
    \node[state, accepting] (q1) at (4, -3) {$\WorkLocTwo$};
    \node[state] (q2) at (4, 0) {$\WorkLocOne$};
    \workbranch{c12}{(1.8,1.5)}
    \workbranch{c012}{(0,-3)}
    \path[->]
      (q0) edge node[above left] {$a,[0,3]$} (c12)
      (q0) edge[bend right=15] node[left] {$b,[2,3)$} (c012)
      (q0) edge[bend left=10] node[above right,pos=.3] {$b,[3,7]$} (q1)
      (q0) edge[bend left=15] node[right] {$b,(7,\infty)$} (c012)
      (q1) edge[bend left=10] node[below] {$a,[0,3]$} (c012)
      (q1) edge[loop right] node {$b,[0,\infty)$} (q1)
      (q2) edge node[above,pos=.7] {$a,[0,\infty)$} (q0)
      (q2) edge node[right] {$b,[0,\infty)$} (q1)
      (c12) edge (q1)
      (c12) edge (q2)
      (c012) edge (q0)
      (c012) edge[bend left=10] (q1)
      (c012) edge (q2);
  \end{tikzpicture}
  \caption{Seventh hypothesis $\hypothesisA^7$.}%
  \label{figure:hypothesis:7}
 \end{subfigure}
 \caption{Observation table and the corresponding hypothesis in the seventh step.}
\end{figure}

Then, the learner adds $(a,3)(a,3.5)(b,7)$ to $\SuffixSet$ and makes the observation table cohesive.
The resulting observation table $\Table_8$ and the corresponding hypothesis $\hypothesisA^8$ are shown in \cref{figure:observation_tables:8,figure:hypothesis:8}.
Here, we still have $\Basis=\{\emptyword,\Athree,\Bseven\}$, but the target of the $a$-transition from $\WorkLocZero$ is refined from $\WorkLocOne \lor \WorkLocTwo$ to $\WorkLocTwo$.
The learner asks an equivalence query, and the teacher returns a counterexample $(a,3)(b,7.5)(a,3)(b,7)$, which is rejected by the target ARTA but accepted by $\hypothesisA^8$.

\begin{figure}[tbp]
 \begin{subfigure}{\linewidth}
  \centering
  \scriptsize
  \rotatebox{90}{\begin{tabular}{r|cccccccccc}
   & $\emptyword$ & $\Bzero$ & $\Btwohalf$ & $\Bseven$ & $\Bsevenhalf$ & $\Athree\Bzero$ & $\Athreehalf\Bseven$ & $\Athreehalf\Bsevenhalf$ & $\Bonehalf\Bseven$ & $\Athree\Athreehalf\Bseven$\\\hline
   \rowcolor{gray!15}
   $\emptyword$ & $\bot$ & $\bot$ & $\bot$ & $\top$ & $\bot$ & $\top$ & $\bot$ & $\bot$ & $\bot$ & $\top$\\
   \rowcolor{gray!15}
   $\Athree$ & $\bot$ & $\top$ & $\top$ & $\top$ & $\top$ & $\bot$ & $\top$ & $\bot$ & $\top$ & $\bot$\\
   $\Athreehalf$ & $\bot$ & $\bot$ & $\bot$ & $\bot$ & $\bot$ & $\bot$ & $\bot$ & $\bot$ & $\bot$ & $\bot$\\
   $\Bzero$ & $\bot$ & $\bot$ & $\bot$ & $\bot$ & $\bot$ & $\bot$ & $\bot$ & $\bot$ & $\bot$ & $\bot$\\
   $\Bonehalf$ & $\bot$ & $\bot$ & $\bot$ & $\bot$ & $\bot$ & $\bot$ & $\bot$ & $\bot$ & $\bot$ & $\bot$\\
   $\Btwohalf$ & $\bot$ & $\bot$ & $\bot$ & $\top$ & $\bot$ & $\bot$ & $\bot$ & $\bot$ & $\bot$ & $\bot$\\
   \rowcolor{gray!15}
   $\Bseven$ & $\top$ & $\top$ & $\top$ & $\top$ & $\top$ & $\bot$ & $\bot$ & $\bot$ & $\top$ & $\bot$\\
   $\Bsevenhalf$ & $\bot$ & $\bot$ & $\bot$ & $\top$ & $\bot$ & $\bot$ & $\bot$ & $\bot$ & $\bot$ & $\bot$\\
   $\Athree\Athree$ & $\bot$ & $\bot$ & $\bot$ & $\top$ & $\bot$ & $\top$ & $\bot$ & $\bot$ & $\bot$ & $\top$\\
   $\Athree\Athreehalf$ & $\bot$ & $\bot$ & $\bot$ & $\top$ & $\bot$ & $\top$ & $\bot$ & $\bot$ & $\bot$ & $\top$\\
   $\Athree\Bzero$ & $\top$ & $\top$ & $\top$ & $\top$ & $\top$ & $\bot$ & $\bot$ & $\bot$ & $\top$ & $\bot$\\
   $\Athree\Bonehalf$ & $\top$ & $\top$ & $\top$ & $\top$ & $\top$ & $\bot$ & $\bot$ & $\bot$ & $\top$ & $\bot$\\
   $\Athree\Btwohalf$ & $\top$ & $\top$ & $\top$ & $\top$ & $\top$ & $\bot$ & $\bot$ & $\bot$ & $\top$ & $\bot$\\
   $\Athree\Bseven$ & $\top$ & $\top$ & $\top$ & $\top$ & $\top$ & $\bot$ & $\bot$ & $\bot$ & $\top$ & $\bot$\\
   $\Athree\Bsevenhalf$ & $\top$ & $\top$ & $\top$ & $\top$ & $\top$ & $\bot$ & $\bot$ & $\bot$ & $\top$ & $\bot$\\
   $\Bseven\Athree$ & $\bot$ & $\bot$ & $\bot$ & $\top$ & $\bot$ & $\bot$ & $\bot$ & $\bot$ & $\bot$ & $\bot$\\
   $\Bseven\Athreehalf$ & $\bot$ & $\bot$ & $\bot$ & $\bot$ & $\bot$ & $\bot$ & $\bot$ & $\bot$ & $\bot$ & $\bot$\\
   $\Bseven\Bzero$ & $\top$ & $\top$ & $\top$ & $\top$ & $\top$ & $\bot$ & $\bot$ & $\bot$ & $\top$ & $\bot$\\
   $\Bseven\Bonehalf$ & $\top$ & $\top$ & $\top$ & $\top$ & $\top$ & $\bot$ & $\bot$ & $\bot$ & $\top$ & $\bot$\\
   $\Bseven\Btwohalf$ & $\top$ & $\top$ & $\top$ & $\top$ & $\top$ & $\bot$ & $\bot$ & $\bot$ & $\top$ & $\bot$\\
   $\Bseven\Bseven$ & $\top$ & $\top$ & $\top$ & $\top$ & $\top$ & $\bot$ & $\bot$ & $\bot$ & $\top$ & $\bot$\\
   $\Bseven\Bsevenhalf$ & $\top$ & $\top$ & $\top$ & $\top$ & $\top$ & $\bot$ & $\bot$ & $\bot$ & $\top$ & $\bot$
  \end{tabular}}
  \caption{Eighth observation table $\Table_8$.}%
  \label{figure:observation_tables:8}
 \end{subfigure}
 \hfill
 \begin{subfigure}{0.9\linewidth}
  \centering
  \footnotesize
  \begin{tikzpicture}[auto, semithick,node distance=1.5cm,scale=0.78,every node/.style={initial text={},transform shape}]
    \node[state, initial] (q0) at (0, 0) {$\WorkLocZero$};
    \node[state, accepting] (q1) at (4, -3) {$\WorkLocTwo$};
    \node[state] (q2) at (4, 0) {$\WorkLocOne$};
    \workbranch{c012}{(0,-3)}
    \path[->]
      (q0) edge[bend left=10] node[above] {$a,[0,3]$} (q2)
      (q0) edge[bend right=20] node[left] {$b,[2,3)$} (c012)
      (q0) edge[bend left=10] node[pos=0.7,above right] {$b,[3,7]$} (q1)
      (q0) edge[bend left=20] node[pos=0.3,right] {$b,(7,\infty)$} (c012)
      (q1) edge[bend left=15] node[below] {$a,[0,3]$} (c012)
      (q1) edge[loop right] node {$b,[0,\infty)$} (q1)
      (q2) edge[bend left=10] node[below] {$a,[0,\infty)$} (q0)
      (q2) edge node[right] {$b,[0,\infty)$} (q1)
      (c012) edge (q0)
      (c012) edge (q1)
      (c012) edge (q2);
  \end{tikzpicture}
  \caption{Eighth hypothesis $\hypothesisA^8$.}%
  \label{figure:hypothesis:8}
 \end{subfigure}
 \caption{Observation table and the corresponding hypothesis in the eighth step.}
\end{figure}
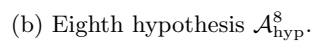

Then, the learner adds the suffixes of $(a,3)(b,7.5)(a,3)(b,7)$ to $\SuffixSet$ and makes the observation table cohesive.
The resulting observation table $\Table_9$ and the corresponding hypothesis $\hypothesisA^9$ are shown in \cref{figure:observation_tables:9,figure:running_arta}.
Here, the basis changes to $\Basis=\{\emptyword,\Athree,\Athree\Bzero\}$, so the accepting basis row changes from $\Bseven$ to $\Athree\Bzero$.
The learner asks an equivalence query, and the teacher returns no counterexample.

\begin{figure}[tbp]
  \centering
  \scriptsize
  \scalebox{0.92}{\rotatebox{90}{
  \begin{tabular}{r|ccccccccccccc}
   & $\emptyword$ & $\Bzero$ & $\Btwohalf$ & $\Bseven$ & $\Bsevenhalf$ & $\Athree\Bzero$ & $\Athree\Bseven$ & $\Athreehalf\Bseven$ & $\Athreehalf\Bsevenhalf$ & $\Bonehalf\Bseven$ & $\Athree\Athreehalf\Bseven$ & $\Bsevenhalf\Athree\Bseven$ & $\Athree\Bsevenhalf\Athree\Bseven$\\\hline
   \rowcolor{gray!15}
   $\emptyword$ & $\bot$ & $\bot$ & $\bot$ & $\top$ & $\bot$ & $\top$ & $\top$ & $\bot$ & $\bot$ & $\bot$ & $\top$ & $\top$ & $\bot$\\
   \rowcolor{gray!15}
   $\Athree$ & $\bot$ & $\top$ & $\top$ & $\top$ & $\top$ & $\bot$ & $\top$ & $\top$ & $\bot$ & $\top$ & $\bot$ & $\bot$ & $\top$\\
   $\Athreehalf$ & $\bot$ & $\bot$ & $\bot$ & $\bot$ & $\bot$ & $\bot$ & $\bot$ & $\bot$ & $\bot$ & $\bot$ & $\bot$ & $\bot$ & $\bot$\\
   $\Bzero$ & $\bot$ & $\bot$ & $\bot$ & $\bot$ & $\bot$ & $\bot$ & $\bot$ & $\bot$ & $\bot$ & $\bot$ & $\bot$ & $\bot$ & $\bot$\\
   $\Bonehalf$ & $\bot$ & $\bot$ & $\bot$ & $\bot$ & $\bot$ & $\bot$ & $\bot$ & $\bot$ & $\bot$ & $\bot$ & $\bot$ & $\bot$ & $\bot$\\
   $\Btwohalf$ & $\bot$ & $\bot$ & $\bot$ & $\top$ & $\bot$ & $\bot$ & $\top$ & $\bot$ & $\bot$ & $\bot$ & $\bot$ & $\bot$ & $\bot$\\
   $\Bseven$ & $\top$ & $\top$ & $\top$ & $\top$ & $\top$ & $\bot$ & $\top$ & $\bot$ & $\bot$ & $\top$ & $\bot$ & $\bot$ & $\bot$\\
   $\Bsevenhalf$ & $\bot$ & $\bot$ & $\bot$ & $\top$ & $\bot$ & $\bot$ & $\top$ & $\bot$ & $\bot$ & $\bot$ & $\bot$ & $\bot$ & $\bot$\\
   $\Athree\Athree$ & $\bot$ & $\bot$ & $\bot$ & $\top$ & $\bot$ & $\top$ & $\top$ & $\bot$ & $\bot$ & $\bot$ & $\top$ & $\top$ & $\bot$\\
   $\Athree\Athreehalf$ & $\bot$ & $\bot$ & $\bot$ & $\top$ & $\bot$ & $\top$ & $\top$ & $\bot$ & $\bot$ & $\bot$ & $\top$ & $\top$ & $\bot$\\
   \rowcolor{gray!15}
   $\Athree\Bzero$ & $\top$ & $\top$ & $\top$ & $\top$ & $\top$ & $\bot$ & $\bot$ & $\bot$ & $\bot$ & $\top$ & $\bot$ & $\bot$ & $\bot$\\
   $\Athree\Bonehalf$ & $\top$ & $\top$ & $\top$ & $\top$ & $\top$ & $\bot$ & $\bot$ & $\bot$ & $\bot$ & $\top$ & $\bot$ & $\bot$ & $\bot$\\
   $\Athree\Btwohalf$ & $\top$ & $\top$ & $\top$ & $\top$ & $\top$ & $\bot$ & $\bot$ & $\bot$ & $\bot$ & $\top$ & $\bot$ & $\bot$ & $\bot$\\
   $\Athree\Bseven$ & $\top$ & $\top$ & $\top$ & $\top$ & $\top$ & $\bot$ & $\bot$ & $\bot$ & $\bot$ & $\top$ & $\bot$ & $\bot$ & $\bot$\\
   $\Athree\Bsevenhalf$ & $\top$ & $\top$ & $\top$ & $\top$ & $\top$ & $\bot$ & $\bot$ & $\bot$ & $\bot$ & $\top$ & $\bot$ & $\bot$ & $\bot$\\
   $\Bseven\Athree$ & $\bot$ & $\bot$ & $\bot$ & $\top$ & $\bot$ & $\bot$ & $\top$ & $\bot$ & $\bot$ & $\bot$ & $\bot$ & $\bot$ & $\bot$\\
   $\Bseven\Athreehalf$ & $\bot$ & $\bot$ & $\bot$ & $\bot$ & $\bot$ & $\bot$ & $\bot$ & $\bot$ & $\bot$ & $\bot$ & $\bot$ & $\bot$ & $\bot$\\
   $\Bseven\Bzero$ & $\top$ & $\top$ & $\top$ & $\top$ & $\top$ & $\bot$ & $\bot$ & $\bot$ & $\bot$ & $\top$ & $\bot$ & $\bot$ & $\bot$\\
   $\Bseven\Bonehalf$ & $\top$ & $\top$ & $\top$ & $\top$ & $\top$ & $\bot$ & $\bot$ & $\bot$ & $\bot$ & $\top$ & $\bot$ & $\bot$ & $\bot$\\
   $\Bseven\Btwohalf$ & $\top$ & $\top$ & $\top$ & $\top$ & $\top$ & $\bot$ & $\bot$ & $\bot$ & $\bot$ & $\top$ & $\bot$ & $\bot$ & $\bot$\\
   $\Bseven\Bseven$ & $\top$ & $\top$ & $\top$ & $\top$ & $\top$ & $\bot$ & $\bot$ & $\bot$ & $\bot$ & $\top$ & $\bot$ & $\bot$ & $\bot$\\
   $\Bseven\Bsevenhalf$ & $\top$ & $\top$ & $\top$ & $\top$ & $\top$ & $\bot$ & $\bot$ & $\bot$ & $\bot$ & $\top$ & $\bot$ & $\bot$ & $\bot$\\
   $\Athree\Bzero\Athree$ & $\bot$ & $\bot$ & $\bot$ & $\bot$ & $\bot$ & $\bot$ & $\bot$ & $\bot$ & $\bot$ & $\bot$ & $\bot$ & $\bot$ & $\bot$\\
   $\Athree\Bzero\Athreehalf$ & $\bot$ & $\bot$ & $\bot$ & $\bot$ & $\bot$ & $\bot$ & $\bot$ & $\bot$ & $\bot$ & $\bot$ & $\bot$ & $\bot$ & $\bot$\\
   $\Athree\Bzero\Bzero$ & $\top$ & $\top$ & $\top$ & $\top$ & $\top$ & $\bot$ & $\bot$ & $\bot$ & $\bot$ & $\top$ & $\bot$ & $\bot$ & $\bot$\\
   $\Athree\Bzero\Bonehalf$ & $\top$ & $\top$ & $\top$ & $\top$ & $\top$ & $\bot$ & $\bot$ & $\bot$ & $\bot$ & $\top$ & $\bot$ & $\bot$ & $\bot$\\
   $\Athree\Bzero\Btwohalf$ & $\top$ & $\top$ & $\top$ & $\top$ & $\top$ & $\bot$ & $\bot$ & $\bot$ & $\bot$ & $\top$ & $\bot$ & $\bot$ & $\bot$\\
   $\Athree\Bzero\Bseven$ & $\top$ & $\top$ & $\top$ & $\top$ & $\top$ & $\bot$ & $\bot$ & $\bot$ & $\bot$ & $\top$ & $\bot$ & $\bot$ & $\bot$\\
   $\Athree\Bzero\Bsevenhalf$ & $\top$ & $\top$ & $\top$ & $\top$ & $\top$ & $\bot$ & $\bot$ & $\bot$ & $\bot$ & $\top$ & $\bot$ & $\bot$ & $\bot$
  \end{tabular}}}
  \caption{Ninth observation table $\Table_9$.}%
  \label{figure:observation_tables:9}
\end{figure}
\section{Details of the experiments}\label{section:detail_experiments}

\paragraph{Stopping criteria in approximate basis identification.}

\ourTool{} identifies the basis using the BIP encoding explained in \cref{section:exact_basis_milp}.
To avoid extremely long execution for some corner cases, we used the following approximate optimization.
\begin{itemize}
  \item The optimality (MIP) gap is set to 0.01.
  \item The time limit for solving each optimization problem is 5.0 seconds.
\end{itemize}

\cref{table:full_results:1,table:full_results:2,table:full_results:3,table:full_results:4,table:full_results:5} show the full results of experiments. The column ``DRTA size'' shows the number of states of the minimum DRTA equivalent to the target RTA.\@

\begin{table}[tbp]
  \centering
  \footnotesize
  \caption{The full results of the experiments (part 1).}
  \label{table:full_results:1}
  \begin{tabular}{lrrrrrrrrr}
    \toprule
    \multirow{3}{*}{$(|\Loc|,|\Alphabet|)-i$} & \multirow{3}{*}{\begin{tabular}{c}DRTA \\ size\end{tabular}} & \multicolumn{4}{c}{\ALstarRTA{}} & \multicolumn{4}{c}{\NLstarRTA{}} \\
    \cmidrule(lr){3-6}\cmidrule(lr){7-10}
    & & \# EqQ & \# MemQ & $|\Loc_{\hypothesisA}|$ & Total Time & \# EqQ & \# MemQ & $|\Loc_{\hypothesisA}|$ & Total Time \\
    \midrule
(3,2)-1 & 5 & 9 & 357 & \tbcolor{}3 & \tbcolor{}0:00.00 & \tbcolor{}7 & \tbcolor{}144 & 4 & 0:00.06 \\
(3,2)-2 & 7 & 18 & 1902 & \tbcolor{}3 & \tbcolor{}0:00.01 & \tbcolor{}17 & \tbcolor{}497 & 4 & 0:00.15 \\
(3,2)-3 & 4 & 12 & 785 & \tbcolor{}3 & \tbcolor{}0:00.01 & \tbcolor{}11 & \tbcolor{}184 & 4 & 0:00.07 \\
(3,2)-4 & 8 & 14 & 933 & \tbcolor{}3 & \tbcolor{}0:00.01 & \tbcolor{}9 & \tbcolor{}216 & 4 & 0:00.06 \\
(3,2)-5 & 5 & 15 & 949 & \tbcolor{}3 & \tbcolor{}0:00.01 & \tbcolor{}13 & \tbcolor{}343 & 4 & 0:00.07 \\
(3,2)-6 & 6 & \tbcolor{}13 & 1039 & \tbcolor{}3 & \tbcolor{}0:00.01 & \tbcolor{}13 & \tbcolor{}450 & 4 & 0:00.07 \\
(3,2)-7 & 5 & \tbcolor{}14 & 1326 & \tbcolor{}3 & \tbcolor{}0:00.01 & \tbcolor{}14 & \tbcolor{}540 & 4 & 0:00.09 \\
(3,2)-8 & 4 & 13 & 706 & \tbcolor{}3 & \tbcolor{}0:00.00 & \tbcolor{}10 & \tbcolor{}200 & 4 & 0:00.06 \\
(3,2)-9 & 6 & 13 & 645 & \tbcolor{}3 & \tbcolor{}0:00.01 & \tbcolor{}8 & \tbcolor{}350 & 4 & 0:00.06 \\
(3,2)-10 & 4 & 11 & 819 & \tbcolor{}3 & \tbcolor{}0:00.02 & \tbcolor{}9 & \tbcolor{}156 & 4 & 0:00.05 \\
(3,2)-11 & 6 & 15 & 1271 & \tbcolor{}3 & \tbcolor{}0:00.01 & \tbcolor{}12 & \tbcolor{}208 & 4 & 0:00.07 \\
(3,2)-12 & 8 & 16 & 1153 & \tbcolor{}3 & \tbcolor{}0:00.01 & \tbcolor{}14 & \tbcolor{}405 & 4 & 0:00.07 \\
(3,2)-13 & 8 & 16 & 1035 & \tbcolor{}3 & \tbcolor{}0:00.01 & \tbcolor{}13 & \tbcolor{}396 & 4 & 0:00.07 \\
(3,2)-14 & 4 & 9 & 353 & \tbcolor{}3 & \tbcolor{}0:00.00 & \tbcolor{}7 & \tbcolor{}180 & 4 & 0:00.04 \\
(3,2)-15 & 8 & 15 & 1326 & \tbcolor{}3 & \tbcolor{}0:00.01 & \tbcolor{}12 & \tbcolor{}378 & 4 & 0:00.06 \\
(3,2)-16 & 4 & 7 & 288 & \tbcolor{}3 & \tbcolor{}0:00.00 & \tbcolor{}4 & \tbcolor{}72 & 4 & 0:00.04 \\
(3,2)-17 & 8 & 15 & 2014 & \tbcolor{}3 & \tbcolor{}0:00.02 & \tbcolor{}14 & \tbcolor{}432 & 4 & 0:00.09 \\
(3,2)-18 & 5 & 12 & 908 & \tbcolor{}3 & \tbcolor{}0:00.01 & \tbcolor{}10 & \tbcolor{}156 & 4 & 0:00.06 \\
(3,2)-19 & 5 & 8 & 319 & \tbcolor{}3 & \tbcolor{}0:00.01 & \tbcolor{}7 & \tbcolor{}96 & 4 & 0:00.04 \\
(3,2)-20 & 8 & 17 & 1373 & \tbcolor{}3 & \tbcolor{}0:00.01 & \tbcolor{}14 & \tbcolor{}672 & 4 & 0:00.11 \\
\midrule
(4,2)-1 & 9 & 21 & 3025 & \tbcolor{}4 & \tbcolor{}0:00.02 & \tbcolor{}16 & \tbcolor{}836 & 5 & 0:00.13 \\
(4,2)-2 & 6 & 14 & 1471 & \tbcolor{}4 & \tbcolor{}0:00.01 & \tbcolor{}11 & \tbcolor{}423 & 5 & 0:00.07 \\
(4,2)-3 & 7 & 15 & 1421 & \tbcolor{}4 & \tbcolor{}0:00.01 & \tbcolor{}14 & \tbcolor{}441 & 5 & 0:00.08 \\
(4,2)-4 & 7 & \tbcolor{}13 & 935 & \tbcolor{}4 & \tbcolor{}0:00.01 & \tbcolor{}13 & \tbcolor{}536 & 5 & 0:00.10 \\
(4,2)-5 & 15 & \tbcolor{}25 & 7486 & \tbcolor{}4 & \tbcolor{}0:00.06 & 26 & \tbcolor{}1716 & 5 & 0:00.22 \\
(4,2)-6 & 7 & 15 & 2017 & \tbcolor{}4 & \tbcolor{}0:00.02 & \tbcolor{}12 & \tbcolor{}484 & 5 & 0:00.08 \\
(4,2)-7 & 10 & 15 & 1314 & \tbcolor{}4 & \tbcolor{}0:00.01 & \tbcolor{}12 & \tbcolor{}680 & 5 & 0:00.09 \\
(4,2)-8 & 6 & 15 & 1700 & \tbcolor{}4 & \tbcolor{}0:00.01 & \tbcolor{}11 & \tbcolor{}296 & 5 & 0:00.06 \\
(4,2)-9 & 7 & \tbcolor{}16 & 2019 & \tbcolor{}4 & \tbcolor{}0:00.01 & 17 & \tbcolor{}1188 & 5 & 0:00.15 \\
(4,2)-10 & 12 & 20 & 4007 & \tbcolor{}4 & \tbcolor{}0:00.02 & \tbcolor{}16 & \tbcolor{}855 & 5 & 0:00.12 \\
(4,2)-11 & 6 & 12 & 1804 & \tbcolor{}4 & \tbcolor{}0:00.01 & \tbcolor{}11 & \tbcolor{}384 & 5 & 0:00.06 \\
(4,2)-12 & 16 & 19 & 2853 & \tbcolor{}4 & \tbcolor{}0:00.03 & \tbcolor{}16 & \tbcolor{}792 & 5 & 0:00.14 \\
(4,2)-13 & 9 & 16 & 1617 & \tbcolor{}4 & \tbcolor{}0:00.02 & \tbcolor{}14 & \tbcolor{}440 & 5 & 0:00.08 \\
(4,2)-14 & 9 & 17 & 1362 & \tbcolor{}4 & \tbcolor{}0:00.01 & \tbcolor{}14 & \tbcolor{}763 & 5 & 0:00.17 \\
(4,2)-15 & 14 & 19 & 3510 & \tbcolor{}4 & \tbcolor{}0:00.03 & \tbcolor{}18 & \tbcolor{}1104 & 5 & 0:00.12 \\
(4,2)-16 & 16 & \tbcolor{}22 & 3650 & \tbcolor{}4 & \tbcolor{}0:00.03 & 23 & \tbcolor{}2325 & 5 & 0:00.31 \\
(4,2)-17 & 7 & 16 & 1747 & \tbcolor{}4 & \tbcolor{}0:00.01 & \tbcolor{}14 & \tbcolor{}468 & 5 & 0:00.09 \\
(4,2)-18 & 11 & \tbcolor{}16 & 2109 & \tbcolor{}5 & \tbcolor{}0:00.02 & 19 & \tbcolor{}1045 & 6 & 0:00.15 \\
(4,2)-19 & 8 & \tbcolor{}19 & 3036 & \tbcolor{}4 & \tbcolor{}0:00.03 & \tbcolor{}19 & \tbcolor{}820 & 5 & 0:00.13 \\
(4,2)-20 & 7 & 18 & 1915 & \tbcolor{}4 & \tbcolor{}0:00.01 & \tbcolor{}17 & \tbcolor{}612 & 5 & 0:00.09 \\
\bottomrule
\end{tabular}
\end{table}

\begin{table}[tbp]
  \centering
  \footnotesize
  \caption{The full results of the experiments (part 2).}
  \label{table:full_results:2}
  \begin{tabular}{lrrrrrrrrr}
    \toprule
    \multirow{3}{*}{$(|\Loc|,|\Alphabet|)-i$} & \multirow{3}{*}{\begin{tabular}{c}DRTA \\ size\end{tabular}} & \multicolumn{4}{c}{\ALstarRTA{}} & \multicolumn{4}{c}{\NLstarRTA{}} \\
    \cmidrule(lr){3-6}\cmidrule(lr){7-10}
    & & \# EqQ & \# MemQ & $|\Loc_{\hypothesisA}|$ & Total Time & \# EqQ & \# MemQ & $|\Loc_{\hypothesisA}|$ & Total Time \\
    \midrule
(5,2)-1 & 13 & 20 & 4822 & \tbcolor{}5 & \tbcolor{}0:00.03 & \tbcolor{}19 & \tbcolor{}793 & 6 & 0:00.14 \\
(5,2)-2 & 11 & 18 & 3566 & \tbcolor{}5 & \tbcolor{}0:00.02 & \tbcolor{}16 & \tbcolor{}616 & 6 & 0:00.10 \\
(5,2)-3 & 20 & 32 & 8386 & \tbcolor{}5 & \tbcolor{}0:00.09 & \tbcolor{}26 & \tbcolor{}1558 & 6 & 0:00.26 \\
(5,2)-4 & 10 & \tbcolor{}22 & 8086 & \tbcolor{}5 & \tbcolor{}0:00.05 & 24 & \tbcolor{}2142 & 6 & 0:00.21 \\
(5,2)-5 & 17 & \tbcolor{}18 & 4219 & \tbcolor{}5 & \tbcolor{}0:00.03 & 19 & \tbcolor{}1106 & 6 & 0:00.12 \\
(5,2)-6 & 17 & 28 & 9734 & \tbcolor{}5 & \tbcolor{}0:00.17 & \tbcolor{}26 & \tbcolor{}1818 & 6 & 0:00.23 \\
(5,2)-7 & 21 & 27 & 10321 & \tbcolor{}5 & \tbcolor{}0:00.06 & \tbcolor{}21 & \tbcolor{}1092 & 6 & 0:00.14 \\
(5,2)-8 & 10 & 18 & 4351 & \tbcolor{}5 & \tbcolor{}0:00.03 & \tbcolor{}17 & \tbcolor{}1200 & 6 & 0:00.12 \\
(5,2)-9 & 6 & 20 & 3461 & \tbcolor{}5 & \tbcolor{}0:00.02 & \tbcolor{}15 & \tbcolor{}876 & 6 & 0:00.09 \\
(5,2)-10 & 21 & \tbcolor{}23 & 5859 & \tbcolor{}5 & \tbcolor{}0:00.05 & 24 & \tbcolor{}1590 & 6 & 0:00.28 \\
    \midrule
(6,2)-1 & 25 & \tbcolor{}27 & 19411 & \tbcolor{}7 & \tbcolor{}0:00.30 & 33 & \tbcolor{}2808 & 8 & 0:00.56 \\
(6,2)-2 & 13 & 20 & 8517 & \tbcolor{}6 & \tbcolor{}0:00.04 & \tbcolor{}14 & \tbcolor{}876 & 7 & 0:00.15 \\
(6,2)-3 & 11 & \tbcolor{}22 & 12015 & \tbcolor{}6 & \tbcolor{}0:00.08 & 34 & \tbcolor{}3584 & 8 & 0:00.61 \\
(6,2)-4 & 13 & \tbcolor{}23 & 9258 & \tbcolor{}6 & \tbcolor{}0:00.06 & 24 & \tbcolor{}1692 & 7 & 0:00.26 \\
(6,2)-5 & 28 & \tbcolor{}28 & 18853 & \tbcolor{}7 & 0:00.97 & 32 & \tbcolor{}2672 & 8 & \tbcolor{}0:00.64 \\
(6,2)-6 & 13 & 19 & 6046 & \tbcolor{}6 & \tbcolor{}0:00.02 & \tbcolor{}17 & \tbcolor{}946 & 7 & 0:00.14 \\
(6,2)-7 & 32 & 22 & 5935 & \tbcolor{}6 & \tbcolor{}0:00.12 & \tbcolor{}18 & \tbcolor{}1068 & 7 & 0:00.16 \\
(6,2)-8 & 30 & 24 & 10159 & \tbcolor{}6 & \tbcolor{}0:00.09 & \tbcolor{}17 & \tbcolor{}1078 & 7 & 0:00.16 \\
(6,2)-9 & 11 & 21 & 9300 & \tbcolor{}6 & \tbcolor{}0:00.04 & \tbcolor{}18 & \tbcolor{}1131 & 7 & 0:00.18 \\
(6,2)-10 & 6 & \tbcolor{}17 & 4907 & \tbcolor{}5 & \tbcolor{}0:00.03 & \tbcolor{}17 & \tbcolor{}1520 & 6 & 0:00.22 \\
\midrule
(8,2)-1 & 26 & \tbcolor{}24 & 15739 & \tbcolor{}8 & \tbcolor{}0:00.07 & 26 & \tbcolor{}2178 & 9 & 0:00.40 \\
(8,2)-2 & 60 & 26 & 17560 & \tbcolor{}8 & 0:00.43 & \tbcolor{}19 & \tbcolor{}1764 & 9 & \tbcolor{}0:00.31 \\
(8,2)-3 & 40 & 34 & 24742 & \tbcolor{}8 & \tbcolor{}0:00.26 & \tbcolor{}25 & \tbcolor{}2580 & 9 & 0:00.38 \\
(8,2)-4 & 34 & 27 & 17519 & \tbcolor{}8 & \tbcolor{}0:00.14 & \tbcolor{}24 & \tbcolor{}2538 & 9 & 0:00.28 \\
(8,2)-5 & 97 & 30 & 19216 & \tbcolor{}8 & 0:00.50 & \tbcolor{}27 & \tbcolor{}2920 & 9 & \tbcolor{}0:00.45 \\
(8,2)-6 & 16 & 35 & 23660 & \tbcolor{}8 & 0:00.24 & \tbcolor{}21 & \tbcolor{}1729 & 9 & \tbcolor{}0:00.20 \\
(8,2)-7 & 52 & \tbcolor{}36 & 56084 & \tbcolor{}8 & 0:01.20 & 38 & \tbcolor{}5626 & 9 & \tbcolor{}0:00.73 \\
(8,2)-8 & 64 & 43 & 39172 & \tbcolor{}8 & 0:03.68 & \tbcolor{}37 & \tbcolor{}5616 & 10 & \tbcolor{}0:00.95 \\
(8,2)-9 & 19 & \tbcolor{}22 & 11588 & \tbcolor{}8 & \tbcolor{}0:00.05 & 24 & \tbcolor{}2331 & 9 & 0:00.27 \\
(8,2)-10 & 73 & 36 & 55303 & \tbcolor{}8 & 0:01.08 & \tbcolor{}35 & \tbcolor{}3654 & 9 & \tbcolor{}0:00.54 \\
\midrule
(8,4)-1 & 35 & 41 & 42866 & \tbcolor{}8 & \tbcolor{}0:00.22 & \tbcolor{}36 & \tbcolor{}3969 & 9 & 0:00.85 \\
(8,4)-2 & 37 & \tbcolor{}42 & 47398 & \tbcolor{}8 & \tbcolor{}0:00.39 & 45 & \tbcolor{}5562 & 9 & 0:01.22 \\
(8,4)-3 & 54 & 55 & 50070 & \tbcolor{}8 & \tbcolor{}0:00.81 & \tbcolor{}40 & \tbcolor{}13056 & 9 & 0:03.03 \\
(8,4)-4 & 58 & 51 & 42353 & \tbcolor{}8 & \tbcolor{}0:00.49 & \tbcolor{}48 & \tbcolor{}13175 & 9 & 0:03.39 \\
(8,4)-5 & 52 & \tbcolor{}49 & 42196 & \tbcolor{}8 & \tbcolor{}0:00.33 & \tbcolor{}49 & \tbcolor{}7128 & 9 & 0:01.62 \\
(8,4)-6 & 92 & 56 & 50723 & \tbcolor{}8 & \tbcolor{}0:00.50 & \tbcolor{}53 & \tbcolor{}7072 & 9 & 0:01.56 \\
(8,4)-7 & 37 & 46 & 39821 & \tbcolor{}8 & \tbcolor{}0:00.49 & \tbcolor{}40 & \tbcolor{}5481 & 9 & 0:01.56 \\
(8,4)-8 & 49 & 48 & 85662 & \tbcolor{}8 & \tbcolor{}0:01.46 & \tbcolor{}35 & \tbcolor{}9175 & 9 & 0:01.85 \\
(8,4)-9 & 90 & \tbcolor{}50 & 71963 & \tbcolor{}8 & \tbcolor{}0:01.13 & 74 & \tbcolor{}11877 & 10 & 0:03.23 \\
(8,4)-10 & 42 & 46 & 80570 & \tbcolor{}8 & \tbcolor{}0:00.82 & \tbcolor{}44 & \tbcolor{}8000 & 9 & 0:02.22 \\
\bottomrule
\end{tabular}
\end{table}

\begin{table}[tbp]
  \centering
  \footnotesize
  \caption{The full results of the experiments (part 3).}
  \label{table:full_results:3}
  \begin{tabular}{lrrrrrrrrr}
    \toprule
    \multirow{3}{*}{$(|\Loc|,|\Alphabet|)-i$} & \multirow{3}{*}{\begin{tabular}{c}DRTA \\ size\end{tabular}} & \multicolumn{4}{c}{\ALstarRTA{}} & \multicolumn{4}{c}{\NLstarRTA{}} \\
    \cmidrule(lr){3-6}\cmidrule(lr){7-10}
    & & \# EqQ & \# MemQ & $|\Loc_{\hypothesisA}|$ & Total Time & \# EqQ & \# MemQ & $|\Loc_{\hypothesisA}|$ & Total Time \\
    \midrule
(10,2)-1 & 36 & \tbcolor{}30 & 32129 & \tbcolor{}10 & \tbcolor{}0:00.30 & 36 & \tbcolor{}4176 & 11 & 0:00.59 \\
(10,2)-2 & 63 & 39 & 43102 & \tbcolor{}10 & 0:01.23 & \tbcolor{}34 & \tbcolor{}5819 & 11 & \tbcolor{}0:01.06 \\
(10,2)-3 & 30 & 42 & 48027 & \tbcolor{}10 & \tbcolor{}0:00.59 & \tbcolor{}39 & \tbcolor{}5740 & 11 & 0:00.94 \\
(10,2)-4 & 40 & 39 & 65153 & \tbcolor{}10 & \tbcolor{}0:00.49 & \tbcolor{}38 & \tbcolor{}5330 & 11 & 0:00.87 \\
(10,2)-5 & 42 & 39 & 73581 & \tbcolor{}10 & 0:01.23 & \tbcolor{}33 & \tbcolor{}3630 & 11 & \tbcolor{}0:00.62 \\
(10,2)-6 & 30 & 32 & 46930 & \tbcolor{}9 & \tbcolor{}0:00.48 & \tbcolor{}31 & \tbcolor{}5500 & 11 & 0:00.77 \\
(10,2)-7 & 144 & \tbcolor{}40 & 66313 & \tbcolor{}10 & 0:03.48 & \tbcolor{}40 & \tbcolor{}5772 & 11 & \tbcolor{}0:01.08 \\
(10,2)-8 & 133 & 43 & 41985 & \tbcolor{}11 & 0:04.16 & \tbcolor{}37 & \tbcolor{}7506 & 12 & \tbcolor{}0:01.53 \\
(10,2)-9 & 103 & \tbcolor{}36 & 53640 & \tbcolor{}10 & 0:03.80 & 39 & \tbcolor{}9660 & 11 & \tbcolor{}0:02.16 \\
(10,2)-10 & 44 & 42 & 34729 & \tbcolor{}10 & \tbcolor{}0:00.33 & \tbcolor{}37 & \tbcolor{}3520 & 11 & 0:00.60 \\
\midrule
(10,4)-1 & 32 & \tbcolor{}46 & 61658 & \tbcolor{}10 & \tbcolor{}0:00.71 & 48 & \tbcolor{}13545 & 11 & 0:02.82 \\
(10,4)-2 & 22 & 37 & 37080 & \tbcolor{}10 & \tbcolor{}0:00.28 & \tbcolor{}33 & \tbcolor{}5160 & 11 & 0:00.82 \\
(10,4)-3 & 45 & 58 & 135669 & \tbcolor{}10 & \tbcolor{}0:00.99 & \tbcolor{}53 & \tbcolor{}8510 & 11 & 0:01.62 \\
(10,4)-4 & 82 & \tbcolor{}55 & 129186 & \tbcolor{}11 & \tbcolor{}0:02.42 & 59 & \tbcolor{}10570 & 12 & 0:03.42 \\
(10,4)-5 & 107 & 62 & 258537 & \tbcolor{}10 & 0:09.91 & \tbcolor{}46 & \tbcolor{}9990 & 11 & \tbcolor{}0:02.32 \\
(10,4)-6 & 24 & 43 & 52379 & \tbcolor{}10 & \tbcolor{}0:00.31 & \tbcolor{}37 & \tbcolor{}4975 & 11 & 0:01.13 \\
(10,4)-7 & 150 & 62 & 141180 & \tbcolor{}10 & 0:04.55 & \tbcolor{}56 & \tbcolor{}12750 & 11 & \tbcolor{}0:04.21 \\
(10,4)-8 & 63 & \tbcolor{}58 & 94305 & \tbcolor{}9 & \tbcolor{}0:01.73 & 62 & \tbcolor{}15192 & 11 & 0:04.12 \\
(10,4)-9 & 71 & 59 & 100102 & \tbcolor{}10 & \tbcolor{}0:01.46 & \tbcolor{}50 & \tbcolor{}9880 & 11 & 0:01.94 \\
(10,4)-10 & 31 & \tbcolor{}44 & 59776 & \tbcolor{}10 & \tbcolor{}0:00.46 & 48 & \tbcolor{}6072 & 11 & 0:01.76 \\
\midrule
(10,6)-1 & 66 & 71 & 158398 & \tbcolor{}10 & \tbcolor{}0:02.54 & \tbcolor{}62 & \tbcolor{}22304 & 11 & 0:06.55 \\
(10,6)-2 & 57 & 63 & 125573 & \tbcolor{}10 & \tbcolor{}0:01.76 & \tbcolor{}56 & \tbcolor{}26910 & 11 & 0:07.84 \\
(10,6)-3 & 140 & 75 & 174459 & \tbcolor{}10 & \tbcolor{}0:02.67 & \tbcolor{}62 & \tbcolor{}13062 & 11 & 0:03.15 \\
(10,6)-4 & 31 & 67 & 178443 & \tbcolor{}10 & \tbcolor{}0:01.71 & \tbcolor{}58 & \tbcolor{}8903 & 11 & 0:03.06 \\
(10,6)-5 & 124 & 77 & 261261 & \tbcolor{}10 & 0:07.12 & \tbcolor{}61 & \tbcolor{}15300 & 11 & \tbcolor{}0:06.50 \\
(10,6)-6 & 56 & 59 & 98942 & \tbcolor{}10 & \tbcolor{}0:01.35 & \tbcolor{}52 & \tbcolor{}8876 & 11 & 0:02.28 \\
(10,6)-7 & 80 & 78 & 171743 & \tbcolor{}10 & \tbcolor{}0:02.53 & \tbcolor{}70 & \tbcolor{}14112 & 11 & 0:03.52 \\
(10,6)-8 & 111 & 82 & 177902 & \tbcolor{}10 & \tbcolor{}0:03.89 & \tbcolor{}76 & \tbcolor{}21560 & 11 & 0:10.47 \\
(10,6)-9 & 57 & 78 & 191526 & \tbcolor{}10 & \tbcolor{}0:02.78 & \tbcolor{}70 & \tbcolor{}19418 & 11 & 0:07.23 \\
(10,6)-10 & 68 & 67 & 177132 & \tbcolor{}10 & \tbcolor{}0:02.88 & \tbcolor{}61 & \tbcolor{}23018 & 11 & 0:11.02 \\
\midrule
(10,8)-1 & 98 & 99 & 772855 & \tbcolor{}10 & 1:05.00 & \tbcolor{}80 & \tbcolor{}14756 & 11 & \tbcolor{}0:07.23 \\
(10,8)-2 & 67 & 93 & 340084 & \tbcolor{}10 & \tbcolor{}0:05.97 & \tbcolor{}73 & \tbcolor{}21120 & 11 & 0:06.67 \\
(10,8)-3 & 173 & 103 & 678927 & \tbcolor{}10 & 0:32.54 & \tbcolor{}97 & \tbcolor{}19188 & 11 & \tbcolor{}0:08.45 \\
(10,8)-4 & 79 & 109 & 675535 & \tbcolor{}10 & 0:19.93 & \tbcolor{}78 & \tbcolor{}17112 & 11 & \tbcolor{}0:05.32 \\
(10,8)-5 & 100 & 95 & 442030 & \tbcolor{}10 & 0:28.30 & \tbcolor{}83 & \tbcolor{}21217 & 11 & \tbcolor{}0:06.55 \\
(10,8)-6 & 94 & 106 & 335891 & \tbcolor{}10 & \tbcolor{}0:07.54 & \tbcolor{}95 & \tbcolor{}21801 & 11 & 0:09.82 \\
(10,8)-7 & 114 & \tbcolor{}91 & 226714 & \tbcolor{}10 & \tbcolor{}0:05.93 & 93 & \tbcolor{}29886 & 11 & 0:10.38 \\
(10,8)-8 & 160 & 99 & 564469 & \tbcolor{}10 & 0:35.65 & \tbcolor{}83 & \tbcolor{}22725 & 11 & \tbcolor{}0:08.25 \\
(10,8)-9 & 98 & 97 & 243198 & \tbcolor{}10 & \tbcolor{}0:03.67 & \tbcolor{}84 & \tbcolor{}18172 & 11 & 0:05.86 \\
(10,8)-10 & 34 & 78 & 158310 & \tbcolor{}10 & \tbcolor{}0:02.95 & \tbcolor{}65 & \tbcolor{}17010 & 11 & 0:06.99 \\
\bottomrule
\end{tabular}
\end{table}

\begin{table}[tbp]
  \centering
  \footnotesize
  \caption{The full results of the experiments (part 4).}
  \label{table:full_results:4}
  \begin{tabular}{lrrrrrrrrr}
    \toprule
    \multirow{3}{*}{$(|\Loc|,|\Alphabet|)-i$} & \multirow{3}{*}{\begin{tabular}{c}DRTA \\ size\end{tabular}} & \multicolumn{4}{c}{\ALstarRTA{}} & \multicolumn{4}{c}{\NLstarRTA{}} \\
    \cmidrule(lr){3-6}\cmidrule(lr){7-10}
    & & \# EqQ & \# MemQ & $|\Loc_{\hypothesisA}|$ & Total Time & \# EqQ & \# MemQ & $|\Loc_{\hypothesisA}|$ & Total Time \\
    \midrule
(10,10)-1 & 69 & 104 & 653576 & \tbcolor{}10 & 0:20.78 & \tbcolor{}81 & \tbcolor{}17568 & 11 & \tbcolor{}0:08.09 \\
(10,10)-2 & 134 & 105 & 248412 & \tbcolor{}10 & 0:07.60 & \tbcolor{}83 & \tbcolor{}20516 & 11 & \tbcolor{}0:07.42 \\
(10,10)-3 & 89 & 137 & 723723 & \tbcolor{}10 & 0:26.16 & \tbcolor{}117 & \tbcolor{}29355 & 11 & \tbcolor{}0:11.28 \\
(10,10)-4 & 82 & 128 & 875092 & \tbcolor{}10 & 0:53.12 & \tbcolor{}97 & \tbcolor{}20874 & 11 & \tbcolor{}0:10.06 \\
(10,10)-5 & 129 & 101 & 277874 & \tbcolor{}10 & \tbcolor{}0:08.70 & \tbcolor{}82 & \tbcolor{}56840 & 11 & 0:31.43 \\
(10,10)-6 & 62 & 83 & 254169 & \tbcolor{}10 & 0:04.95 & \tbcolor{}67 & \tbcolor{}12078 & 11 & \tbcolor{}0:03.81 \\
(10,10)-7 & 99 & 126 & 1426815 & \tbcolor{}10 & 1:33.94 & \tbcolor{}94 & \tbcolor{}33369 & 11 & \tbcolor{}0:17.00 \\
(10,10)-8 & 131 & 138 & 1139758 & \tbcolor{}10 & 0:53.15 & \tbcolor{}118 & \tbcolor{}33264 & 11 & \tbcolor{}0:26.41 \\
(10,10)-9 & 205 & 102 & 737974 & \tbcolor{}10 & 0:22.71 & \tbcolor{}92 & \tbcolor{}22260 & 11 & \tbcolor{}0:10.52 \\
(10,10)-10 & 92 & 118 & 892851 & \tbcolor{}10 & 0:42.62 & \tbcolor{}89 & \tbcolor{}19461 & 11 & \tbcolor{}0:08.74 \\
\midrule
(12,2)-1 & 125 & 61 & 178443 & \tbcolor{}12 & 0:03.00 & \tbcolor{}42 & \tbcolor{}9331 & 13 & \tbcolor{}0:02.07 \\
(12,2)-2 & 81 & 48 & 60450 & \tbcolor{}12 & \tbcolor{}0:01.10 & \tbcolor{}46 & \tbcolor{}6699 & 13 & 0:01.31 \\
(12,2)-3 & 327 & 52 & 133136 & \tbcolor{}12 & 0:15.19 & \tbcolor{}42 & \tbcolor{}7776 & 13 & \tbcolor{}0:01.49 \\
(12,2)-4 & 162 & 59 & 108006 & \tbcolor{}13 & 0:05.31 & \tbcolor{}48 & \tbcolor{}11808 & 14 & \tbcolor{}0:02.55 \\
(12,2)-5 & 105 & 55 & 131301 & \tbcolor{}12 & 0:03.26 & \tbcolor{}46 & \tbcolor{}10208 & 13 & \tbcolor{}0:02.93 \\
(12,2)-6 & 121 & 48 & 140763 & \tbcolor{}12 & 0:03.77 & \tbcolor{}41 & \tbcolor{}7956 & 13 & \tbcolor{}0:01.27 \\
(12,2)-7 & 92 & 56 & 112251 & \tbcolor{}12 & 0:02.37 & \tbcolor{}51 & \tbcolor{}7910 & 13 & \tbcolor{}0:02.10 \\
(12,2)-8 & 263 & 83 & 289990 & \tbcolor{}14 & 0:35.43 & \tbcolor{}65 & \tbcolor{}9920 & 15 & \tbcolor{}0:02.41 \\
(12,2)-9 & 65 & 46 & 55809 & \tbcolor{}12 & \tbcolor{}0:00.53 & \tbcolor{}44 & \tbcolor{}10240 & 13 & 0:02.43 \\
(12,2)-10 & 176 & \tbcolor{}52 & 238212 & \tbcolor{}12 & 0:10.60 & 59 & \tbcolor{}10608 & 13 & \tbcolor{}0:02.76 \\
\midrule
(12,4)-1 & 17 & 52 & 113416 & \tbcolor{}12 & \tbcolor{}0:00.62 & \tbcolor{}50 & \tbcolor{}8652 & 13 & 0:01.94 \\
(12,4)-2 & 146 & 84 & 386488 & \tbcolor{}12 & 0:09.11 & \tbcolor{}73 & \tbcolor{}18122 & 13 & \tbcolor{}0:07.10 \\
(12,4)-3 & 125 & 62 & 302007 & \tbcolor{}12 & 0:06.79 & \tbcolor{}59 & \tbcolor{}9824 & 13 & \tbcolor{}0:02.52 \\
(12,4)-4 & 174 & \tbcolor{}61 & 211788 & \tbcolor{}12 & \tbcolor{}0:04.46 & 71 & \tbcolor{}18840 & 13 & 0:05.10 \\
(12,4)-5 & 105 & 72 & 161991 & \tbcolor{}12 & \tbcolor{}0:01.67 & \tbcolor{}68 & \tbcolor{}19575 & 13 & 0:05.42 \\
(12,4)-6 & 135 & 84 & 583090 & \tbcolor{}12 & 0:19.07 & \tbcolor{}72 & \tbcolor{}15826 & 13 & \tbcolor{}0:04.26 \\
(12,4)-7 & 85 & 69 & 171065 & \tbcolor{}12 & \tbcolor{}0:01.52 & \tbcolor{}62 & \tbcolor{}13120 & 13 & 0:03.06 \\
(12,4)-8 & 96 & \tbcolor{}76 & 307694 & \tbcolor{}12 & 0:05.66 & 80 & \tbcolor{}15414 & 13 & \tbcolor{}0:05.61 \\
(12,4)-9 & 172 & \tbcolor{}65 & 245963 & \tbcolor{}12 & 0:02.89 & \tbcolor{}65 & \tbcolor{}9214 & 13 & \tbcolor{}0:02.46 \\
(12,4)-10 & 216 & 75 & 471623 & \tbcolor{}12 & 0:14.10 & \tbcolor{}60 & \tbcolor{}21096 & 13 & \tbcolor{}0:06.40 \\
\midrule
(14,4)-1 & 65 & 66 & 208091 & \tbcolor{}14 & \tbcolor{}0:01.83 & \tbcolor{}58 & \tbcolor{}12680 & 15 & 0:03.42 \\
(14,4)-2 & 60 & \tbcolor{}56 & 192180 & \tbcolor{}15 & \tbcolor{}0:01.02 & 71 & \tbcolor{}12506 & 16 & 0:04.01 \\
(14,4)-3 & 70 & 74 & 265824 & \tbcolor{}14 & \tbcolor{}0:03.42 & \tbcolor{}72 & \tbcolor{}21168 & 15 & 0:06.58 \\
(14,4)-4 & 143 & 97 & 986803 & \tbcolor{}14 & 0:28.54 & \tbcolor{}76 & \tbcolor{}31728 & 15 & \tbcolor{}0:14.06 \\
(14,4)-5 & 114 & 71 & 204594 & \tbcolor{}14 & \tbcolor{}0:01.46 & \tbcolor{}68 & \tbcolor{}10290 & 15 & 0:03.24 \\
(14,4)-6 & 74 & \tbcolor{}74 & 280501 & \tbcolor{}14 & \tbcolor{}0:03.64 & \tbcolor{}74 & \tbcolor{}32438 & 15 & 0:13.26 \\
(14,4)-7 & 161 & 75 & 333600 & \tbcolor{}14 & 0:08.02 & \tbcolor{}69 & \tbcolor{}17848 & 15 & \tbcolor{}0:05.71 \\
(14,4)-8 & 326 & 91 & 360942 & \tbcolor{}14 & \tbcolor{}0:07.92 & \tbcolor{}85 & \tbcolor{}20188 & 15 & 0:08.28 \\
(14,4)-9 & 81 & \tbcolor{}77 & 289537 & \tbcolor{}14 & \tbcolor{}0:04.02 & 80 & \tbcolor{}15834 & 15 & 0:04.24 \\
(14,4)-10 & 209 & 136 & 2780462 & \tbcolor{}14 & 5:42.07 & \tbcolor{}87 & \tbcolor{}32292 & 15 & \tbcolor{}0:14.08 \\
\bottomrule
\end{tabular}
\end{table}

\begin{table}[tbp]
  \centering
  \footnotesize
  \caption{The full results of the experiments (part 5).}
  \label{table:full_results:5}
  \begin{tabular}{lrrrrrrrrr}
    \toprule
    \multirow{3}{*}{$(|\Loc|,|\Alphabet|)-i$} & \multirow{3}{*}{\begin{tabular}{c}DRTA \\ size\end{tabular}} & \multicolumn{4}{c}{\ALstarRTA{}} & \multicolumn{4}{c}{\NLstarRTA{}} \\
    \cmidrule(lr){3-6}\cmidrule(lr){7-10}
    & & \# EqQ & \# MemQ & $|\Loc_{\hypothesisA}|$ & Total Time & \# EqQ & \# MemQ & $|\Loc_{\hypothesisA}|$ & Total Time \\
    \midrule
(16,4)-1 & 923 & 116 & 2761765 & \tbcolor{}16 & 10:32.02 & \tbcolor{}85 & \tbcolor{}43155 & 17 & \tbcolor{}0:17.72 \\
(16,4)-2 & 208 & \tbcolor{}88 & 453181 & \tbcolor{}16 & 0:09.22 & 89 & \tbcolor{}22949 & 17 & \tbcolor{}0:07.65 \\
(16,4)-3 & 91 & 93 & 478472 & \tbcolor{}16 & 0:23.27 & \tbcolor{}73 & \tbcolor{}26606 & 17 & \tbcolor{}0:07.46 \\
(16,4)-4 & 463 & \tbcolor{}88 & 674578 & \tbcolor{}16 & 0:48.42 & 103 & \tbcolor{}54015 & 17 & \tbcolor{}0:24.24 \\
(16,4)-5 & 260 & 95 & 1162588 & \tbcolor{}16 & 1:25.93 & \tbcolor{}86 & \tbcolor{}38709 & 17 & \tbcolor{}0:23.21 \\
(16,4)-6 & 245 & \tbcolor{}81 & 294252 & \tbcolor{}17 & 0:08.68 & 86 & \tbcolor{}22400 & 18 & \tbcolor{}0:08.59 \\
(16,4)-7 & 117 & 83 & 483237 & \tbcolor{}15 & 0:16.23 & \tbcolor{}78 & \tbcolor{}23744 & 17 & \tbcolor{}0:08.86 \\
(16,4)-8 & 81 & 78 & 317988 & \tbcolor{}16 & \tbcolor{}0:08.03 & \tbcolor{}74 & \tbcolor{}18258 & 17 & 0:09.22 \\
(16,4)-9 & 489 & 101 & 943769 & \tbcolor{}16 & \tbcolor{}2:08.91 & \tbcolor{}85 & \tbcolor{}132928 & 17 & 2:48.27 \\
(16,4)-10 & 382 & \tbcolor{}91 & 554421 & \tbcolor{}16 & 0:31.97 & 96 & \tbcolor{}27280 & 17 & \tbcolor{}0:09.13 \\
\midrule
(18,4)-1 & 87 & \tbcolor{}67 & 361888 & \tbcolor{}18 & \tbcolor{}0:05.13 & 73 & \tbcolor{}14256 & 19 & 0:05.51 \\
(18,4)-2 & 557 & 131 & 1464980 & \tbcolor{}21 & 9:29.91 & \tbcolor{}128 & \tbcolor{}61464 & 22 & \tbcolor{}0:34.99 \\
(18,4)-3 & 575 & 130 & 1891770 & \tbcolor{}18 & 4:56.15 & \tbcolor{}114 & \tbcolor{}56637 & 19 & \tbcolor{}0:37.49 \\
(18,4)-4 & 343 & 107 & 1749460 & \tbcolor{}18 & 1:49.53 & \tbcolor{}87 & \tbcolor{}37863 & 19 & \tbcolor{}0:15.18 \\
(18,4)-5 & 400 & \tbcolor{}96 & 590103 & \tbcolor{}18 & 0:32.48 & 110 & \tbcolor{}35091 & 19 & \tbcolor{}0:19.70 \\
(18,4)-6 & 255 & \tbcolor{}98 & 1286708 & \tbcolor{}18 & 0:49.59 & 99 & \tbcolor{}32085 & 19 & \tbcolor{}0:12.51 \\
(18,4)-7 & 228 & 86 & 572519 & \tbcolor{}18 & 0:21.90 & \tbcolor{}84 & \tbcolor{}19928 & 19 & \tbcolor{}0:07.95 \\
(18,4)-8 & 516 & 110 & 1514586 & \tbcolor{}18 & 2:01.59 & \tbcolor{}104 & \tbcolor{}35308 & 20 & \tbcolor{}0:21.01 \\
(18,4)-9 & 190 & \tbcolor{}80 & 684605 & \tbcolor{}18 & \tbcolor{}0:10.96 & 93 & \tbcolor{}37408 & 19 & 0:17.25 \\
(18,4)-10 & 156 & 102 & 1245921 & \tbcolor{}18 & 1:05.25 & \tbcolor{}83 & \tbcolor{}24000 & 19 & \tbcolor{}0:09.13 \\
\midrule
(20,4)-1 & 395 & \tbcolor{}102 & 1009877 & \tbcolor{}20 & 0:32.94 & 111 & \tbcolor{}38080 & 21 & \tbcolor{}0:22.18 \\
(20,4)-2 & 612 & 115 & 1429673 & \tbcolor{}20 & 1:55.26 & \tbcolor{}109 & \tbcolor{}27885 & 21 & \tbcolor{}0:15.16 \\
(20,4)-3 & 341 & \tbcolor{}108 & 960338 & \tbcolor{}20 & 0:41.23 & 117 & \tbcolor{}60324 & 21 & \tbcolor{}0:33.09 \\
(20,4)-4 & 237 & \tbcolor{}91 & 911588 & \tbcolor{}20 & 0:27.22 & 97 & \tbcolor{}41400 & 21 & \tbcolor{}0:18.76 \\
(20,4)-5 & 542 & 110 & 885578 & \tbcolor{}20 & \tbcolor{}0:45.93 & \tbcolor{}105 & \tbcolor{}135675 & 21 & 2:38.04 \\
(20,4)-6 & 410 & 128 & 934053 & \tbcolor{}20 & 1:09.11 & \tbcolor{}116 & \tbcolor{}46480 & 21 & \tbcolor{}0:25.00 \\
(20,4)-7 & 437 & 104 & 1464721 & \tbcolor{}20 & 1:21.30 & \tbcolor{}94 & \tbcolor{}37914 & 21 & \tbcolor{}0:14.07 \\
(20,4)-8 & 563 & 98 & 778467 & \tbcolor{}20 & 0:29.29 & \tbcolor{}97 & \tbcolor{}26255 & 21 & \tbcolor{}0:12.46 \\
(20,4)-9 & 455 & 89 & 846221 & \tbcolor{}20 & 0:28.21 & \tbcolor{}86 & \tbcolor{}29097 & 21 & \tbcolor{}0:12.33 \\
(20,4)-10 & 235 & 112 & 2003803 & \tbcolor{}18 & \tbcolor{}2:20.95 & \tbcolor{}107 & \tbcolor{}124959 & 21 & 2:36.37 \\
\bottomrule
\end{tabular}
\end{table}
}

\ifdefined\VersionWithComments%
\section{Additional private information for us, the authors}
\instructions{The final paper and the signed copyright form are due by June 24, 2026 (AoE). The page limit is strictly 16 pages for regular papers and 7 pages for short papers, both excluding references, as in the original submission. Appendices are not allowed}

\begin{lncs}
 \setcounter{tocdepth}{1}
\end{lncs}
\listoftodos{}
\fi

\end{document}